\documentclass[11pt]{amsart}
\usepackage{graphicx}
\usepackage{amscd}
\usepackage{amsmath}
\usepackage{amsthm}
\usepackage{amsfonts}
\usepackage{amssymb}
\usepackage{setspace}
\usepackage{enumerate} 
\usepackage{color}
\usepackage{url}
\usepackage{hyperref}
\usepackage{color}
\usepackage{comment}

\usepackage{geometry}
\theoremstyle{plain}
\newtheorem{theorem}{Theorem}[section]

\newtheorem{corollary}[theorem]{Corollary}

\newtheorem{lemma}[theorem]{Lemma}

\newtheorem{proposition}[theorem]{Proposition}
\newtheorem{definition}[theorem]{Definition}
\newtheorem{condition}[theorem]{Condition}
\newtheorem{assumption}[theorem]{Assumption}

\theoremstyle{remark}
\newtheorem{remark}[theorem]{Remark}
\newtheorem{example}[theorem]{Example}

\numberwithin{equation}{section}

\newcommand{\be}{\begin{equation}}
\newcommand{\ee}{\end{equation}}
\newcommand{\ba}{\begin{aligned}}
\newcommand{\ea}{\end{aligned}}

\newcommand{\ind}{\mathbf{1}}
\newcommand{\dbra}[1]{[\kern-0.15em[ #1 ]\kern-0.15em]}
\newcommand{\dbraco}[1]{[\kern-0.15em[ #1 [\kern-0.15em[}
\newcommand{\dbraoc}[1]{]\kern-0.15em] #1 ]\kern-0.15em]}
\newcommand{\dbraoo}[1]{]\kern-0.15em] #1 [\kern-0.15em[}
\newcommand{\rsto}{]\!\kern-1.8pt ]}
\newcommand{\lsto}{[\!\kern-1.7pt [}

\newcommand{\R}{\mathbb{R}}
\newcommand{\QQ}{\mathbb{Q}}
\newcommand{\PP}{\mathbb{P}}
\newcommand{\N}{\mathbb{N}}
\newcommand{\EE}{\mathbb{E}}
\newcommand{\FF}{\mathbb{F}}

\newcommand{\bF}{\mathbb{F}}
\newcommand{\cU}{\mathcal{U}}
\newcommand{\cA}{\mathcal{A}}
\newcommand{\cB}{\mathcal{B}}
\newcommand{\cD}{\mathcal{D}}
\newcommand{\cE}{\mathcal{E}}
\newcommand{\cF}{\mathcal{F}}

\newcommand{\cH}{\mathcal{H}}
\newcommand{\cI}{\mathcal{I}}

\newcommand{\cP}{\mathcal{P}}

\newcommand{\cS}{\mathcal{S}}

\newcommand{\cX}{\mathcal{X}}
\newcommand{\cXfin}{\cX^{\mathrm{fin}}}

\newcommand{\cZ}{\mathcal{Z}}

\newcommand{\ud}{{\rm d}}

\newcommand{\Lip}{{\rm Lip}}
\newcommand{\LG}{{\rm LG}}
\newcommand{\loc}{{\rm loc}}
\newcommand{\B}{{\rm B}}
\newcommand{\Id}{{\rm Id}}

\newcommand{\bbI}{\mathbf{1}}

\makeatletter
\newcommand{\subjclassname@JEL}{JEL Classification}
\makeatother

\begin{document}
\title[Real-world HJM semimartingale models for multiple term structures]{Real-world models for multiple term structures: \\ a unifying HJM semimartingale framework}

\author{Claudio Fontana}
\address[Claudio Fontana]{Department of Mathematics ``Tullio Levi-Civita'', University of Padova, via Trieste 63, Padova, Italy.}
\email{fontana@math.unipd.it}

\author{Eckhard Platen}
\address[Eckhard Platen]{School of Mathematical Sciences and Finance Discipline Group, University of Technology Sydney, Broadway NSW 2007, Sydney, Australia.}
\email{Eckhard.Platen@uts.edu.au}%

\author{Stefan Tappe}
\address[Stefan Tappe]{Department of Mathematical Stochastics, Albert Ludwig University of Freiburg, Ernst-Zermelo-Stra\ss e 1, D-79104 Freiburg, Germany.}
\email{stefan.tappe@math.uni-freiburg.de}%

\thanks{The authors are grateful to David Criens for fruitful discussions. 
Claudio Fontana gratefully acknowledges the support of the Bruti-Liberati visiting fellowship and the hospitality of the Quantitative Finance Research Centre at the Finance Discipline Group at the University of Technology Sydney, where this work was started, and also financial support from the Europlace Institute of Finance.
Stefan Tappe gratefully acknowledges financial support from the Deutsche Forschungsgemeinschaft (DFG, German Research Foundation) -- project number 444121509.}

\date{\today}

\begin{abstract}
We develop a unified framework for modeling multiple term structures arising in financial, insurance, and energy markets, adopting an extended Heath-Jarrow-Morton (HJM) approach under the real-world probability. We study market viability and characterize the set of local martingale deflators. 
We conduct an analysis of the associated stochastic partial differential equation (SPDE), addressing existence and uniqueness of solutions, invariance properties and existence of affine realizations.
\end{abstract}

\keywords{Heath-Jarrow-Morton framework; real-world probability; large financial market; local martingale deflator; stochastic partial differential equation; invariance; affine realization.}
\subjclass[2020]{60G44, 60H15, 91G15, 91G30. 
%\textit{JEL Classification} C60, E43, F31, G12
}
\maketitle
\tableofcontents

\section{Introduction}

In financial mathematics, a term structure is a family of stochastic processes representing the prices of contracts that deliver payoffs at different future dates (maturities). The canonical example is given by the term structure of interest rates, which encodes information about the value of money at different points in time. A term structure is an inherently complex mathematical object: first, it is infinite-dimensional, being a collection of prices indexed over a continuous maturity spectrum; second, it evolves randomly in time.
Continuous-time modeling of the term structure of interest rates started with the seminal work of Heath-Jarrow-Morton (HJM) \cite{HJM:92}, which also revealed a fundamental drift restriction that must be respected in order to ensure absence of arbitrage in the infinite-dimensional bond market. 
Such drift restriction makes the well-posedness of an HJM-type model highly non-trivial, requiring a careful analysis of the associated stochastic partial differential equation (see, e.g., \cite{F:01,fitate2010}).
% maybe we can mention that \cite{HJM:92} started a whole branch of financial mathematics and give some references for general semimartingale HJM models

In many contexts, multiple term structures coexist, as illustrated by the examples in Section \ref{sec:examples}. This introduces further elements of complexity into the analysis, since the stochastic evolution equations describing the individual term structures cannot be treated independently and, depending on the specific modeling context, are often required to respect certain ordering properties.
Moreover, absence of arbitrage must hold for all the term structures jointly considered, due to the possibility of simultaneously trading contracts referring to different term structures. 

In this paper, we aim at developing a general and unifying framework for multiple term structures, based on the HJM philosophy. 
We can outline as follows the main contributions of the paper:
\begin{enumerate}
\item[(i)] 
We introduce an abstract semimartingale parameterization for markets with multiple term structures and analyze it from the viewpoint of large financial markets with uncountably many assets, in the spirit of \cite{CKT16}. 
We work under the real-world probability and derive a version of the fundamental theorem of asset pricing which characterizes market viability (as introduced in \cite{KK07} for finite-dimensional semimartingale models) for large financial markets under an infinite time horizon.
Working under the real-world probability is practically relevant because viable market models may not admit a risk-neutral probability, while valuation and hedging problems can still be meaningfully posed and consistently solved, as shown for instance in \cite{FernholzKaratzas09, FR13, KKbook, PH}.
\item[(ii)] 
We develop a semimartingale HJM framework for multiple term structures under the real-world probability imposing only mild regularity conditions. We provide a complete characterization of the set of local martingale deflators, whose existence ensures market viability, and show that it can be explicitly characterized in terms of drift restrictions that generalize the classical HJM condition of \cite{HJM:92}. 
We also derive a simple condition which ensures ordered term structures.
As a special case of our general framework, we recover the risk-neutral setup and obtain risk-neutral HJM drift restrictions under conditions weaker than those found in the literature.
\item[(iii)] 
We prove a new existence and uniqueness theorem for a class of semilinear stochastic partial differential equations (SPDEs) with random locally Lipschitz coefficients, driven by a Brownian motion and a Poisson random measure. By relying on this general result, we establish the well-posedness of the HJM semimartingale dynamics by proving existence and uniqueness of solutions to the associated SPDE. 
Besides extending the results of \cite{fitate2010} to a multi-dimensional setting, we substantially weaken the technical requirements, thereby covering new classes of models that cannot be treated by the existing theory, also in the risk-neutral setup, as shown by means of an explicit example.
Finally, we analyze the invariance properties of the SPDE and provide conditions for the existence of finite-dimensional (affine) realizations.
\end{enumerate}
 
The paper is structured as follows. In Section \ref{sec:examples}, we present some examples showing how multiple term structures arise in different contexts. Section \ref{sec:NA} introduces an abstract parameterization of multiple term structures and analyzes the issue of market viability. In Section \ref{sec:HJM}, we develop a general HJM framework and characterize the set of local martingale deflators. Section \ref{sec:SPDE} contains the existence and uniqueness result for semilinear SPDEs and a study of the SPDEs arising in the HJM framework. 
The paper is completed by four appendices containing some proofs and technical results.

\subsection{Examples}\label{sec:examples}

In this section, we illustrate how multiple term structures arise naturally in different contexts in finance, insurance and energy markets. In the following examples, we assume the existence of a term structure of risk-free zero-coupon bond (ZCB) prices, denoted by $B^0(t,T)$, for $t\leq T$.

\subsubsection{Foreign exchange markets}\label{ex:FXmarkets}

A first example of multiple term structures arises from yield curves in international financial markets. Similarly to \cite{JarrowTurnbull98}, this example will inspire our general parameterization of multiple term structures.
We consider a domestic economy associated to a reference currency and a family $I=\{1,\ldots,m\}$ of foreign economies, each of them associated to a distinct foreign currency. For each $i\in I$, the value at time $t$ of one unit of the $i$-th foreign currency in units of the domestic currency is given by the spot exchange rate $S^i_t$, while the $i$-th yield curve is specified by foreign ZCB prices $B^i(t,T)$, representing the value at time $t$ in units of the $i$-th foreign currency of one unit of the same currency delivered at time $T\geq t$. 
From the perspective of a domestic investor, the value of the $i$-th foreign ZCB at time $t$ is $S^i_tB^i(t,T)$, for each $i\in I$. If international trading is allowed, foreign ZCBs constitute risky assets for domestic investors due to  currency risk, unlike domestic ZCBs. This framework yields $m+1$ term structures, corresponding to the domestic yield curve and the $m$ foreign yield curves, which must coexist in an arbitrage-free way.
HJM models for FX markets were first introduced in \cite{AminJarrow91} and later analyzed in a semimartingale setup in \cite{phdkoval}. 
Inflation-linked term structure models, as considered for instance in the seminal framework of \cite{JarrowYildirim03}, also naturally fit into this setting, since they involve the joint modeling of both nominal and real interest rate term structures.

\subsubsection{Interbank interest rates}\label{ex:multicurve}

Multiple term structures arise in interest rate markets when considering benchmark rates indexed by a family $I=\{\delta_1,\ldots,\delta_m\}$ of ordered tenors (see \cite{GR15} for an overview of multi-curve interest rate models). The most basic contract referencing a benchmark rate is a single-period swap (forward rate agreement), in which at maturity $T+\delta_i$ the benchmark rate $L_T(\delta_i)$ fixed at time $T$ for tenor $\delta_i$ is exchanged for a fixed rate $K$.
It has been shown in \cite{FGGS20} that the value of this contract at time $t\leq T$ equals $S^i_tB^i(t,T)-(1+\delta_i K)B^0(t,T+\delta_i)$, where $S^i_t:=(1+\delta L_t(\delta_i))B^0_t(t+\delta_i)$ is the multiplicative spread at time $t$ between the $\delta_i$-tenor rate and the simply compounded risk-free rate over $[t,t+\delta_i]$ and $B^i(t,T)$ is a term structure factor with $B^i(T,T)=1$, for all $i=1,\ldots,m$ and $T\geq0$.
The tenor dependence reflects the differing impact of credit, funding, and liquidity risks across tenors in the interbank market. Since longer tenors involve greater risks, $S^i_tB^i(t,T)$ is generally increasing in $i$.
The existence of swaps of different maturities referencing benchmark rates of different tenors, together with the risk-free term structure, leads to a market with $m+1$ term structures that must coexist in an arbitrage-free way. Multi-curve interest rate models have been studied in a semimartingale setup in \cite{CFG:16,CFGaffine,FGGS20}.
Similarly, multiple term structures also arise in the cross-currency HJM framework of \cite{GnoattoLavagnini26}, where each term structure is indexed by both a base and a collateral currency.

\subsubsection{Credit-risky bonds}\label{ex:credit}

Multiple term structures also emerge in financial markets where bonds issued by different entities are traded. Let $I$ denote a set of obligors with different credit quality and let $S^i_t$ represent the credit quality (default loss) indicator of obligor $i\in I$ at time $t$, with the convention that $S^i_t=1$ if obligor $i$ has perfect creditworthiness at time $t$, while $S^i_t=0$ if obligor $i$ has defaulted by time $t$. Let $B^i(t,T)$ denote the value at time $t$ of a ZCB maturing at $T$ issued by obligor $i$ under the assumption of perfect credit quality, so that $S^i_tB^i(t,T)$ represents its actual market value accounting for credit risk. This setting naturally gives rise to a multiplicity of term structures which must coexist in an arbitrage-free manner, since all obligors are issuing securities in the same financial market. 
The application of an HJM model to credit-risky term structures was first proposed in  \cite{JT:95}. An analogous situation arises when considering corporate bonds with different credit ratings (see \cite[Chapter 13]{BRbook}).

\subsubsection{Longevity bonds}\label{ex:longevity}
 
Multiple term structures also arise in insurance markets where longevity bonds are traded, i.e., zero-coupon bonds whose payoff depends on the value of a specified survivor index. We consider a collection of $m$ survivor indices, each corresponding to a different age group. For each $i=1,\ldots,m$, let $S^i_t$ denote the $i$-th survivor index at time $t$, defined as the proportion of individuals alive at time $t$ from the initial cohort belonging to the $i$-th age group (survival ratio). A longevity bond maturing at $T$ linked to the $i$-th survivor index has payoff $S^i_T$ at time $T$. The market value of this bond at time $t\leq T$ can be expressed as $S^i_tB^i(t,T)$, with $B^i(t,T)$ satisfying the terminal condition $B^i(T,T)=1$, for all $T\geq0$. 
Together with the risk-free term structure, this setup gives rise to $m+1$ term structures coexisting in the same market. Moreover, if the age groups are ordered increasingly, then $S^i_tB^i(t,T)$ should be decreasing in $i$, as longevity bonds linked to older cohorts should trade at lower prices.
Adopting this parameterization of longevity bonds, a HJM model for the term structure of longevity bonds was first developed in \cite{Barbarin08} (see also \cite{TappeWeber14} for a more general HJM formulation).

\subsubsection{Energy forward contracts}\label{ex:energy}

In energy markets, swap contracts for the delivery of energy over a specified future time interval $[T,T+\delta]$ are widely traded (this is for instance the case of flow forward contracts on electricity, gas and temperature), where $\delta$ is the delivery length. Let $I$ be an index set corresponding to all possible delivery lengths. For each $i\in I$, let $S^i_t$ denote the spot price of a swap contract with immediate delivery over the interval $[t,t+\delta_i]$, for all $t\geq0$. The market value at time $t$ of a swap contract delivering energy over $[T,T+\delta_i]$, with $T\geq t$, can then be expressed as $S^i_tB^i(t,T)$, with $B^i(t,T)$ representing a forward adjustment factor for maturity $T$ associated with delivery length $\delta_i$, satisfying $B^i(T,T)=1$.
This setup generates multiple term structures which must coexist in an arbitrage-free way, since swap contracts for different delivery lengths are traded in the same market. If the delivery lengths are increasingly ordered, then also $S^i_tB^i(t,T)$ should be increasing in $i$, thereby implying a monotonicity property among the term structures.
A HJM model for swap contracts with different delivery lengths has been proposed in \cite{BenthKoekebakker08} and then generalized in \cite{Benth-Kruehner}, \cite[Chapter 6]{BenthKruhnerbook} and \cite{Benth_et_al19}.

\section{Market viability with multiple term structures}	\label{sec:NA}

In this section, we address the issue of market viability for a general financial market comprising multiple term structures, abstracting from the specific examples discussed in Section \ref{sec:examples}. As mentioned in the introduction, we aim at a modeling framework that is based on the real-world probability and does not necessitate the existence of a risk-neutral measure. We shall therefore address  market viability in the sense of no unbounded profit with bounded risk (NUPBR, see Definition \ref{def:NUPBR} below).

\subsection{Abstract market setup}	\label{sec:market}

We start by describing a general setup for a financial market, covering all the examples discussed in Section \ref{sec:examples}. 
We work in an infinite time horizon on a probability space $(\Omega,\cF,\PP)$ endowed with a right-continuous filtration $\bF=(\cF_t)_{t\geq0}$ with respect to which all processes introduced below are assumed to be adapted. The measure $\PP$ stands for the real-world probability. 
We assume that prices are denominated in units of a fixed (but otherwise arbitrary) reference currency and assume the existence of a tradable {\em num\'eraire} with strictly positive price process $X^0$.
In this work, adopting a generic terminology inspired by the examples of Section \ref{sec:examples}, we assume the existence of
\begin{enumerate}
\item[(i)] a {\em riskless term structure}, represented by $B^0(t,T)$;
\item[(ii)] a family of {\em risky term structures}, indexed by elements $i$ of a set $I$ and represented by $B^i(t,T)$. 
In addition, for each $i\in I$, a {\em spot process} $S^i$ is associated to the $i$-th risky term structure.
\end{enumerate}  
The quantity $B^0(t,T)$ represents the value at time $t$ (in units of the reference currency) of one unit of currency delivered at time $T\geq t$. Therefore, $B^0(t,T)$ corresponds as usual to the price of a riskless zero-coupon bond with maturity $T$, satisfying the terminal condition $B^0(T,T)=1$, for all $T\in\R_+$.
For each $i\in I$, the quantity $B^i(t,T)$ represents the value at time $t$, expressed in units of the spot process $S^i_t$, of the future random payoff $S^i_T$ to be delivered at time $T\geq t$. Therefore, the value of this random payoff at time $t$, measured in units of the reference currency, is given by $S^i_tB^i(t,T)$. This definition implies that $B^i(T,T)=1$, for all $T\in\R_+$, and it is, therefore, natural to interpret $B^i(t,T)$ as a risky zero-coupon bond with maturity $T$ associated to $i$, since it delivers a payoff which is random from the perspective of the reference currency.
The index set $I$ is assumed to be a subset of $\R$ and may be either finite or not, depending on the specific market setting under consideration.
For notational convenience, we define $I_0:=I\cup\{0\}$, in order to identify the riskless term structure with index $i=0$.
It is easy to see that this abstract setup encompasses all the examples discussed in Section \ref{sec:examples}.

Motivated by the above description, we define as follows our abstract financial market.

\begin{definition}	\label{def:market}
The {\em financial market} consists of the following family of processes:
\[
\bigl\{X^0, \,B^0(\cdot,T), \,S^iB^i(\cdot,T); \text{ for all }i\in I\text{ and }T\in\R_+\bigr\}.
\] 
\end{definition}

We set $S^0\equiv 1$, in order to represent the family of $X^0$-discounted processes considered in Definition \ref{def:market} by the set $(X^0)^{-1}\{S^iB^i(\cdot,T):(i,T)\in I_0\times\R_+\}\cup\{1\}$. 
Without further mention, we shall work under the {\bf standing assumption} that every element of that set is a semimartingale on $(\Omega,\bF,\PP)$.

\subsection{Market viability}	\label{sec:FTAP}

In this section, we derive a version of the fundamental theorem of asset pricing that is applicable to an abstract financial market with multiple term structures as introduced above. We adopt a mild notion of market viability (Definition \ref{def:NUPBR}) corresponding to the no unbounded profit with bounded risk (NUPBR) condition of \cite{KK07}.\footnote{We refer to \cite{FGGS20} for a version of the fundamental theorem of asset pricing for multi-curve interest rate models based on the stronger notion of no asymptotic free lunch with vanishing risk introduced in \cite{CKT16}.}
The abstract financial market, as introduced in Definition \ref{def:market}, is a {\em large financial market} with uncountably many assets. In order to define and characterize market viability, we adopt the approach of \cite{CKT16} (see in particular Example 2.2 therein\footnote{We point out that the index set $I_0\times\R_+$ identifying the assets of the financial market of Definition \ref{def:market} can be mapped in a bijective way onto a subset of $\R_+$, therefore satisfying the requirements of a parameter space as considered in \cite{CKT16}.}), considering wealth processes that are limits in the semimartingale topology of admissible portfolios involving a finite number of arbitrarily chosen assets.
More precisely, for each $n\in\N$, we define 
\[
\cA^n := \{\text{all sets }A\subset I_0\times\R_+\text{ such that }|A|=n\},
\]
where $|A|$ denotes the cardinality of $A$, in such a way that $\cA^n$ represents the family of all subsets of the market containing $n$ assets.
For $A=\{(i_1,T_1),\ldots,(i_n,T_n)\}\in\cA^n$, for some $n\in\N$, we define the $n$-dimensional semimartingale $X^A=(X^{(i_1,T_1)}, \ldots,X^{(i_n,T_n)})$, where
\[
X^{(i_k,T_k)} := (X^0)^{-1}S^{i_k}B^{i_k}(\cdot,T_k),
\qquad\text{ for }k=1,\ldots,n,
\]
corresponding to the $X^0$-discounted price process (in the reference currency) of the bond with maturity $T_k$ associated with the $i_k$-th term structure.
We assume that trading in a subset $A$ of the market is done through self-financing $1$-admissible strategies. Denoting by $L(X^A)$ the set of $\R^{|A|}$-valued predictable $X^A$-integrable processes, this amounts to considering the following set of discounted wealth processes:
\[
\cX^A_1 := \bigl\{H\cdot X^A : H\in L(X^A)\text{ and }H\cdot X^A\geq-1\bigr\},
\]
where $H\cdot X^A$ denotes the stochastic integral of $H\in L(X^A)$ with respect to the semimartingale $X^A$.
Finally, we consider admissible portfolios involving arbitrary choices of any finite number of assets:
\[
\cX^n_1 := \bigcup_{A\in\cA^n}\cX^A_1
\qquad\text{ and }\qquad
\cX_1 := \overline{\bigcup_{n\geq1}\cX^n_1},
\]
where the bar denotes the closure in \'Emery's semimartingale topology.
As noted in \cite[Section 2.1]{CKT16}, this corresponds to considering $1$-admissible generalized strategies as initially introduced in \cite{DeDonnoPratelli06}.

\begin{definition}\label{def:NUPBR}
We say that the condition of {\em no unbounded profit with bounded risk (NUPBR)} holds if the set $\cX_1(T):=\{X_T : X\in\cX_1\}$ is bounded in probability, for every $T>0$.
\end{definition}

\begin{remark}
Definition \ref{def:NUPBR} extends \cite[Definition 4.1]{CKT16} to an infinite time horizon by requiring NUPBR to hold over every arbitrary finite time horizon. This extension is analogous to the one first adopted in \cite{Kardaras10} in the context of a finite-dimensional market. 
We remark that requiring NUPBR on $[0,T]$, for all $T>0$, is weaker than having NUPBR on $[0,+\infty)$ as considered in \cite{KK07} (see also \cite[Remark 5.3]{BalintSchweizer20}).
\end{remark}

In this setup, the dual elements are supermartingale deflators, as introduced in the next definition.

\begin{definition}\label{def:deflators}
A strictly positive c\`adl\`ag process $D$ with $D_0\leq 1$ is a {\em supermartingale deflator} if $D(1+X)$ is a supermartingale, for every $X\in\cX_1$. 
A strictly positive local martingale $D$ with $D_0=1$ is a {\em local martingale deflator (LMD)} if $DX^{(i,T)}$ is a local martingale, for every $(i,T)\in I_0\times\R_+$. 
\end{definition}

The next theorem is a version of the fundamental theorem of asset pricing based on the above notion of NUPBR. This result extends \cite[Theorem A.1]{CKT16} in two directions: first, by considering an infinite time horizon; second, by proving the existence of a supermartingale deflator which is the reciprocal of a wealth process  ({\em num\'eraire portfolio}, see \cite{KK07}).
Although the proof relies substantially on existing results such as \cite[Theorem 1.7]{Kardaras13b}, we are not aware of a formulation that directly applies to our setting. We therefore provide a complete proof in Appendix \ref{app:prop_FTAP}.
%While the extension is relatively straightforward, relying on existing results such as \cite[Theorem 1.7]{Kardaras13b}, we give a detailed proof in Appendix \ref{app:prop_FTAP}, as we could not find in the literature a formulation that is directly applicable to our setting.
We note that, in the present setup, the process $1+\widehat{X}$ appearing in Theorem \ref{thm:FTAP} corresponds to the wealth process of the {\em growth-optimal portfolio}, a central object in the benchmark approach to finance (see, e.g., \cite{PH}).

\begin{theorem}\label{thm:FTAP}
For the financial market of Definition \ref{def:market}, NUPBR holds if and only if there exists an element $\widehat{X}\in\cX_1$ with $\widehat{X}>-1$ such that $1/(1+\widehat{X})$ is a supermartingale deflator.
Moreover, every local martingale deflator is a supermartingale deflator.
\end{theorem}

We point out that, at the present level of generality, NUPBR does not imply the existence of LMDs. Indeed, \cite{CKT16} gives an explicit example showing that in large financial markets market viability (in the sense of NUPBR, but also in the stronger sense of no asymptotic free lunch with vanishing risk) does only ensure the existence of supermartingale deflators, and not necessarily of LMDs.\footnote{In the specific case of continuous-path processes NUPBR always implies the existence of an LMD, as shown in \cite{Kardaras24}.}
However, given a concrete model, it is usually difficult to characterize supermartingale deflators, since their definition is based on the set $\cX_1$ and does not give information on the basic assets included in Definition \ref{def:market}.
On the contrary, LMDs can be directly described in terms of the characteristics of the basic assets, as we are going to show in Section \ref{sec:HJM}.
Whenever it is nonempty, we denote by $\cD$ the set of all LMDs. 

\begin{remark}
Trading in markets with uncountably many assets can also be described through measure-valued strategies, as initially considered in the context of bond markets in \cite{BDMKR97,bjkaru97}. 
It can be shown that $\cD\neq\emptyset$ suffices to ensure NUPBR also with respect to measure-valued strategies. This is also related to the results of \cite{DeDonnoPratelli05}, who showed that stochastic integrals of $1$-admissible measure-valued strategies according to \cite{BDMKR97} are elements of $\cX_1$. In view of this observation, we have adopted the more general approach of \cite{CKT16}, where the set of $1$-admissible wealth processes is directly defined as $\cX_1$. 
\end{remark}

\section{A real-world HJM semimartingale framework}	\label{sec:HJM}

In this section, we develop and study a general modeling framework for multiple term structures based on the Heath-Jarrow-Morton approach under the real-world probability. The framework is described in Section \ref{sec:HJM_setup}, while Section \ref{sec:thm_HJM} contains the main result of this section, providing a complete description of the family of LMDs. In Section \ref{sec:monotone} we provide conditions ensuring the monotonicity of the risky term structures, while in Section \ref{sec:RN} we specialize our results to the risk-neutral setting.

\subsection{Probabilistic setup}	\label{sec:HJM_setup}

Let the filtered probability space $(\Omega,\cF,\FF,\PP)$ support a $d$-dimensional Brownian motion $W$ and an integer-valued random measure $\mu(\ud t,\ud x)$ on $\R_+\times E$ with compensator $\nu(\ud t,\ud x)=F_t(\ud x)\ud t$, where $E$ is a Polish space with its Borel sigma-field $\cB(E)$ and $F_t(\ud x)$ is a kernel from $(\Omega\times\R_+,\mathcal{P})$ into $(E,\cB(E))$, with $\mathcal{P}$ denoting the predictable sigma-field on $\Omega\times\R_+$. The compensated random measure is denoted by $\tilde{\mu}(\ud t,\ud x):=\mu(\ud t,\ud x)-\nu(\ud t,\ud x)$.
We denote by $L^2_{\rm loc}(W)$ the set of all progressively measurable $\R^d$-valued processes $\theta=(\theta_t)_{t\geq0}$ such that $\int_0^T\|\theta_t\|^2\ud t<+\infty$ a.s., for all $T>0$, and by $G_{\rm loc}(\mu)$ the set of all $\cP\otimes\cB(E)$-measurable functions $\varphi:\Omega\times\R_+\times E\to\R$ such that $\int_0^T\int_E((\varphi^i_t(x))^2\wedge|\varphi^i_t(x)|)F_t(\ud x)\ud t<+\infty$ a.s., for all $T>0$ (compare with \cite[Theorem II.1.33-c]{jashi03}).
We refer the reader to \cite{jashi03} for all unexplained notions of stochastic calculus.

The spot processes introduced in Section \ref{sec:market} are assumed to be non-negative semimartingales of the form $S^i=S^i_0\cE(Z^i)$, for all $i\in I$, where $\cE(Z^i)$ is the stochastic exponential of the special semimartingale 
\be	\label{eq:log_spread}
Z^i = \int_0^{\cdot}a^i_s\ud s + \int_0^{\cdot}b^i_s\ud W_s + \int_0^{\cdot}\int_Ec^i_s(x)\tilde{\mu}(\ud s,\ud x),
\ee
where $a^i$ is a real-valued adapted process satisfying $\int_0^T|a^i_t|\ud t<+\infty$ a.s. for all $T>0$, $b^i\in L^2_{\rm loc}(W)$ and $c^i\in G_{\rm loc}(\mu)$ with $c^i\geq-1$. These conditions are the minimal requirements for the well-posedness of \eqref{eq:log_spread}.
Note that if the set $\{(\omega,t)\in\Omega\times\R_+:\int_Ec^i_t(\omega,x)\mu(\omega;\{t\}\times \ud x)=-1\}$ is not evanescent, then $S^i$ can become null with positive probability. 
As discussed in Section \ref{sec:examples}, vanishing spot processes must be allowed if one wants to embed defaultable term structures into this general framework.

The num\'eraire $X^0$ introduced in Section \ref{sec:market} is assumed to be generated by a locally riskless interest rate $r=(r_t)_{t\geq0}$, which is a real-valued adapted process satisfying $\int_0^T|r_t|\ud t<+\infty$ a.s., for all $T>0$. The num\'eraire process $X^0$ is therefore given by $X^0=\exp(\int_0^{\cdot}r_t\ud t)$.

As explained in Section \ref{sec:market}, riskless and risky term structures can be represented by zero-coupon bond prices $B^i(t,T)$,  for $i\in I_0$ and $0\leq t\leq T<+\infty$, assumed to have the following structure: 
\[
B^i(t,T) = \exp\left(-\int_t^Tf^i(t,u)\ud u\right),
\]
where, for all $i\in I_0$ and $T>0$, the forward rate process $f^i(\cdot,T)=(f^i(t,T))_{t\in[0,T]}$ is given by
\be	\label{eq:fwd_rate}
f^i(t,T) = f^i(0,T) + \int_0^t\alpha^i(s,T)\ud s + \int_0^t\beta^i(s,T)\ud W_s + \int_0^t\int_E\gamma^i(s,T,x)\tilde{\mu}(\ud s,\ud x),
\ee
with $\alpha^i$, $\beta^i$ and $\gamma^i$ satisfying the mild technical requirements stated in the following assumption.

\begin{assumption}	\label{ass}
The following conditions hold a.s. for every $i\in I_0$:
\begin{enumerate}[(i)]
\item 
The {\em initial forward curve} $T\mapsto f^i(0, T)$ is $\cF_0\otimes\cB(\R_+)$-measurable, real-valued and satisfies $\int_0^T|f^i(0,t)|\ud t < + \infty$, for all $T>0$.
\item
The {\em drift process} $\alpha^i: \Omega \times\R_+^2\rightarrow\R$ is such that its restriction $\alpha^i|_{[0,t]}:\Omega\times[0,t]\times\R_+\rightarrow\R$ is $\cF_t\otimes \mathcal{B}([0,t])\otimes \mathcal{B}(\R_+)$-measurable, for every $t\in\R_+$. 
Moreover, $\alpha^i(t,T) = 0$ for all $t>T$ and
\[
\int_0^T\int_0^u |\alpha^i(s,u)|\ud s\ud u < + \infty,
\qquad\text{ for all }T>0.
\]
\item 
The {\em volatility process} $\beta^i:\Omega\times\R_+^2\rightarrow\R^d$ is such that its restriction $\beta^i|_{[0,t]}:\Omega\times[0,t]\times\R_+\rightarrow\R^d$ is $\cF_t\otimes \mathcal{B}([0,t])\otimes \mathcal{B}(\R_+)$-measurable, for every $t\in\R_+$. 
Moreover, $\beta^i(t,T) = 0$ for all $t>T$ and
\[
\sum_{j=1}^d\int_0^T\left(\int_0^u(\beta^{i,j}(s,u))^2\ud s\right)^{1/2}\ud u < + \infty,
\qquad\text{ for all }T>0.
\]
\item 
The {\em jump function} $\gamma^i:\Omega\times\R^2_+\times E\rightarrow\R$ is a $\mathcal{P}\otimes\mathcal{B}(\R_+)\otimes \mathcal{B}(E)$-measurable function. Moreover, $\gamma^i(t,T,x) = 0$ for all $t>T$ and $x \in E$, and
\[
\int_0^T\int_E  \int_0^T (\gamma^i(s,u,x))^2 \ud u\, F_s(\ud x)\ud s < +\infty,
\qquad\text{ for all }T>0.
\]
\end{enumerate}
\end{assumption}

Assumption \ref{ass} implies that the integrals appearing in the forward rate equation \eqref{eq:fwd_rate} are well-defined for a.e. $T$. In addition, the integrability requirements appearing in parts (ii)-(iv) of Assumption \ref{ass} ensure the applicability of ordinary and stochastic Fubini theorems, in the  versions of \cite[Theorem 2.2]{Veraar12} for the Brownian motion $W$ and \cite[Proposition A.2]{BDMKR97} for the compensated random measure $\tilde{\mu}$.
By \cite[Remark 2.1]{Veraar12}, the mild measurability requirement in parts (ii)-(iii) holds if the processes $\alpha^i$ and $\beta^i$ are $\mathrm{Prog}\otimes\mathcal{B}(\R_+)$-measurable, with $\mathrm{Prog}$ denoting the progressive sigma-field on $\Omega\times\R_+$.

\begin{remark}
(1) Even in the case of a single term structure (i.e., $I_0=\{0\}$), our setup generalizes the usual formulations of HJM semimartingale models found in the literature. In particular, Assumption \ref{ass} is weaker than the requirements on the forward rate dynamics stated in \cite[Assumption 5.1]{BDMKR97}.

(2) Some recent works (see for instance \cite{FGGS20,FGS24}) have considered HJM models that do not satisfy quasi-left-continuity (i.e., the set $\{(\omega,t)\in\Omega\times\R_+:\nu(\omega;\{t\}\times E)>0\}$ is not evanescent).
Theorem \ref{thm:HJM} below can be generalized to this situation with an analogous proof. However, we do not pursue this generalization here since the SPDE analysis of Section \ref{sec:SPDE} will necessitate quasi-left-continuity.

(3) The present setup can be extended to the case of an infinite-dimensional Brownian motion $W$, as considered for instance in \cite{F:01,fitate2010} in the context of HJM interest rate models. This generalization is straightforward and all results of our work continue to hold with almost identical proofs.
\end{remark}

For all $i\in I_0$, $x\in E$ and $0\leq t\leq T<+\infty$, let us define
\begin{align*}
\bar{\alpha}^i(t,T) &:= \int_t^T\alpha^i(t,u)\ud u,\\
\bar{\beta}^i(t,T) &:= \int_t^T\beta^i(t,u)\ud u,\\
\bar{\gamma}^i(t,T,x) &:= \int_t^T\gamma^i(t,u,x)\ud u.
\end{align*}
We recall from Section \ref{sec:market} that $S^0$ denotes the constant process equal to one. In analogy to above, this corresponds to $S^0=\cE(Z^0)$, with $Z^0$  given as in \eqref{eq:log_spread} for $i=0$, with $a^0:=0$, $b^0:=0$ and $c^0:=0$.

As a preliminary to Theorem \ref{thm:HJM}, in the following lemma we derive the stochastic exponential representation of the elements of the set $(X^0)^{-1}\{S^iB^i(\cdot,T) : (i,T)\in I_0\times\R_+\}$ (see Definition \ref{def:market}).

\begin{lemma}	\label{lem:int_fwd}
Suppose that Assumption \ref{ass} holds. Then, for every $i\in I_0$ and $T>0$, it holds that 
\[
(X^0)^{-1}S^iB^i(\cdot,T) = S^i_0B^i(0,T)\cE(Y^i(\cdot,T)),
\]
where $Y^i(\cdot,T)=(Y^i(t,T))_{t\in[0,T]}$ is a semimartingale given by
\begin{align*}
Y^i(t,T) &:= \int_0^t\bigl(b^i_s-\bar{\beta}^i(s,T)\bigr)\ud W_s
+\int_0^t\int_E\bigl(c^i_s(x)-\bar{\gamma}^i(s,T,x)\bigr)\tilde{\mu}(\ud s,\ud x)\\
&\quad
+\int_0^t\int_E\Bigl((1+c^i_s(x))(e^{-\bar{\gamma}^i(s,T,x)}-1)+\bar{\gamma}^i(s,T,x)\Bigr)\mu(\ud s,\ud x)	\\
&\quad
+\int_0^t\Bigl(
a^i_s -r_s + f^i(s,s) - \bar{\alpha}^i(s,T) + \frac{1}{2}\|\bar{\beta}^i(s,T)\|^2 -\bar{\beta}^i(s,T)^{\top}b^i_s\Bigr)\ud s.
\end{align*}
\end{lemma}
\begin{proof}
Under Assumption \ref{ass} and by proceeding as in the proof of \cite[Proposition 5.2]{BDMKR97}, we have that
\be\label{eq:exp_bond}\ba
B^i(t,T)
= B^i(0,T)&
\exp\left(\int_0^tf^i(s,s)\ud s - \int_0^t\bar{\alpha}^i(s,T)\ud s - \int_0^t\bar{\beta}^i(s,T)\ud W_s\right.\\
&\left.\qquad
-\int_0^t\int_E\bar{\gamma}^i(s,T,x)\tilde{\mu}(\ud s,\ud x)\right),
\ea\ee
for all $i\in I_0$ and $0\leq t\leq T<+\infty$.
The well-posedness of the ordinary and stochastic integrals in \eqref{eq:exp_bond} is ensured by Assumption \ref{ass}. More precisely, the finiteness of $\int_0^t\bar{\alpha}^i(s,T)\ud s$ follows directly from condition  (ii) in Assumption \ref{ass}. 
Then, Minkowski's integral inequality and condition (iii)  imply that
\begin{align*}
\left(\int_0^T \|\bar{\beta}^i(s,T)\|^2 \ud s\right)^{1/2}
  & =  \left(\int_0^T \Big\|\int_s^T\beta^i(s,u)\ud u\Big\|^2 \ud s\right)^{1/2}\\
& \leq \sum_{j=1}^d\left(\int_0^T\biggl(\int_s^T|\beta^{i,j}(s,u)|\ud u\biggr)^2\ud s\right)^{1/2}	\\
&\leq \sum_{j=1}^d\int_0^T\left(\int_0^u(\beta^{i,j}(s,u))^2\ud s\right)^{1/2}\ud u < +\infty\qquad\text{a.s.},
  \end{align*}
so that $\bar{\beta}^i(\cdot,T)\in L^2_{\rm loc}(W)$. 
By H\"older's inequality and condition (iv) in Assumption \ref{ass}, it holds that 
\begin{align*}
\int_{0}^{T} \int_E (\bar{\gamma}^i(s,T,x))^2 \nu(\ud s,\ud x)
& = 	\int_{0}^{T} \int_E \Bigl( \int_s^T \gamma^i(s,u,x) \ud u \Bigr)^2 \nu(\ud s,\ud x)\\
& \leq   T\int_{0}^{T}\int_E \int_s^T (\gamma^i(s,u,x) )^2 \ud u\,\nu(\ud s,\ud x) < +\infty\qquad\text{a.s.},
\end{align*}
thus ensuring that the stochastic integral $\int_0^{\cdot}\int_E\bar{\gamma}^i(s,x,T)\tilde{\mu}(\ud s,\ud x)$ is well-defined as a local martingale, see \cite[Theorem II.1.33]{jashi03}. 
The finiteness of the integral $\int_0^{\cdot}f^i(s,s)\ud s$ follows similarly as above under the validity of Assumption \ref{ass}.
An application of \cite[Theorem II.8.10]{jashi03} to \eqref{eq:exp_bond} gives the representation 
\be\label{eq:stoch_bond}\ba
B^i(t,T)
= B^i(0,T)&\,
\cE\left(-\int_0^{\cdot}\bar{\beta}^i(s,T)\ud W_s
-\int_0^{\cdot}\int_E\bar{\gamma}^i(s,T,x)\tilde{\mu}(\ud s,\ud x)\right.\\
&\left.\quad\;\,
+\int_0^{\cdot}\int_E\bigl(e^{-\bar{\gamma}^i(s,T,x)}-1+\bar{\gamma}^i(s,T,x)\bigr)\mu(\ud s,\ud x)\right)_t	\\
&\times \exp\left(\int_0^tf^i(s,s)\ud s - \int_0^t\bar{\alpha}^i(s,T)\ud s + \frac{1}{2}\int_0^t\|\bar{\beta}^i(s,T)\|^2\ud s\right).
\ea\ee
The result of the lemma then follows by an application of Yor's formula (see, e.g., \cite[formula II.8.19]{jashi03}), making use of \eqref{eq:stoch_bond} and recalling that $S^i=S^i_0\cE(Z^i)$, where $Z^i$ is given by  \eqref{eq:log_spread}.
\end{proof}

\subsection{Characterization of LMDs}	\label{sec:thm_HJM}

The next theorem is the main result of Section \ref{sec:HJM} and provides necessary and sufficient conditions for the existence of LMDs, together with an explicit description of their structure. Besides considering multiple term structures, this result represents the first complete characterization of LMDs in the context of HJM-type semimartingale models. In a finite-dimensional setup, a related result is \cite[Lemma 2.11]{ChoulliSchweizer16}, from which some arguments in the proof of Theorem \ref{thm:HJM} are adapted.
Adopting the notation of \cite{jashi03}, we denote by the symbols $\cdot$ and $\ast$ stochastic integration with respect to a semimartingale and with respect to a random measure, respectively. 
We refer to Appendix \ref{app:Doleans} for the notion of the Dol\'eans measure $M_{\mu}$ on $(\Omega\times\R_+\times E,\cF\otimes\cB(\R_+)\otimes\cB(E))$ associated to $\mu$ and the corresponding conditional expectation with respect to the sigma-field $\widetilde{\cP}:=\cP\otimes\cB(E)$. 

\begin{theorem}	\label{thm:HJM}
Suppose that Assumption \ref{ass} holds. Then $\cD\neq\emptyset$ if and only if there exist $\lambda\in L^2_{\rm loc}(W)$ and $\psi\in G_{\rm loc}(\mu)$ with $\psi>-1$ such that, for all $i\in I_0$, $T>0$ and a.e. $t\in[0,T]$,
\be	\label{eq:int_jumps}
\int_{\mathcal{C}_{t,T}^i}\bigl((1+c^i_t(x))e^{-\bar{\gamma}^i(t,T,x)}-1\bigr)\bigl(1+\psi_t(x)\bigr)F_t(\ud x) < +\infty
\quad\text{a.s.},
\ee
where $\mathcal{C}_{t,T}^i:=\{x\in E : c^i_t(x)>2e^{\bar{\gamma}^i(t,T,x)}-1\}$, and the following two conditions hold a.s.: 
\begin{enumerate}
\item[(i)]
for all $i\in I_0$ and a.e. $t\in\R_+$, it holds that
\[
a^i_t = r_t - f^i(t,t) - \lambda_t^{\top}b^i_t - \int_Ec^i_t(x)\psi_t(x)F_t(\ud x);
\]
\item[(ii)]
for  all $i\in I_0$, $T>0$ and a.e. $t\in[0,T]$, it holds that
\[
\bar{\alpha}^i(t,T) = \frac{1}{2}\|\bar{\beta}^i(t,T)\|^2 - \bar{\beta}^i(t,T)^{\top}(b^i_t+\lambda_t)
+ \int_E\Bigl(\bigl(1+\psi_t(x)\bigr)\bigl(1+c^i_t(x)\bigr)\bigl(e^{-\bar{\gamma}^i(t,T,x)}-1\bigr)+\bar{\gamma}^i(t,T,x)\Bigr)F_t(\ud x).
\]
\end{enumerate}
Moreover, a strictly positive local martingale $D=(D_t)_{t\geq0}$ belongs to $\cD$ if and only if 
\be	\label{eq:deflator}
D = \cE\left(\lambda\cdot W + \psi\ast\tilde{\mu} + N\right),
\ee
where $\lambda\in L^2_{\rm loc}(W)$ and $\psi\in G_{\rm loc}(\mu)$ satisfy the above properties and $N=(N_t)_{t\geq0}$ is a local martingale with $N_0=0$, satisfying $\langle N,W^j\rangle=0$, for all $j=1,\ldots,d$, and
$M_{\mu}[\Delta N|\widetilde{\cP}]=0$.
\end{theorem}
\begin{proof}
As a preliminary, let us introduce a shorthand notation that will be used throughout the proof. For fixed but arbitrary elements $i\in I_0$ and $T>0$, we define $Y := Y^i(\cdot,T)$ (see Lemma \ref{lem:int_fwd}) and
\begin{align*}
\sigma_t &:= b^i_t-\bar{\beta}^i(t,T),\\
v_t(x) &:= c^i_t(x)-\bar{\gamma}^i(t,T,x),\\
u_t(x) &:= (1+c^i_t(x))(e^{-\bar{\gamma}^i(t,T,x)}-1)+\bar{\gamma}^i(t,T,x),\\
A_t &:= \int_0^t\Bigl(a^i_s -r_s + f^i(s,s) - \bar{\alpha}^i(s,T) + \frac{1}{2}\|\bar{\beta}^i(s,T)\|^2 -\bar{\beta}^i(s,T)^{\top}b^i_s\Bigr)\ud s,
\end{align*}
for all $t\in[0,T]$ and $x\in E$.
By Lemma \ref{lem:int_fwd}, the semimartingale $Y$ can be written as follows:
\be\label{eq:Y1}
Y = A + \sigma\cdot W + v\ast\tilde{\mu} + u\ast\mu.
\ee
To ease the presentation, we divide the proof into four steps.

{\em (1)}
Let us consider the function $h:\R\to[0,1]$ given by $h(x):=x\ind_{\{|x|\leq 1\}}$. By \cite[Theorem II.2.34]{jashi03}, the canonical representation of the semimartingale $Y$ corresponding to the truncation function $h$  is
\be\label{eq:Y2}
Y = B(h) + \sigma\cdot W + (h\circ(u+v))\ast\tilde{\mu} + \bigl(u+v-h\circ(u+v)\bigr)\ast\mu,
\ee
where $B(h)$ is a predictable process of finite variation. By comparing \eqref{eq:Y1} and \eqref{eq:Y2} we obtain that
\be\label{eq:Y3}
\bigl(v-h\circ(u+v)\bigr)\ast\mu
= A-B(h)+\bigl(v-h\circ(u+v)\bigr)\ast\tilde{\mu},
\ee
so that the finite variation process $(v-h\circ(u+v))\ast\mu$ is a special semimartingale and, therefore, of locally integrable variation (see, e.g., \cite[Proposition I.4.23]{jashi03}). In turn, this implies that the process $|v-h\circ(u+v)|\ast\nu$ is locally integrable and, hence, equation \eqref{eq:Y3} can be rewritten more simply as 
\be\label{eq:Bh}
B(h) = A - \bigl(v-h\circ(u+v)\bigr)\ast\nu.
\ee

{\em (2)}
Let $D\in\cD$. Since $D$ is a strictly positive local martingale with $D_0=1$, its stochastic logarithm $L:=D^{-1}_-\cdot D$ is a local martingale with $L_0=0$ and $\Delta L>-1$. By  \cite[Lemma III.4.24]{jashi03}, it holds that
\be\label{eq:L}
L = \lambda\cdot W + \psi\ast\tilde{\mu} + N,
\ee
for some $\lambda\in L^2_{\rm loc}(W)$, $\psi\in G_{\rm loc}(\mu)$ and a local martingale $N$ with $N_0=0$ satisfying $\langle N,W^j\rangle=0$, for all $j=1,\ldots,d$, and $M_{\mu}[\Delta N|\widetilde{\cP}]=0$.
Moreover, by part b) of \cite[Theorem III.4.20]{jashi03}, it holds that $\psi=M_{\mu}[\Delta L|\widetilde{\cP}]$. Since $\Delta L>-1$, this directly implies that the function $\psi$ takes values in $(-1,+\infty)$.
By Lemma \ref{lem:int_fwd} and an application of Yor's formula (see, e.g., \cite[formula II.8.19]{jashi03}) we obtain that
\[
D(X^0)^{-1}S^iB^i(\cdot,T) = S^i_0B^i(0,T)\cE(L+Y+[L,Y]).
\]
This shows that $D\in\cD$ if and only if $Y+[L,Y]$ is a local martingale. 
From \eqref{eq:Y2} and \eqref{eq:L} we obtain
\[
[L,Y] 
= \int_0^{\cdot}\lambda^{\top}_s\sigma_s\ud s + \sum_{s>0}\Delta L_s\Delta Y_s
%= \int_0^{\cdot}\lambda^{\top}_s\sigma_s\ud s + \bigl(\psi(u+v)\bigr)\ast\mu + \sum_{s>0}\Delta N_s\Delta Y_s
= \int_0^{\cdot}\lambda^{\top}_s\sigma_s\ud s + \bigl((\psi+\Delta N)(u+v)\bigr)\ast\mu.
\]
Making use of \eqref{eq:Y2} again, we obtain that the local martingale property of  $Y+[L,Y]$ is equivalent to
\be\label{eq:FV}
B(h) + \int_0^{\cdot}\lambda^{\top}_s\sigma_s\ud s 
+ \bigl((u+v)(1+\psi+\Delta N)-h\circ(u+v)\bigr)\ast\mu
\ee
being a local martingale. 
In that case, the process $C(h):=((u+v)(1+\psi+\Delta N)-h\circ(u+v))\ast\mu$ is  a finite variation process of locally integrable variation. By Lemma \ref{lem:cond_expec}, its compensator is given by
\be\label{eq:Cp}\ba
C(h)^p
&= M_{\mu}\bigl[(u+v)(1+\psi+\Delta N)-h\circ(u+v)\big|\widetilde{\cP}\bigr]\ast\nu	\\
&= \bigl((u+v)(1+\psi+M_{\mu}[\Delta N|\widetilde{\cP}])-h\circ(u+v)\bigr)\ast\nu	\\
&= \bigl((u+v)(1+\psi)-h\circ(u+v)\bigr)\ast\nu,
\ea\ee
where we made use of the $\widetilde{\cP}$-measurability of the functions $u$, $v$ and $\psi$ together with the fact that $M_{\mu}[\Delta N|\widetilde{\cP}]=0$.
By compensating $C(h)$, we obtain that the process in \eqref{eq:FV} is a local martingale if and only if the process
$
H:=B(h) + \int_0^{\cdot}\lambda^{\top}_s\sigma_s\ud s + C(h)^p
$
is a local martingale. By \eqref{eq:Bh}, the process $H$ can be equivalently rewritten as
\be\label{eq:pred_FV}
%B(h) + \int_0^{\cdot}\lambda^{\top}_s\sigma_s\ud s + C(h)^p
H = A + \int_0^{\cdot}\lambda^{\top}_s\sigma_s\ud s
+ \bigl((u+v)\psi+u\bigr)\ast\nu.
\ee
Since $H$ is a predictable process of finite variation, it can be a local martingale if and only if it is equal to zero, up to an evanescent set (see \cite[Corollary I.3.16]{jashi03}). Recalling the notation introduced at the beginning of the proof, this means that, outside of a subset of $\Omega\times\R_+$ of $(\PP\otimes \ud t)$-measure zero,
\be\label{eq:drift0}\ba
&a^i_t -r_t+ f^i(t,t) - \bar{\alpha}^i(t,T) + \frac{1}{2}\|\bar{\beta}^i(t,T)\|^2 -\bar{\beta}^i(t,T)^{\top}b^i_t + \lambda^{\top}_t\bigl(b^i_t-\bar{\beta}^i(t,T)\bigr)	\\
&\; + \int_E\left(\bigl(1+c^i_t(x)\bigr)\bigl(1+\psi_t(x)\bigr)\bigl(e^{-\bar{\gamma}^i(t,T,x)}-1\bigr)+\bar{\gamma}^i(t,T,x)+c^i_t(x)\psi_t(x)\right)F_t(\ud x)
=0.
\ea\ee
Since \eqref{eq:drift0} must hold for all $T\in\R_+$, we can take $T=t$ and obtain condition (i) in the statement of the theorem. In turn, inserting condition (i) into \eqref{eq:drift0} gives condition (ii).

{\em (3)}
To complete the first part of the proof, it remains to show that \eqref{eq:int_jumps} holds. Since the compensator $C(h)^p$ introduced above is a predictable process of finite variation, its variation
\[
\bigl|(u+v)(1+\psi)-h\circ(u+v)\bigr|\ast\nu
\]  
is locally integrable. In particular, noting that $u+v\geq-1$, this implies that the increasing process
\[
\bigl((u+v)(1+\psi)\ind_{\{u+v>1\}}\bigr)\ast\nu
= \int_0^{\cdot}\int_{\mathcal{C}^i_{s,T}}\bigl((1+c^i_s(x))e^{-\bar{\gamma}^i(s,T,x)}-1\bigr)\bigl(1+\psi_s(x)\bigr)F_s(\ud x)\ud s
\]
is locally integrable, where the set $\mathcal{C}^i_{s,T}\subset E$ has been defined as in the statement of theorem. This proves the validity of the integrability condition \eqref{eq:int_jumps} for a.e. $t\in[0,T]$ whenever $\cD\neq\emptyset$.

{\em (4)}
Conversely, suppose there exist $\lambda\in L^2_{\rm loc}(W)$ and $\psi\in G_{\rm loc}(\mu)$ with $\psi>-1$ such that \eqref{eq:int_jumps} and conditions (i)-(ii) in the statement of the theorem hold. We shall prove that $\cD\neq\emptyset$. To this effect, let us define a strictly positive local martingale $D$ as in \eqref{eq:deflator}, i.e., 
\[
D := \cE(L),
\qquad\text{ with } L:=\lambda\cdot W + \psi\ast\tilde{\mu} + N,
\]
for some local martingale $N$ with $N_0=0$ and satisfying $\langle N,W^j\rangle=0$, for all $j=1,\ldots,d$, and $M_{\mu}[\Delta N|\widetilde{\cP}]=0$.
We shall first prove that \eqref{eq:int_jumps} together with conditions (i)-(ii) implies that the finite variation process $C(h)=((u+v)(1+\psi+\Delta N)-h\circ(u+v))\ast\mu$ is of locally integrable variation.
Observe that, using the notation introduced above and proceeding as in the proof of \cite[Lemma 2.11]{ChoulliSchweizer16},
\be\label{eq:estimates}\ba
%\ind_{\{u+v\}\leq 1}\bigl|(u+v)(1+\psi+\Delta N)-h\circ(u+v)\bigr|\ast\mu
\bigl|h\circ(u+v)(\psi+\Delta N)\bigr|\ast\mu
&= \sum_{s>0}|\Delta Y_s|\ind_{\{|\Delta Y_s|\leq1\}}|\Delta L_s|	\\
&\leq \biggl(\sum_{s>0}\Delta Y_s^2\ind_{\{|\Delta Y_s|\leq 1\}}\biggr)^{1/2}\biggl(\sum_{s>0}\Delta L^2_s\biggr)^{1/2}	\\
&\leq \left(\ind_{\{|\Delta Y|\leq 1\}}\cdot[Y]\right)^{1/2}\biggl(\sum_{s>0}\Delta L^2_s\biggr)^{1/2}.
\ea\ee
The process $\ind_{\{|\Delta Y|\leq1\}}\cdot[Y]$ is locally bounded, being an increasing process with bounded jumps, while the process $(\sum_{s>0}\Delta L_s^2)^{1/2}$ is locally integrable, due to the fact that $L$ is a local martingale (see, e.g., \cite[Corollary I.4.55]{jashi03}). This implies that the increasing process $|h\circ(u+v)(\psi+\Delta N)|\ast\mu$ is locally integrable.
We then proceed by showing that the increasing process $|\ind_{\{u+v>1\}}(u+v)(1+\psi+\Delta N)|\ast\mu$ is locally integrable too. Notice that $u+v\geq-1$ and also $1+\psi+\Delta N>0$, since $\Delta L>-1$. Therefore, by Lemma \ref{lem:cond_expec} and using the fact that $M_{\mu}[\Delta N|\widetilde{\cP}]=0$, the process $(\ind_{\{u+v>1\}}(u+v)(1+\psi+\Delta N))\ast\mu$ is locally integrable if and only if $(\ind_{\{u+v>1\}}(u+v)(1+\psi))\ast\nu$ is locally integrable.
To prove that the latter property holds, observe first that \eqref{eq:int_jumps} corresponds to the following condition, for a.e. $t\in[0,T]$:
\[
\int_E\ind_{\{x:u_t(x)+v_t(x)>1\}}\bigl(u_t(x)+v_t(x)\bigr)\bigl(1+\psi_t(x)\bigr)F_t(\ud x)<+\infty.
\]
Moreover, conditions (i)-(ii) together imply that equation \eqref{eq:drift0} is satisfied, which leads to
\be\label{eq:use_drift}\ba
&\int_E\ind_{\{x:u_t(x)+v_t(x)>1\}}\bigl(u_t(x)+v_t(x)\bigr)\bigl(1+\psi_t(x)\bigr)F_t(\ud x)	\\
& = -\dot{A}-\lambda^{\top}_t\sigma_t
-\int_E\Bigl(h\bigl(u_t(x)+v_t(x)\bigr)\bigl(1+\psi_t(x)\bigr)-v_t(x)\Bigr)F_t(\ud x)	\\
&= -\dot{A}-\lambda^{\top}_t\sigma_t
-\int_Eh\bigl(u_t(x)+v_t(x)\bigr)\psi_t(x)F_t(\ud x)
+\int_E\bigl(v_t(x)-h\bigl(u_t(x)+v_t(x)\bigr)\bigr)F_t(\ud x)	\\
&= -\dot{B}(h)-\lambda^{\top}_t\sigma_t
-\int_Eh\bigl(u_t(x)+v_t(x)\bigr)\psi_t(x)F_t(\ud x),
\ea\ee
where we have denoted by $\dot{A}$ and $\dot{B}(h)$ the densities of the absolutely continuous processes $A$ and $B(h)$ with respect to the Lebesgue measure, respectively. 
In the second equality of \eqref{eq:use_drift}, we made use of the fact that the increasing process $|v-h\circ(u+v)|\ast\nu$ is locally integrable and, therefore, the integral $\int_E(v_t(x)-h(u_t(x)+v_t(x)))F_t(\ud x)$ is always finite outside a set of $(\PP\otimes\ud t)$-measure zero.
Since $\lambda$ and $\sigma$ belong to $L^2_{\rm loc}(W)$, the Cauchy-Schwarz inequality implies that $\int_0^T|\lambda^{\top}_t\sigma_t|\ud t<+\infty$ a.s. 
Moreover, the same estimates as in \eqref{eq:estimates} together with an application of Lemma \ref{lem:cond_expec} allow to show that $(h\circ(u+v)\psi)\ast\nu$ is well-defined as a predictable process of finite variation. In view of \eqref{eq:use_drift}, these facts enable us to deduce that $(\ind_{\{u+v>1\}}(u+v)(1+\psi))\ast\nu$ is well-defined as an increasing process.
Since every predictable increasing process is locally integrable (see, e.g., \cite[Lemma I.3.10]{jashi03}) and
\[
\bigl|(u+v)(1+\psi+\Delta N)-h\circ(u+v)\bigr|\ast\mu
= \bigl|h\circ(u+v)(\psi+\Delta N)\bigr|\ast\mu
+ \bigl(\ind_{\{u+v>1\}}(u+v)(1+\psi+\Delta N)\bigr)\ast\mu,
\]
we have thus proved that the finite variation process $C(h)$ is of locally integrable variation. Therefore, its compensator $C(h)^p$ exists and is given by \eqref{eq:Cp}. Since conditions (i)-(ii) together imply that the process $H$ in \eqref{eq:pred_FV} vanishes, the process \eqref{eq:FV} results to be a local martingale. As explained above, the latter property is equivalent to the local martingale property of $Y+[L,Y]$. Since $D$ is strictly positive, this suffices to conclude that $D\in\cD$, thus completing the proof.
\end{proof}

Theorem \ref{thm:HJM} generalizes and unifies existing results on absence of arbitrage for HJM semimartingale models in the quasi-left-continuous case, even for a single term structure (see also Remark \ref{rem:comparison} below). Conditions (i)-(ii) of the theorem share the same structural form of conditions (i)-(ii) in \cite[Corollary 3.10]{FGGS20}. However, unlike \cite[Corollary 3.10]{FGGS20}, which imposes more stringent technical requirements and can only be applied (with $\mathbb{Q}=\mathbb{P}$) to verify whether a given local martingale is an LMD,  Theorem \ref{thm:HJM} provides a complete characterization of the whole set of LMDs under minimal technical conditions.

\begin{remark}
Theorem \ref{thm:HJM} implies that, whenever $\cD\neq\emptyset$, it always holds that $\widehat{D}:=\cE(\lambda\cdot W + \psi\ast\tilde{\mu})\in\cD$, where $\lambda$ and $\psi$ satisfy \eqref{eq:int_jumps} and conditions (i)-(ii) of the theorem. The LMD $\widehat{D}$ corresponds to taking $N\equiv 0$ in \eqref{eq:deflator} and can be regarded as the {\em minimal} LMD, in the spirit of the minimal martingale measure introduced in \cite{FollmerSchweizer91} (see also \cite[Corollary 2.14]{ChoulliSchweizer16}
in a finite-dimensional semimartingale setup).
\end{remark}

\begin{remark}	\label{rem:HJM}
Let us briefly comment on conditions (i)-(ii) in Theorem \ref{thm:HJM}. Condition (i), for $i=0$, reduces to the classical consistency condition $r_t=f^0(t,t)$ between the locally riskless short rate and the short end of the riskless forward rate. 
For $i\in I$, condition (i) requires that, at the short end, the instantaneous yield of $S^iB^i(\cdot,T)$ equals the  riskless short rate $r_t$ plus a risk premium term which depends on the volatility of the spot process $S^i$. 
For $i=0$, differentiating condition (ii) leads to
\[
\alpha^0(t,T) = \beta^0(t,T)^{\top}\bigl(\bar{\beta}^0(t,T)-\lambda_t\bigr)+\int_E\gamma^0(t,T,x)\left(1-\bigl(1+\psi_t(x)\bigr)e^{-\bar{\gamma}^0(t,T,x)}\right)F_t(\ud x),
\]
provided that $\int_E\sup_{T\geq t}|\gamma^0(t,T,x)(1-(1+\psi_t(x))\exp(-\bar{\gamma}^0(t,T,x)))|F_t(\ud x)<+\infty$, so that we are allowed to differentiate under the integral sign.
This condition represents the real-world HJM drift restriction for the riskless term structure and has already been derived in \cite{BLNSP10} and \cite[Section 3]{PT15}, albeit under more stringent technical requirements.
For $i\in I$, differentiating condition (ii) leads to the real-world HJM drift restriction for the risky term structures:
\[
\alpha^i(t,T) = \beta^i(t,T)^{\top}\bigl(\bar{\beta}^i(t,T)-b^i_t-\lambda_t\bigr) + \int_E\gamma^i(t,T,x)\left(1-\bigl(1+\psi_t(x)\bigr)\bigl(1+c^i_t(x)\bigr)e^{-\bar{\gamma}^i(t,T,x)}\right)F_t(\ud x),
\]
provided that we are allowed to differentiate under the integral sign, similarly as above.
\end{remark}

\subsection{Monotonicity of term structures}
\label{sec:monotone}

The examples discussed in Section \ref{sec:examples} show that, in many cases, there exists a natural ordering among the risky term structures according to their level of risk. It is, therefore, of interest to provide clear conditions ensuring that, for each $T>0$, the family $\{S^iB^i(\cdot,T);i\in I\}$ is ordered with respect to $i$. 
Here, we do not aim at a general characterization of ordered term structures, but rather at simple criteria that can be easily verified in concrete specifications. To this end, we introduce the following condition.

\begin{condition}	\label{cond:order}
Outside some subset of $\Omega\times\R_+$ of $(\PP\otimes\ud t)$-measure zero, it holds that, for all $t\in\R_+$,
\[
\int_0^ta^i_u\ud u \leq \int_0^ta^j_u\ud u,
\qquad\text{ for all $i,j\in I$ with $i<j$},
\]
%\[
%a^i_t\leq a^j_t,
%\qquad\text{ for all $i,j\in I$ with $i<j$},
%\]
and there exists an element $i_0\in I$ such that $b^i_t=b^{i_0}_t$ and $c^i_t(x)=c^{i_0}_t(x)$, for all $x\in E$ and $i\in I$.
\end{condition}

By \eqref{eq:log_spread} and a straightforward comparison argument, Condition \ref{cond:order} ensures that $\cE(Z^i)\leq\cE(Z^j)$ up to an evanescent set, for all $i,j\in I$ with $i<j$. 
%By \eqref{eq:log_spread} and a straightforward comparison argument, Condition \ref{cond:order} ensures that $\PP(Z^i_t\leq Z^j_t,\text{ for all } t\in\R_+)=1$, for all $i,j\in I$ with $i<j$. 
In turn, this implies that the spot processes are also ordered with respect to $i$, provided that their initial values $\{S^i_0;i\in I\}$ respect the same ordering.
%While ordered spot processes can be obtained under weaker conditions, Condition \ref{cond:order} has the advantage of being easily verifiable in concrete applications.

We now introduce a simple property that will enable us to transfer the ordering of the spot processes onto the whole family of risky term structures.

\begin{definition}\label{def:fair}
Let $D\in\cD$ and $i\in I_0$. We say that the term structure indexed by $i$ is {\em fairly priced} by $D$ if the process $D(X^0)^{-1}S^iB^i(\cdot,T)$ is a true martingale, for every $T>0$.
\end{definition}

\begin{remark}
In view of Definition \ref{def:deflators}, the process $D(X^0)^{-1}S^iB^i(\cdot,T)$ is a local martingale, for every $D\in\cD$ and $(i,T)\in I_0\times\R_+$. Definition \ref{def:fair} strengthens this property into a true martingale property. 
This corresponds to a generalization of the concept of fair pricing in the context of the benchmark approach, where a price process is said to be fair if it is a true  martingale when denominated in units of the growth-optimal portfolio (see \cite[Chapter 10]{PH} and \cite[Section 2.2]{BLNSP10} for term structure models).
\end{remark}

\begin{proposition}	\label{prop:order}
Suppose that Assumption \ref{ass} holds and $\cD\neq\emptyset$. For $D\in\cD$ and $i,j\in I$ with $i<j$, assume that the term structures indexed by $i$ and $j$ are fairly priced by $D$, in the sense of Definition \ref{def:fair}. Assume furthermore that Condition \ref{cond:order} holds and $S^i_0\leq S^j_0$.
Then, it holds a.s. that
\be\label{eq:order}
S^i_tB^i(t,T) \leq S^j_tB^j(t,T),
\qquad\text{ for all $t\in[0,T]$ and $T>0$.}
\ee
\end{proposition}
\begin{proof}
The assumptions imply that the process $D(X^0)^{-1}S^kB^k(\cdot,T)$ is a true martingale, for $k\in\{i,j\}$ and for all $T>0$. Therefore, making use of the terminal condition $B^k(T,T)=1$, it follows that
\[
S_t^kB^k(t,T)
= \frac{\EE[D_T(X_T^0)^{-1}S_T^k|\cF_t]}{D_t(X_t^0)^{-1}}
\quad\text{ a.s.,}
\]
for all $T>0$, $t\in[0,T]$ and $k\in\{i,j\}$. Inequality \eqref{eq:order} follows by noticing that, as explained above, Condition \ref{cond:order} and the assumption that $S^i_0\leq S^j_0$ together imply that $S^i_T\leq S^j_T$ a.s.
\end{proof}

\begin{remark}
The ordering of the risky term structures is especially relevant for multi-curve interest rate models (see Section \ref{ex:multicurve}). In that context, sufficient conditions for ordered term structures have been derived in \cite[Corollary 3.17]{CFG:16}, which can be recovered as a special case of Proposition \ref{prop:order} above. 
\end{remark}

In Section \ref{sec:SPDE_invariant}, we shall establish property \eqref{eq:order}  by studying the invariance properties of the SPDE associated to the model without requiring the fair pricing condition (see Proposition \ref{prop-monotonicity} below).

\subsection{The risk-neutral setup}
\label{sec:RN}

We have so far developed a general framework which does not rely on the existence of a risk-neutral probability measure, being based on a weaker concept of market viability (see Section \ref{sec:FTAP}).
Given the widespread use of risk-neutral valuation in finance, it is nevertheless useful to specialize Theorem \ref{thm:HJM} to the risk-neutral setup.
Adopting the notation of Section \ref{sec:FTAP}, we say that a probability measure $\QQ\sim\PP$ on $(\Omega,\cF)$ is a {\em risk-neutral probability} if $X^{(i,T)}$ is a $\QQ$-local martingale, for every $(i,T)\in I_0\times\R_+$.
We recall from Girsanov's theorem (see \cite[Theorems III.3.11 and III.3.17]{jashi03}) that, if $\QQ\sim\PP$, there exist $\lambda \in L^2_{\rm loc}(W)$ and $\psi\in G_{\rm loc}(\mu)$ with $\psi>-1$ such that $W^{\QQ}:=W+\int_0^{\cdot}\lambda_s\ud s$ is a $d$-dimensional Brownian motion under $\QQ$ and the compensator of $\mu$ under $\QQ$ is given by
\[
\nu^{\QQ}(\ud t,\ud x)=F^{\QQ}_t(\ud x)\ud t
=\bigl(1+\psi_t(x)\bigr)F_t(\ud x)\ud t.
\]
We denote by $\tilde{\mu}^{\QQ}(\ud t,\ud x):=\mu(\ud t,\ud x)-\nu^{\QQ}(\ud t,\ud x)$ the compensated random measure under $\QQ$.

The next result is a consequence of Theorem \ref{thm:HJM} and provides necessary and sufficient conditions for the existence of a risk-neutral probability.

\begin{corollary}\label{cor:RN}
Suppose that Assumption \ref{ass} holds. Then there exists a risk-neutral probability $\QQ$ if and only if there exist $\lambda \in L^2_{\rm loc}(W)$, $\psi\in G_{\rm loc}(\mu)$ and a local martingale $N=(N_t)_{t\geq0}$ satisfying all conditions of Theorem \ref{thm:HJM} such that $\cE(\lambda\cdot W + \psi\ast\tilde{\mu}+N)$ is a strictly positive uniformly integrable martingale.
In that case, for all $i\in I_0$, it holds that
\be	\label{eq:spread_RN}
Z^i = \int_0^{\cdot}a^{\QQ,i}_s\ud s + \int_0^{\cdot}b^i_s\ud W^{\QQ}_s + \int_0^{\cdot}\int_Ec^i_s(x)\tilde{\mu}^{\QQ}(\ud s,\ud x),
\ee
with
\be	\label{eq:drift_spread_RN}
a^{\QQ,i}_t = r_t - f^i(t,t),
\ee
and, if $\int_0^T\int_{\{x\in E:|\gamma^i(t,T,x)|>1\}}|\gamma^i(t,T,x)\psi_t(x)|F_t(\ud x)\ud t<+\infty$ a.s., for all $T>0$, it also holds that
\be	\label{eq:fwd_RN}
f^i(\cdot,T) = f^i(0,T) + \int_0^{\cdot}\alpha^{\QQ,i}(s,T)\ud s + \int_0^{\cdot}\beta^i(s,T)\ud W^{\QQ}_s + \int_0^{\cdot}\int_E\gamma^i(s,T,x)\tilde{\mu}^{\QQ}(\ud s,\ud x),
\ee
where
\be	\label{eq:drift_fwd_RN}
\bar{\alpha}^{\QQ,i}(t,T) = \frac{1}{2}\|\bar{\beta}^i(t,T)\|^2 - (b^i_t)^{\top}\bar{\beta}^i(t,T) + \int_E\bigl((1+c^i_t(x))(e^{-\bar{\gamma}^i(t,T,x)}-1)+\bar{\gamma}^i(t,T,x)\bigr)F^{\QQ}_t(\ud x).
\ee
\end{corollary}
\begin{proof}
Let $\QQ$ be a probability measure on $(\Omega,\cF)$ with $\QQ\sim\PP$ and density process $D=(D_t)_{t\geq0}$, i.e., $D_t=\ud\QQ|_{\cF_t}/\ud\PP|_{\cF_t}$, for all $t\geq0$. \cite[Lemma III.5.17]{jashi03} implies that $D$ admits a representation of the form \eqref{eq:deflator}. By definition, $\QQ$ is a risk-neutral probability if and only if $D\in\cD$.
In view of Theorem \ref{thm:HJM}, we only have to prove that \eqref{eq:spread_RN}--\eqref{eq:drift_fwd_RN} hold under $\QQ$.
For every $i\in I$, the process $Z^i$ can be represented as in \eqref{eq:spread_RN} if and only if it is a special semimartingale under $\QQ$ and this holds if and only if $c^i\ind_{\{c^i>1\}}\ast\nu^{\QQ}$ is of finite variation. 
Condition \eqref{eq:int_jumps} for $t=T$ implies that $\int_{\{x\in E:c^i_t(x)>1\}}c^i_t(x)F^{\QQ}_t(\ud x)<+\infty$ a.s. for a.e. $t\in\R_+$.
Making use of condition (i) in Theorem \ref{thm:HJM} and the fact that $Z^i$ is a special semimartingale under $\PP$, as a consequence of \eqref{eq:log_spread}, we can compute
\begin{align*}
& \int_{\{x\in E:c^i_t(x)>1\}}c^i_t(x)F^{\QQ}_t(\ud x)	\\
&\; = \int_{\{x\in E:c^i_t(x)>1\}}c^i_t(x)F_t(\ud x)
+ \int_Ec^i_t(x)\psi_t(x)F_t(\ud x)
- \int_{\{x\in E:c^i_t(x)\leq 1\}}c^i_t(x)\psi_t(x)F_t(\ud x)	\\
&\;= \int_{\{x\in E:c^i_t(x)>1\}}c^i_t(x)F_t(\ud x)
+r_t-f^i(t,t)- \lambda_t^{\top}b^i_t-a^i_t
- \int_{\{x\in E:c^i_t(x)\leq 1\}}c^i_t(x)\psi_t(x)F_t(\ud x).
\end{align*}
Since all processes appearing on the last line of the above equation are integrable with respect to $\ud t$, this implies that $c^i\ind_{\{c^i>1\}}\ast\nu^{\QQ}$ is of finite variation. Equation \eqref{eq:drift_spread_RN} then follows as a direct consequence of condition (i) of Theorem \ref{thm:HJM} and Girsanov's theorem.
The proof of \eqref{eq:fwd_RN}-\eqref{eq:drift_fwd_RN} is analogous, but in this case the additional integrability requirement $\int_0^T\int_{\{x\in E:|\gamma^i(t,T,x)|>1\}}|\gamma^i(t,T,x)\psi_t(x)|F_t(\ud x)\ud t<+\infty$ a.s. is necessary, since it does not follow automatically from the conditions of Theorem \ref{thm:HJM}.
\end{proof}

\begin{remark}	\label{rem:comparison}
To the best of our knowledge, Corollary \ref{cor:RN} provides the most general characterization of risk-neutral probabilities for HJM semimartingale quasi-left-continuous models. In particular:
\begin{enumerate}
\item[a)] Conditions \eqref{eq:drift_spread_RN} and \eqref{eq:drift_fwd_RN} in Corollary \ref{cor:RN} correspond respectively to conditions (i) and (ii) in \cite[Corollary 3.10]{FGGS20} in the quasi-left-continuous case. However, our Corollary \ref{cor:RN} holds under less stringent requirements on the processes $\psi$, $c^i$, $\gamma^i$, for $i\in I$.
\item[b)] Under the additional assumptions that $\mu$ is a Poisson random measure corresponding to the jump measure of an $\R^d$-valued It\^o-semimartingale $X$ and that $\gamma^i(t,T,x)=\beta^i(t,T)x$, for all $i\in I_0$ and $(t,T,x)\in\R^2_+\times\R^d$, it can be easily checked that conditions \eqref{eq:drift_spread_RN} and \eqref{eq:drift_fwd_RN} in Corollary \ref{cor:RN} reduce respectively to the {\em consistency} and {\em drift} conditions of \cite[Theorem 3.15]{CFG:16} in the context of risk-neutral HJM multi-curve models driven by It\^o-semimartingales.
\item[c)] In the case $I_0=\{0\}$, corresponding to a HJM model with a single term structure, Corollary \ref{cor:RN} allows to recover the results of \cite[Propositions 5.3 and 5.6]{BDMKR97}, while requiring the weaker integrability condition \eqref{eq:int_jumps}.
\item[d)] When applied to foreign exchange models (see Section \ref{ex:FXmarkets}), Corollary \ref{cor:RN} extends the result of \cite[Proposition 2.1.15]{phdkoval} by requiring substantially weaker integrability properties.
\end{enumerate}
\end{remark}

\subsection{An example: the minimal market model}

We present a simple example of a model which falls into the real-world HJM setup developed above, while not admitting a risk-neutral probability. The example corresponds to an extension of the {\em minimal market model} (MMM), see \cite[Chapter 13]{PH} and \cite[Section 3.2]{BLNSP10}. 
Denoting by $X^0=\exp(\int_0^{\cdot}r_s\ud s)$ the num\'eraire, we recall that in the stylized MMM the $X^0$-discounted growth-optimal portfolio (GOP) process $\bar{X}^*=(\bar{X}^*_t)_{t\geq0}$ satisfies the SDE
\[
\ud \bar{X}_t^* = \alpha^*(t) \ud t + \sqrt{ \bar{X}_t^*\, \alpha^*(t) } \, \ud W_t,
\qquad \bar{X}^*_0=1,
\]
with a one-dimensional Brownian process $W$ and where the drift $\alpha^*$ is modeled as a function of the form $\alpha^*(t)=\alpha_0\exp(\eta t)$, for all $t\geq0$.
The model is parameterized by a positive initial value $\alpha_0 > 0$ and a constant net growth rate $\eta > 0$. 
We let $\{S^i;i\in I\}$ be a family of strictly positive processes representing the spot processes.
Similarly to \cite[Section 3]{Miller-Platen}, we impose the following assumption.\footnote{We recall that two processes $X=(X_t)_{t\geq0}$ and $Y=(Y_t)_{t\geq0}$ on $(\Omega,\bF,\PP)$ are said to be $\bF$-conditionally independent if
\[
\EE[f(X_T) g(Y_T) | \cF_t] = \EE[f(X_T) | \cF_t]  \EE[g(Y_T)| \cF_t], 
\qquad\text{ for all } 0\leq t\leq T<+\infty,
\]
for all measurable functions $f,g : \R \to \R_+$. This property is obviously satisfied if one of the two processes is deterministic.}

\begin{assumption}\label{ass-GOP}
The following conditions hold:
\begin{enumerate}
\item the processes $\bar{X}^*$ and $X^0$ are $\bF$-conditionally independent;
\item the processes $X^*:=X^0\bar{X}^*$ and $S^i$ are $\bF$-conditionally independent, for each $i\in I$.
\end{enumerate}
\end{assumption}

It is well-known that the MMM does not admit a risk-neutral probability (see \cite[Section 13.3]{PH}). Therefore, prices are computed by relying on the real-world pricing formula (see \cite[Section 10.4]{PH}). In particular, the price $B^0(t,T)$ of a riskless bond is given by
\begin{align}\label{MMM-bond-0}
B^0(t,T) = \EE \bigg[ \frac{X_t^*}{X_T^*} \bigg| \cF_t \bigg] 
= \EE \bigg[ \frac{X^0_t}{X^0_T}  \frac{\bar{X}_t^*}{\bar{X}_T^*} \bigg| \cF_t \bigg] 
= \EE \bigg[ \frac{X^0_t}{X^0_T} \bigg| \cF_t \bigg]  \EE \bigg[ \frac{\bar{X}_t^*}{\bar{X}_T^*}  \bigg| \cF_t \bigg] 
= G(t,T)  M(t,T),
\end{align}
where the \emph{short rate contribution} to the riskless bond price is defined as $G(t,T):=\EE[\exp(-\int_t^T r_s \ud s)|\cF_t]$ and the \emph{market price of risk contribution} is defined as $M(t,T):=\EE [ \bar{X}_t^*/\bar{X}_T^*| \cF_t ]$, for all $0\leq t\leq T$.
According to equation (13.3.4) in \cite{PH}, it holds that
\begin{align}\label{MMM-contribution}
M(t,T) = 1 - \exp \bigg( - \frac{\bar{X}_t^*}{2(\varphi(T) - \varphi(t))} \bigg),
\end{align}
where the transformation $\varphi : \R_+ \to \R_+$ is given by
\[
\varphi(t) = \frac{\alpha_0}{4 \eta} ( e^{\eta t} - 1 ), 
\qquad \text{ for all }t \geq 0.
\]

Risky bond prices $B^i(t,T)$, for $i\in I$, are also computed by relying on the real-world pricing formula:
\be\label{MMM-bond-i}
B^i(t,T) 
= \EE \bigg[ \frac{X_t^*}{X_T^*}  \frac{S^i_T}{S^i_t} \bigg| \cF_t \bigg] 
= \EE \bigg[ \frac{X_t^*}{X_T^*}  \bigg| \cF_t \bigg]  \frac{\EE [ S^i_T | \cF_t ]}{S^i_t}
= B^0(t,T)  \frac{\EE [ S^i_T | \cF_t ]}{S^i_t}.
\ee
We can immediately observe that the definition of riskless and risky bond prices through the real-world pricing formulas \eqref{MMM-bond-0} and \eqref{MMM-bond-i} implies that $1/\bar{X}^*\in\cD$. Under the additional assumption that $\bF$ is the $\PP$-augmented natural filtration of $W$, it can also be shown that $\cD=\{1/\bar{X}^*\}$.
The fact that $1/\bar{X}^*\in\cD$ implies that the simple model described above satisfies the conditions of Theorem \ref{thm:HJM}, regardless of the specification of $\{S^i;i\in I\}$, as can be also verified by means of direct computations.

We proceed with the calculation of the forward rate processes. From \eqref{MMM-bond-0}, we obtain
\[
f^0(t,T) = -\frac{\partial}{\partial T} \ln(B(t,T)) = g(t,T) + m(t,T),
\]
where the short rate contribution to the forward rate is given by $g(t,T):=-\partial_T\ln(G(t,T))$ and the market price of risk contribution is given by $m(t,T):=-\partial_T\ln(M(t,T))$.
We refer to equation (13.3.9) in \cite{PH} for an explicit formula for $m(t,T)$. 
For every $i\in I$, we have from \eqref{MMM-bond-0} and \eqref{MMM-bond-i} that
%\[
%B^i(t,T) =  \frac{B^0(t,T) \, \EE [ S^i_T | \cF_t ]}{S^i_t} 
%= G(t,T)  M(t,T)  \frac{\EE [ S^i_T | \cF_t ]}{S^i_t},
%\]
%and, hence, we obtain the risky forward rate process
\[
f^i(t,T) = -\frac{\partial}{\partial T} \ln \bigg( \frac{B^0(t,T)  \EE [ S^i_T | \cF_t ]}{S^i_t}  \bigg) 
%= g(t,T) + m(t,T) + s^i(t,T) 
= f^0(t,T) + s^i(t,T),
\]
where $s^i(t,T):=-\partial_T\ln(\EE [ S^i_T | \cF_t ])$, for all $t\leq T$.

\begin{remark}\label{rem:deterministic}
Assumption \ref{ass-GOP} is in particular satisfied if the riskless short rate $r$ is deterministic and 
\[
Z^i = \int_0^{\cdot} a_s^i \ud s, 
\]
for a deterministic function $a^i : \R_+ \to \R$, for every $i\in I$. 
In this case, equation \eqref{MMM-bond-0} reduces to
\[
B^0(t,T) = \exp \bigg( -\int_t^T r_s \ud s \bigg) M(t,T)
\]
with $M(t,T)$ given as in \eqref{MMM-contribution}.
Moreover, we have that $\EE[S^i_T|\cF_t]=S^i_0\exp(\int_0^Ta^i_s\ud s)$, for every $i\in I$.
The forward rate processes are then given by
\[
f^0(t,T) = r_T + m(t,T)
\qquad\text{and}\qquad
f^i(t,T) = r_T + m(t,T) - a_T^i = f^0(t,T) - a_T^i.
\]
This simple example admits obvious conditions ensuring the monotonicity of the risky term structures. Indeed, monotonicity holds if $a^i \leq a^j$ and $S^i_0 \leq S^j_0$, for all $i,j\in I$ with $i < j$.
In particular, $a^i \leq a^j$ implies that $f^i(t,T) \geq f^j(t,T)$. As will become clear in Section \ref{sec:SPDE_invariant}, this case corresponds to the situation where the closed convex cone \eqref{cone} is invariant for the real-world HJMM SPDE \eqref{eq:SPDE}.  
\end{remark}

\begin{remark}
The example described in this section illustrates a generic way of constructing real-world models satisfying market viability, without requiring risk-neutral probabilities. One can start from a strictly positive local martingale $D$, a num\'eraire $X^0$ and a family $\{S^i:i\in I\}$ of strictly positive spot processes. Generalizing \eqref{MMM-bond-0} and \eqref{MMM-bond-i}, one can then specify riskless and risky bond prices as
\begin{equation}	\label{eq:model_construction}
B^i(t,T) := \EE\left[\frac{D_T}{D_t}\frac{X^0_t}{X^0_T}\frac{S^i_T}{S^i_t}\bigg|\cF_t\right],
\qquad\text{ for all }0\leq t\leq T<+\infty\text{ and }i\in I_0,
\end{equation}
where $S^0\equiv 1$. 
%If the right-hand side of \eqref{eq:model_construction} is differentiable in $T$, then one can derive the equations for the forward rates $f^i(t,T)$.
By construction, $D$ is a LMD and, therefore, the conditions of Theorem \ref{thm:HJM} are satisfied. Moreover, all term structures are fairly priced by $D$, in the sense of Definition \ref{def:fair}.  
\end{remark}

\section{The real-world HJMM SPDE}
\label{sec:SPDE}

In this section, we study existence and uniqueness of solutions to the stochastic partial differential equation (SPDE) arising in the general HJM model for multiple term structures developed in Section \ref{sec:HJM}. In Section \ref{sec-SPDE} we prove a general existence result for SPDEs with random locally Lipschitz coefficients. This result is new in the literature and of independent interest.
In Section \ref{sec-appl-rw-HJMM}, we rely on this result to prove existence and uniqueness of the solution to the real-world Heath-Jarrow-Morton-Musiela (HJMM) SPDE under suitable regularity conditions, extending in several directions the results of \cite{fitate2010} (see Remark \ref{rem:extend_fitate} below).
Section \ref{sec:SPDE_invariant} provides conditions ensuring that the SPDE generates ordered term structures, while Section \ref{sec:affine} addresses the issue of the existence of affine realizations.

\subsection{A general existence and uniqueness result for SPDEs}\label{sec-SPDE}

In this section, we establish a new existence and uniqueness theorem for SPDEs with locally Lipschitz coefficients, where the Lipschitz constants are allowed to be stochastic (Theorem \ref{thm-SPDE}). Our general strategy for the proof is as follows. First, making use of the theory of \cite{Metivier}, we establish an existence and uniqueness result for infinite-dimensional SDEs (Theorem \ref{thm-SDE}). Then, we rely on the method of the ``moving frame'', which has originally been introduced in \cite{FTT:10b}, in order to transfer this existence and uniqueness result to SPDEs in the framework of the semigroup approach (Theorem \ref{thm-SPDE}).

Let $(\Omega,\cF,\bF,\PP)$ be a filtered probability space as introduced in Section \ref{sec:HJM_setup} and supporting a $d$-dimensional Brownian motion $W$ and a homogeneous Poisson random measure $\mu(\ud t,\ud x)$ on $\R_+\times E$ with compensator $\nu(\ud t,\ud x)=F(\ud x)\ud t$, where $E$ is a locally compact space equipped with its Borel sigma-field $\cB(E)$ and $F$ is a sigma-finite measure on $(E,\cB(E))$.  The compensated random measure is denoted by $\tilde{\mu}(\ud t,\ud x) := \mu(\ud t,\ud x) - F(\ud x)\ud t$.

Let $\cH$ be a separable Hilbert space and denote by $L_2(\R^d,\cH)$ the space of  Hilbert-Schmidt operators from $\R^d$ to $\cH$. Let
\begin{align}\label{coeff-a}
&a: \R_+ \times \Omega \times \cH \to \cH,
\\ \label{coeff-b} &b: \R_+ \times \Omega \times \cH \to L_2(\R^d,\cH),
\\ \label{coeff-c} &c: \R_+ \times \Omega \times \cH \times E \to \cH
\end{align}
be mappings such that $a$ and $b$ are $\cP \otimes \cB(\cH)$-measurable and $c$ is $\cP \otimes \cB(\cH) \otimes \cB(E)$-measurable. 

We consider the $\cH$-valued SDE
\begin{equation}	\label{SDE}
dY_t = a(t,Y_t) \ud t + b(t,Y_t) \ud W_t + \int_E c(t,Y_{t-},x) \tilde{\mu}(\ud t,\ud x),
\qquad Y_0 = y_0.
\end{equation}
%\begin{align}
%\left\{
%\begin{array}{rcl}
%dY_t & = & a(t,Y_t) \ud t + b(t,Y_t) \ud W_t + \int_E c(t,Y_{t-},x) \tilde{\mu}(\ud t,\ud x),
%\\ Y_0 & = & y_0.
%\end{array}
%\right.
%\end{align}
Given an $\cF_0$-measurable random variable $y_0 : \Omega \to \cH$, an $\cH$-valued c\`{a}dl\`{a}g adapted process $Y = (Y_t)_{t \geq0}$ is called a \emph{strong solution} to the SDE \eqref{SDE} with $Y_0 = y_0$ if it holds a.s. for all $t\in\R_+$ that
\[
\int_0^t \bigg( \| a(s,Y_s) \|_{\cH} + \| b(s,Y_s) \|_{L_2(\R^d,\cH)}^2 + \int_E \| c(s,Y_{s-},x) \|_{\cH}^2 F(\ud x) \bigg) \ud s < +\infty
\]
and
\[
Y_t = y_0 + \int_0^t a(s,Y_s) \ud s + \int_0^t b(s,Y_s) \ud W_s + \int_0^t c(s,Y_{s-},x) \tilde{\mu}(\ud s,\ud x).
\]

\begin{theorem}	\label{thm-SDE}
Suppose that, for each $r \in \R_+$, there exists an optional locally bounded non-negative process $L^{r}$ such that, for all $(t,\omega) \in \R_+ \times \Omega$ and all $y,z \in \cH$ with $\| y \|_{\cH}\vee \| z \|_{\cH} \leq r$, we have
\begin{align}\label{Lip-a}
\| a(t,\omega,y) - a(t,\omega,z) \|_{\cH}^2 &\leq L_t^{r}(\omega) \| y-z \|_{\cH}^2,
\\ \label{Lip-b} \| b(t,\omega,y) - b(t,\omega,z) \|_{L_2(\R^d,\cH)}^2 &\leq L_t^{r}(\omega) \| y-z \|_{\cH}^2,
\\ \label{Lip-c} \int_E \| c(t,\omega,y,x) - c(t,\omega,z,x) \|_{\cH}^2 F(\ud x) &\leq L_t^{r}(\omega) \| y-z \|_{\cH}^2.
\end{align}
Suppose in addition that there exists an optional locally bounded non-negative process $L$ such that, for all $(t,\omega,y) \in \R_+ \times \Omega \times \cH$, we have
\begin{align}\label{LG-a}
\| a(t,\omega,y) \|_{\cH}^2 &\leq L_t(\omega) ( 1 + \| y \|_{\cH}^2 ),
\\ \label{LG-b} \| b(t,\omega,y) \|_{L_2(\R^d,\cH)}^2 &\leq L_t(\omega) ( 1 + \| y \|_{\cH}^2 ),
\\ \label{LG-c} \int_E \| c(t,\omega,y,x) \|_{\cH}^2 F(\ud x) &\leq L_t(\omega) ( 1 + \| y \|_{\cH}^2 ).
\end{align}
Then, for every $\cF_0$-measurable random variable $y_0 : \Omega \to \cH$, there exists a unique strong solution $Y$ to the SDE \eqref{SDE} with $Y_0 = y_0$.
\end{theorem}
\begin{proof}
Let $\overline{W}$ be the $\R^{d+1}$-valued semimartingale given by $\overline{W}_t := (t,W_t)$, for all $t \in \R_+$. Let $r \in \R_+$ be arbitrary and let $A$ be a control process for $\overline{W}$ in the sense of \cite[Definition 23.13]{Metivier}, which exists by \cite[Theorem 23.11]{Metivier}. Since $L^{r}$ is locally bounded, we can define the process $\bar{L}^{r}$ by $\bar{L}^r_t := 3 \int_0^t L_s^{r} \ud A_s$, for all $t \in \R_+$. Clearly, $\bar{L}^{r}$ is c\`{a}dl\`{a}g, increasing and non-negative. Let $\lambda$ and $\beta$ be the functional and the increasing process associated to the compensated Poisson random measure $\tilde{\mu}$ as in \cite[Theorem 31.9]{Metivier}. 
For an $\cH$-valued c\`{a}dl\`{a}g adapted process $Y$, we will use the abbreviations $a(t,Y)$, $b(t,Y)$ and $c(t,Y,x)$ for $a(t,\omega,Y_{t-}(\omega))$, $b(t,\omega,Y_{t-}(\omega))$ and $c(t,\omega,Y_{t-}(\omega),x)$, respectively. Moreover, we will use the shorthand notation $\bar{c}(t,Y)$ for the map $x \mapsto c(t,\omega,Y_{t-}(\omega),x)$. According to the Remark on \cite[page 245]{Metivier}, for all $\cH$-valued c\`{a}dl\`{a}g adapted processes $Y,Z$ and all $t \in \R_+$, it holds that
\[
\int_0^t \lambda_s ( \bar{c}(s,Y) - \bar{c}(s,Z) ) \ud \beta(s) 
= \int_0^t \int_E \| c(s,Y,x) - c(s,Z,x) \|^2_{\cH} F(\ud x) \ud s.
\]
Furthermore, for every Hilbert-Schmidt operator $T \in L_2(\R^d,\cH)$ we have that $T \in L(\R^d,\cH)$ and
\[
\| T \|_{L(\R^d,\cH)} \leq \| T \|_{L_2(\R^d,\cH)}.
\]
Hence, using \eqref{Lip-a}--\eqref{Lip-c}, for all $\cH$-valued c\`{a}dl\`{a}g adapted processes $Y,Z$ and $t \in \R_+$, we obtain
\begin{align*}
&\int_0^t \big( \| a(s,Y) - a(s,Z) \|_{\cH}^2 + \| b(s,Y) - b(s,Z) \|_{L(\R^d,\cH)}^2 \big) \ud A_s + \int_0^t \lambda_s ( \bar{c}(s,Y) - \bar{c}(s,Z) ) \ud \beta(s)
\\ &\leq 3 \int_0^t L_s^r \| Y_{s-} - Z_{s-} \|_{\cH}^2 \ud A_s = \int_0^t \| Y_{s-} - Z_{s-} \|_{\cH}^2 \ud \bar{L}_s^{r} \leq \int_0^t \sup_{u < s} \| Y_u - Z_u \|_{\cH}^2 \ud \bar{L}_s^{r}
\end{align*}
on the set $\{ \sup_{s < t} ( \| Y_s \|_{\cH} \vee \| Z_s \|_{\cH} ) \leq r \}$. 
Similarly, we define the process $\bar{L} := 3 \int_0^{\cdot} L_s \ud A_s$ and using \eqref{LG-a}--\eqref{LG-c}, for every $\cH$-valued c\`{a}dl\`{a}g adapted process $Y$ and $t \in \R_+$, we obtain that
\begin{align*}
&\int_0^t \big( \| a(s,Y) \|_{\cH}^2 + \| b(s,Y) \|_{L(\R^d,\cH)}^2 \big) \ud A_s + \int_0^t \lambda_s ( \bar{c}(s,Y) ) \ud \beta(s)
\\ &\leq 3 \int_0^t L_s ( 1 + \| Y_{s-} \|_{\cH}^2 ) \ud A_s = \int_0^t ( 1 + \| Y_{s-} \|_{\cH}^2 ) \ud \bar{L}_s \leq \int_0^t \Big( 1 + \sup_{u < s} \| Y_u \|_{\cH}^2 \Big) \ud \bar{L}_s.
\end{align*}
By \cite[Theorems 34.7 and 35.2]{Metivier}, for each $\cF_0$-measurable random variable $y_0 : \Omega \to \cH$ there exists a unique strong solution $Y$ to the SDE \eqref{SDE} with $Y_0 = y_0$, where the stochastic integral with respect to the Brownian motion process $W$ is understood in the sense of \cite[Chapter 26]{Metivier}.

In order to complete the proof, it remains to prove that the stochastic integral with respect to $W$ coincides with the isometric stochastic integral. To this effect, note that the process $A$ given by $A_t := 2(t+1)$, for all $t \in \R_+$, is a control process of $W$ in the sense of \cite[Definition 23.13]{Metivier}. Indeed, let $G$ be a separable Hilbert space, $\Phi$ an elementary $L(\R^d,G)$-valued predictable process and  $\tau$ a stopping time. Then, by Doob's $L^2$-inequality and the It\^{o} isometry, we have that
\begin{align*}
\EE \Bigg[ \sup_{t < \tau} \bigg\| \int_0^t \Phi_s \ud W_s \bigg\|_G^2 \Bigg] \leq 4 \, \EE \bigg[ \int_0^{\tau} \| \Phi_s \|_{L_2(\R^d,G)}^2 \ud s \bigg] \leq \EE \bigg[ A_{\tau} \int_0^{\tau} \| \Phi_s \|_{L_2(\R^d,G)}^2 \ud A_s \bigg].
\end{align*}
Furthermore, by the linear growth condition \eqref{LG-b}, for every $\cH$-valued c\`{a}dl\`{a}g adapted process $Y$ we have that
\[
\int_0^t \| b(s,Y) \|_{L_2(\R^d,\cH)}^2 \ud s < +\infty\text{ a.s., for all } t \in \R_+.
\]
In view of \cite[Remark 26.2]{Metivier}, this suffices to complete the proof.
\end{proof}

We shall now make use of Theorem \ref{thm-SDE} in order to prove a general existence and uniqueness result for SPDEs with random locally Lipschitz coefficients. 
Let $H$ be a separable Hilbert space and $A$ the generator of a $C_0$-semigroup $(\cS_t)_{t \geq 0}$ on $H$. Furthermore, let
\begin{align*}
&\alpha: \R_+ \times \Omega \times H \to H,
\\ &\beta: \R_+ \times \Omega \times H \to L_2(\R^d,H),
\\ &\gamma: \R_+ \times \Omega \times H \times E \to H
\end{align*}
be mappings such that $\alpha$ and $\beta$ are $\cP \otimes \cB(H)$-measurable and $\gamma$ is $\cP \otimes \cB(H) \otimes \cB(E)$-measurable. 
We consider the following $H$-valued SPDE:
\begin{equation}\label{SPDE}
dX_t = \bigl( A X_t + \alpha(t,X_t) \bigr) \ud t + \beta(t,X_t) \ud W_t + \int_E \gamma(t,X_{t-},x) \tilde{\mu}(\ud t,\ud x),
\qquad X_0 =  x_0.
\end{equation}
%\begin{align}\label{SPDE}
%\left\{
%\begin{array}{rcl}
%dX_t & = & ( A X_t + \alpha(t,X_t) ) \ud t + \beta(t,X_t) \ud W_t + \int_E \gamma(t,X_{t-},x) \tilde{\mu}(\ud t,\ud x),
%\\ X_0 & = & x_0.
%\end{array}
%\right.
%\end{align}
Given an $\cF_0$-measurable random variable $x_0 : \Omega \to H$, an $H$-valued c\`{a}dl\`{a}g adapted process $X = (X_t)_{t \geq0}$ is called a \emph{mild solution} to the SPDE \eqref{SPDE} with $X_0 = x_0$ if it holds a.s. for all $t\in\R_+$ that
\[
\int_0^t \bigg( \| \alpha(s,X_s) \|_H + \| \beta(s,X_s) \|_{L_2(\R^d,H)}^2 + \int_E \| \gamma(s,X_{s-},x) \|_{H}^2 F(\ud x) \bigg) \ud s < +\infty,
\]
and
\[
X_t = \cS_t x_0 + \int_0^t \cS_{t-s} \alpha(s,X_s) \ud s + \int_0^t \cS_{t-s} \beta(s,X_s) \ud W_s
 + \int_0^t \cS_{t-s} \gamma(s,X_{s-},x) \tilde{\mu}(\ud s,\ud x).
\]

\begin{theorem}\label{thm-SPDE}
Suppose that the semigroup $(\cS_t)_{t \geq 0}$ is pseudo-contractive, i.e., there exists a constant $\eta \geq 0$ such that
\begin{align}\label{pseudo-contr}
\| \cS_t \| \leq e^{\eta t}, 
\quad \text{for all $t \geq 0$.}
\end{align}
Suppose in addition that, for each $r \in \R_+$, there exists an optional locally bounded non-negative process $L^{r}$ such that, for all $(t,\omega) \in \R_+ \times \Omega$ and $h,g \in H$ with $\| h \|_{H}\vee \| g \|_{H} \leq r$,  we have
\begin{align}\label{Lip-1}
\| \alpha(t,\omega,h) - \alpha(t,\omega,g) \|_{H}^2 &\leq L_t^{r}(\omega) \| h-g \|_{H}^2,
\\ \label{Lip-2} \| \beta(t,\omega,h) - \beta(t,\omega,g) \|_{L_2(\R^d,H)}^2 &\leq L_t^{r}(\omega) \| h-g \|_{H}^2,
\\ \label{Lip-3} \int_E \| \gamma(t,\omega,h,x) - \gamma(t,\omega,g,x) \|_{H}^2 F(\ud x) &\leq L_t^{r}(\omega) \| h-g \|_{H}^2.
\end{align}
Suppose in addition that there exists an optional locally bounded non-negative process $L$ such that, for all $(t,\omega,h) \in \R_+ \times \Omega \times H$, we have
\begin{align}\label{LG-1}
\| \alpha(t,\omega,h) \|_{H}^2 &\leq L_t(\omega) ( 1 + \| h \|_{H}^2 ),
\\ \label{LG-2} \| \beta(t,\omega,h) \|_{L_2(\R^d,H)}^2 &\leq L_t(\omega) ( 1 + \| h \|_{H}^2 ),
\\ \label{LG-3} \int_E \| \gamma(t,\omega,h,x) \|_{H}^2 F(\ud x) &\leq L_t(\omega) ( 1 + \| h \|_{H}^2 ).
\end{align}
Then, for every $\cF_0$-measurable random variable $x_0 : \Omega \to H$, there exists a unique mild solution $X$ to the SPDE \eqref{SPDE} with $X_0 = x_0$.
\end{theorem}

The result of Theorem \ref{thm-SPDE} will follow directly from Propositions \ref{prop-SDE-trans} and \ref{prop-X-Y} below. As a preliminary, we state the following lemma, which is a consequence of \cite[Proposition 8.7]{FTT:10b} and its proof.

\begin{lemma}\label{prop-diagram}
Suppose that there exists a constant $\eta \geq 0$ such that \eqref{pseudo-contr} holds. Then there exist another separable Hilbert space $\cH$, a $C_0$-group $(\mathcal{U}_t)_{t \in \mathbb{R}}$ on $\cH$ satisfying
\begin{align}\label{pseudo-contr-U}
\| \mathcal{U}_t \| = e^{\eta t}, \quad \text{for all $t \in \R$,}
\end{align}
and an isometric embedding $\ell \in L(H,\cH)$ such that
\[
\pi\, \mathcal{U}_t \ell = \cS_t, \quad \text{for all $t \in \mathbb{R}_+$,}
\]
where $\pi := \ell^*$ denotes the orthogonal projection from $\cH$ onto $H$.
\end{lemma}

We now consider the $\cH$-valued SDE \eqref{SDE} with coefficients \eqref{coeff-a}--\eqref{coeff-c} given by
\begin{align}\label{a-def}
a(t,\omega,y) &:= \cU_{-t} \ell \alpha(t,\omega,\pi \cU_t y),
\\ \label{b-def} b(t,\omega,y) &:= \cU_{-t} \ell \beta(t,\omega,\pi \cU_t y),
\\ \label{c-def} c(t,\omega,y,x) &:= \cU_{-t} \ell \gamma(t,\omega,\pi \cU_t y,x).
\end{align}
Note that $a$ and $b$ are $\cP \otimes \cB(\cH)$-measurable and $c$ is $\cP \otimes \cB(\cH) \otimes \cB(E)$-measurable. This follows because the projection $(t,\omega) \mapsto t$ is $\cP / \cB(\R_+)$-measurable and the mapping $(t,y) \mapsto \cU_t y$ is continuous.

\begin{proposition}\label{prop-SDE-trans}
Suppose that the assumptions of Theorem \ref{thm-SPDE} are satisfied. Then, for every $\cF_0$-measurable random variable $y_0 : \Omega \to \cH$, there exists a unique strong solution $Y$ to the SDE \eqref{SDE} with coefficients given as in \eqref{a-def}--\eqref{c-def} and with $Y_0 = y_0$.
\end{proposition}
\begin{proof}
It suffices to prove that, for each $T > 0$, there exists a unique strong solution $Y^T = (Y_t^T)_{t \in [0,T]}$ to the SDE \eqref{SDE} with $Y_0^T = y_0$. Indeed, in this case the process
\[
Y := y_0 \bbI_{\{ 0 \} \times \Omega} + \sum_{N=1}^{+\infty} Y^N \bbI_{(N-1,N] \times \Omega}
\]
is the unique strong solution to the SDE \eqref{SDE} with $Y_0 = y_0$.
For $T > 0$ and $r \in \R_+$ arbitrary, let us define the process $\bar{L}^{r} := L^{\exp(\eta T) r}$, where the constant $\eta \geq 0$ stems from \eqref{pseudo-contr}.
Let $(t,\omega) \in [0,T] \times \Omega$ and $y,z \in \cH$ with $\| y \|_{\cH}\vee \| z \|_{\cH} \leq r$ be arbitrary. Then, by \eqref{pseudo-contr-U} we have that
\[
\| \pi \cU_t y \|_H\vee \| \pi \cU_t z \|_H \leq e^{\eta T} r.
\]
Therefore, taking into account \eqref{a-def}, \eqref{pseudo-contr-U} and \eqref{Lip-1}, we obtain
\begin{align*}
\| a(t,\omega,y) - a(t,\omega,z) \|_{\cH}^2 &= e^{-2\eta t} \| \alpha(t,\omega,\pi \cU_t y) - \alpha(t,\omega,\pi \cU_t z) \|_H^2
\\ &\leq e^{-2\eta t} \bar{L}_t^{r}(\omega) \| \pi \cU_t y - \pi \cU_t z \|_H^2 \leq \bar{L}_t^{r}(\omega) \| y-z \|_{\cH}^2,
\end{align*}
showing that condition \eqref{Lip-a} is satisfied with $L^{r}$ replaced by $\bar{L}^{r}$. Similarly, we can show that conditions \eqref{Lip-b} and \eqref{Lip-c} are satisfied with $L^{r}$ replaced by $\bar{L}^{r}$. Now, let $(t,\omega,y) \in [0,T] \times \Omega \times \cH$ be arbitrary.
Taking into account \eqref{a-def}, \eqref{pseudo-contr-U} and \eqref{LG-1}, we obtain
\begin{align*}
\| a(t,\omega,y) \|_{\cH}^2 &= e^{-2\eta t} \| \alpha(t,\omega,\pi \cU_t y) \|_H^2 \leq e^{-2\eta t} L_t(\omega) (1 + \| \pi \cU_t y \|_H^2 ) 
\\ &\leq e^{-2\eta t} L_t(\omega) (1 + e^{2 \eta t} \| y \|_{\cH}^2 ) \leq L_t(\omega) (1 + \| y \|_{\cH}^2),
\end{align*}
showing that condition \eqref{LG-a} is satisfied. In an analogous way we can show that conditions \eqref{LG-b} and \eqref{LG-c} are satisfied as well. The claim then follows directly by an application of Theorem  \ref{thm-SDE}.
\end{proof}

\begin{proposition}\label{prop-X-Y}
Suppose that, for every $\cF_0$-measurable random variable $y_0 : \Omega \to \cH$, there exists a unique strong solution $Y$ to the SDE \eqref{SDE} with coefficients given as in \eqref{a-def}--\eqref{c-def} and with $Y_0 = y_0$. Then, for every $\cF_0$-measurable random variable $x_0 : \Omega \to H$, there exists a unique mild solution $X$ to the SPDE \eqref{SPDE} with $X_0 = x_0$.
\end{proposition}
\begin{proof}
Corollaries 3.9 and 3.11 from \cite{Tappe-YW} also hold true in the present setup with an additional Poisson random measure, with identical proofs. Combining these two results completes the proof.
\end{proof}

At this stage, the proof of Theorem \ref{thm-SPDE} is a direct consequence of Propositions \ref{prop-SDE-trans} and \ref{prop-X-Y}. 

We close this section with an auxiliary result about the space of Hilbert-Schmidt operators, which will be useful later. For $k \in \N$, we denote by $H^k$ the $k$-fold cartesian product $H \times \ldots \times H$, which is a separable Hilbert space when endowed with the norm
\[
\| h \|_{H^k} := \bigg( \sum_{i=1}^k \| h^i \|_H^2 \bigg)^{1/2}, \quad h \in H^k.
\]
Similarly, the space $H^{k \times d}$, endowed with the Frobenius norm
\[
\| h \|_{H^{k \times d}} := \bigg( \sum_{i=1}^k \sum_{j=1}^d \| h^{ij} \|_H^2 \bigg)^{1/2}, \quad h \in H^{k \times d},
\]
is also a separable Hilbert space. Moreover, we recall that the space $L_2(\R^d,H^k)$ of Hilbert-Schmidt operators, equipped with the Hilbert-Schmidt norm
\[
\| T \|_{L_2(\R^d,H^k)} := \bigg( \sum_{j=1}^d \| T e_j \|_{H^k}^2 \bigg)^{1/2},
\quad T\in L_2(\R^d,H^k),
\]
is also a separable Hilbert space.

\begin{lemma}\label{lemma-HS}
$L_2(\R^d,H^k)\cong H^{k \times d}$ (i.e., $L_2(\R^d,H^k)$ and $H^{k \times d}$ are isometrically isomorphic). 
\end{lemma}
\begin{proof}
For each $h \in H^{k \times d}$, we assign $T(h) \in L_2(\R^d,H^k)$ as
\[
T(h) e_j := (h^{ij})_{i=1,\ldots,k}, 
\qquad \text{for $j=1,\ldots,d$}, 
\]
which provides an isometric isomorphism $T : H^{k \times d} \to L_2(\R^d,H^k)$.
\end{proof}

\subsection{Well-posedness of the real-world HJMM SPDE}
\label{sec-appl-rw-HJMM}

In this section, we establish a general existence and uniqueness result for the SPDE arising in the HJM framework developed in Section \ref{sec:HJM}, under the assumption that there exists an LMD (i.e., $\cD\neq\emptyset$) and making use of Theorem \ref{thm-SPDE}.

We adopt the Musiela parametrization (see, e.g., \cite{F:01}) for the instantaneous forward rates and write 
\[
\eta^i_t(\xi):=f^i(t,t+\xi),
\qquad\text{ for all $(t,\xi)\in\R^2_+$ and $i\in I_0$}. 
\]
From now on, we assume that the set $I_0$ is finite and identify it by $\{0,1,\ldots,m\}$, for some $m\in\N$.
The evolution of the $(m+1)$-dimensional family of forward curves is described by the process $\eta=(\eta^0,\eta^1,\ldots,\eta^m)$, taking values in a suitable Hilbert space $H$ of functions $h:\R_+\rightarrow\R^{m+1}$ that will be specified below.
According to this parameterization and assuming continuity of the map $T\mapsto f^i(t,T)$, for all $i=0,1,\ldots,m$, equation \eqref{eq:fwd_rate} can be written in vector form as follows:
\be	\label{eq:shift}
\eta_t(\xi) = \cS_t\eta_0(\xi) + \int_0^t\cS_{t-s}\alpha(s,s+\xi)\ud s + \int_0^t\cS_{t-s}\beta(s,s+\xi)\ud W_s + \int_0^t\int_E\cS_{t-s}\gamma(s,s+\xi,x)\tilde{\mu}(\ud s,\ud x),
\ee
for all $t\geq0$, where $(\cS_t)_{t\geq0}$ is the shift semigroup acting on the second time argument of $\alpha$, $\beta$ and $\gamma$, 
so that $\cS_t \alpha(s,\cdot)=\alpha(s,t+\cdot)$, and similarly for the other terms appearing in \eqref{eq:shift}.
As in Section \ref{sec-SPDE}, we work under the standing assumption that $\mu(\ud t,\ud x)$ is a homogeneous Poisson random measure with compensator $\nu(\ud t,\ud x)=F(\ud x)\ud t$ and $E$  is a locally compact space.

We consider diffusive and jump volatilities $\beta$ and $\gamma$ with the following structure:
\be	\label{eq:volatilities}
\beta(t,t+\xi) = \beta(\eta_{t-})(\xi)
\qquad\text{and}\qquad
\gamma(t,t+\xi,x) = \gamma(\eta_{t-},x)(\xi),
\ee
for all $(t,\xi)\in\R^2_+$ and $x\in E$, where the functions $\beta$ and $\gamma$ on the right-hand sides in \eqref{eq:volatilities} will satisfy the requirements of Assumption \ref{ass:SPDE} below.
For each $i=0,1,\ldots,m$, the $i$-th components $\beta^i(\eta_{t-})(\xi)$ and $\gamma^i(\eta_{t-},x)(\xi)$ are respectively a $d$-dimensional vector and a scalar, for all $t>0$, $\xi\in\R_+$ and $x\in E$, and we write $\beta^{i,j}(\eta_{t-})(\xi)$ for the $j$-th component of the vector $\beta^i(\eta_{t-})(\xi)$, for $j=1,\ldots,d$.

\begin{remark}
The structure \eqref{eq:volatilities} allows for state-dependent volatilities, meaning that the volatility of each forward rate can depend on the whole family $\eta_{t-}=(\eta^0_{t-},\eta^1_{t-},\ldots,\eta^m_{t-})$ of forward curves. It is possible to consider volatilities that depend on additional sources of randomness beyond the forward curves, in particular on the family $\{S^i_{t-};i\in I\}$ of spot processes. The results of this section can be extended to this more general setting with no changes in the proofs (see also Remark \ref{rem:random_drift} below).
\end{remark}

In this setting, an $H$-valued stochastic process $\eta$ satisfying \eqref{eq:shift} with volatilities as in \eqref{eq:volatilities} is a {\em mild solution} to the following Heath-Jarrow-Morton-Musiela (HJMM) SPDE:
\be	\label{eq:SPDE}
\ud\eta_t = \left(\frac{\partial}{\partial\xi}\eta_t + \alpha(t,\eta_{t-})\right)\ud t + \beta(\eta_{t-})\ud W_t + \int_E\gamma(\eta_{t-},x)\tilde{\mu}(\ud t,\ud x).
\ee
We refer to \eqref{eq:SPDE} as the {\em real-world} HJMM SPDE, since it is formulated under the real-world probability $\PP$ (in contrast, standard formulations of the HJMM SPDE are under a risk-neutral probability $\QQ$, see for instance \cite{fitate2010}).
By Theorem \ref{thm:HJM}, the existence of LMDs (which in turn ensures market viability) corresponds to the existence of $\lambda\in L^2_{\rm loc}(W)$ and $\psi\in G_{\rm loc}(\mu)$ with $\psi>-1$ such that the drift term 
\[
\alpha(t,\eta_{t-})=(\alpha^0(t,\eta_{t-}),\alpha^1(t,\eta_{t-}),\ldots,\alpha^m(t,\eta_{t-}))
\] 
has the following structure (see Remark \ref{rem:HJM}):
\be	\label{eq:drift_HJM}	\begin{aligned}
\alpha^i(t,\eta_{t-})
&= \beta^i(\eta_{t-})^{\top}\bar{\beta}^i(\eta_{t-})
- (\lambda_t+b^i_t)^{\top}\beta^i(\eta_{t-})	\\
&\quad+ \int_E\gamma^i(\eta_{t-},x)\left(1-e^{-\bar{\gamma}^i(\eta_{t-},x)}\bigl(1+\psi_t(x)\bigr)\bigl(1+c^i_t(x)\bigr)\right)F(\ud x),
\end{aligned}	\ee
for each $i=0,1,\ldots,m$, where we set $b^0\equiv 0$ and $c^0\equiv0$ and make use of the notation 
\[
\bar{\beta}^i(\eta_{t-})(\cdot):=\int_0^{\cdot}\beta^i(\eta_{t-})(u)\ud u
\qquad\text{ and }\qquad
\bar{\gamma}^i(\eta_{t-},x)(\cdot):=\int_0^{\cdot}\gamma^i(\eta_{t-},x)(u)\ud u.
\]

\begin{remark}	\label{rem:random_drift}
Observe that, for each $i=0,1,\ldots,m$, the drift term $\alpha^i(t,\eta_{t-})$ depends on $\lambda_t$ and $\psi_t$ as well as on the processes $b^i_t$ and $c^i_t$ which are associated to the $i$-th spot process $S^i_t$. This implies that the drift term is not entirely determined by the forward curves $\eta_{t-}$, rather it depends on additional sources of randomness, unlike in the classical setting of the HJMM SPDE formulated under a risk-neutral measure with a single term structure, as considered in \cite{fitate2010}. 
This explains why, in the context of the real-world HJM framework of Section \ref{sec:HJM}, we are naturally led to consider SPDEs with random locally Lipschitz coefficients (see Theorem \ref{thm-SPDE}), which are not covered by the existing theory.
\end{remark}

To proceed, we need to define the space of functions on which we shall study the SPDE \eqref{eq:SPDE}. To this effect, we fix an arbitrary constant $\rho>0$ and denote by $H_{\rho}^k$, for $k\in\N$, the space of all absolutely continuous functions $h:\R_+\rightarrow\R^k$ such that
\be	\label{eq:norm}
\|h\|_{\rho,k} := \left(|h(0)|_k^2+\int_{\R_+}|h'(s)|_k^2e^{\rho s}\ud s\right)^{1/2} < +\infty,
\ee
where $|\cdot|_k$ denotes the Euclidean norm in $\R^k$.  Moreover, we define
\[
H^{0,k}_{\rho} := \{h\in H^k_{\rho} : |h(\infty)|_k=0\}.
\]
The space $H^k_{\rho}$ has been already considered in \cite[Section 4.1]{CFG:16} and represents an extension to the multi-dimensional setting of the so-called Filipovi\'c space first introduced in \cite{F:01} in the one-dimensional case. We collect in Appendix \ref{app-space} several technical properties of the space $H_{\rho} = H_{\rho}^1$.

\begin{remark}\label{rem-product-space}
Note that the space $H^k_{\rho}$ corresponds to the $k$-fold cartesian product $H_{\rho} \times \ldots \times H_{\rho}$ and the norm defined in \eqref{eq:norm} can be expressed as $\| h \|_{\rho,k} = ( \sum_{i=1}^k \| h^i \|_{\rho}^2 )^{1/2}$, for  $h \in H_{\rho}^{k}$.
\end{remark}

The main goal of this section is to establish existence and uniqueness of a  mild solution to the real-world HJMM SPDE \eqref{eq:SPDE} with drift \eqref{eq:drift_HJM} on the space $H:=H^{m+1}_{\rho}$. 
We are interested in {\em global solutions}, i.e., $H$-valued stochastic processes that solve \eqref{eq:SPDE} on arbitrary time intervals. 

We now introduce a set of assumptions which will ensure existence and uniqueness of a global mild solution to the HJMM SPDE \eqref{eq:SPDE}. As in Appendix \ref{app-space}, for each $\rho'>\rho$ we introduce the constant
\[
K_{\rho,\rho'} := \bigg( 1 + \frac{1}{\sqrt{\rho}} \bigg) \sqrt{\frac{1}{\rho'(\rho' - \rho)}}.
\]
For any constant $K > 0$, we denote by $W_K : \R_+ \to \R_+$ the inverse of the strictly increasing function
\begin{align}\label{V_K_def}
V_K : \R_+ \to \R_+, \quad V_K(r) := r (1+r) \exp(Kr).
\end{align}
%Note that $V_K(r) \geq r$ for all $r \in \R_+$, which implies $W_K(r) \leq r$ for all $r \in \R_+$.
It can be easily seen that $W_K(r) \leq r$, for all $r \in \R_+$.
As usual, for $p\in(0,+\infty)$, we denote by $L^p(F)$ the space of all measurable functions $f:E\to\R$ such that $\int_E|f(x)|^pF(\ud x)<+\infty$.

\begin{assumption}	\label{ass:SPDE}
There exist $\rho'>\rho$, an optional locally bounded non-negative process $\Lambda$, a non-negative function $\kappa \in L^1(F) \cap L^2(F) \cap L^3(F)$, a constant $M_{\beta} \in \R_+$, an increasing function $M_{\gamma} : \R_+ \to \R_+$ and a constant $N_{\gamma} \in \R_+$ such that the following hold:
\begin{enumerate}[(i)]
\item
the processes $\lambda : \Omega\times\R_+  \to \R^d$ and $b :\Omega\times \R_+ \to \R^{(m+1) \times d}$ are predictable and locally bounded;
\item
the functions $\psi : \Omega\times\R_+  \times E \to (-1,+\infty)$ and $c : \Omega\times\R_+  \times E \to (-1,+\infty)^{m+1}$ are $\cP \otimes \cB(E)$-measurable and, for all $i=0,1,\ldots,m$, it holds that
\begin{align}\label{est-psi-c}
\big| \big( 1+\psi_t(\omega,x) \big) \big( 1 + c_t^i(\omega,x) \big) \big| \leq \Lambda_t(\omega) \kappa(x), 
\qquad \text{ for all }(\omega,t,x) \in \Omega\times \R_+ \times E;
\end{align}
\item 
the function $\beta : H_{\rho}^{m+1} \to H_{\rho}^{(m+1) \times d}$ satisfies
\begin{align}\label{beta-loc-Lip}
\beta &\in \Lip^{\loc} \big( H_{\rho}^{m+1},H_{\rho}^{0,(m+1) \times d} \big),
\\ \label{LG-beta-i} \|\beta(h)\|_{\rho,(m+1) \times d} &\leq M_{\beta} \sqrt{1+\|h\|_{\rho,m+1}}, 
\qquad\text{ for all } h\in H^{m+1}_{\rho};
\end{align}
\item the function $\gamma : H_{\rho}^{m+1} \times E \to H_{\rho}^{m+1}$ is $\cB(H_{\rho}^{m+1}) \otimes \cB(E)$-measurable and, for each $r \in \R_+$ and $h,g \in H_{\rho}^{m+1}$ with $\| h \|_{\rho,m+1}\vee \| g \|_{\rho,m+1} \leq r$, it holds that
\begin{align}\label{int-E-2-ass}
\| \gamma(h,x) - \gamma(g,x) \|_{\rho',m+1} &\leq \kappa(x) M_{\gamma}(r) \| h - g \|_{\rho,m+1},
 \qquad\text{for all $x \in E$},
\end{align}
and
\begin{align}\label{int-E-3-ass}
\| \gamma(h,x) \|_{\rho',m+1} &\leq W_{K_{\rho,\rho'}} \big( N_{\gamma} \kappa(x)(1+\| h \|_{\rho,m+1}) \big),
\qquad \text{ for all } (h,x)\in H_{\rho}^{m+1} \times E.
\end{align}
\end{enumerate}
\end{assumption}

\begin{remark}
Note that condition \eqref{int-E-2-ass} implies that
\begin{align*}
\gamma(\cdot,x) \in \Lip^{\loc}(H^{m+1}_{\rho},H^{0,m+1}_{\rho'}),
 \qquad\text{for all $x \in E$.}
\end{align*}
\end{remark}

\begin{remark}	\label{rem:extend_fitate}
Beyond the extension to a multi-dimensional setup, Assumption \ref{ass:SPDE} significantly weakens the requirements of \cite{fitate2010} by replacing global Lipschitz continuity and the boundedness condition imposed in \cite[Assumption 3.1]{fitate2010} with local Lipschitz continuity and a growth condition, respectively.
In particular, we explicitly allow for unbounded coefficients $\beta$ and $\gamma$. Example \ref{example:unbounded} provides a model with unbounded $\beta$ and $\gamma$ for which Assumption \ref{ass:SPDE} is fulfilled (and, therefore, Theorem \ref{thm-HJMM-SPDE} applies). This example covers the classical risk-neutral setup and thus demonstrates that our results strictly extend those of \cite{fitate2010}, even within the traditional risk-neutral setup.
Besides its mathematical interest, this generalization is essential to accommodate models that may fail to admit a risk-neutral probability. Indeed, imposing the requirements of \cite{fitate2010} under the real-world measure would typically translate into restrictive assumptions on the market prices of risk, often forcing the true martingale property of an LMD and, hence, the existence of a risk-neutral measure. Our framework avoids this structural limitation while still covering the classical risk-neutral setup as a special instance.
\end{remark}

% we can also mention that \cite{fitate2010} claim that their requirements cannot be substantially weakened

As shown in the following proposition, Assumption \ref{ass:SPDE} suffices to ensure that the drift term \eqref{eq:drift_HJM} of SPDE \eqref{eq:SPDE} satisfies a linear growth condition and is locally Lipschitz. This represents a crucial step towards the applicability of the general existence and uniqueness result of Theorem \ref{thm-SPDE} to the real-world HJMM SPDE \eqref{eq:SPDE}.

\begin{proposition}	\label{prop:alpha}
Suppose that Assumption \ref{ass:SPDE} holds.
Then the following holds:
\begin{enumerate}
\item the function $\alpha$ takes values in $H^{0,m+1}_{\rho}$ and is $\cP \otimes \cB(H^{m+1}_{\rho})$-measurable;
\item for each $r\in\R_+$, there exist an optional locally bounded non-negative process $L^r$ and a constant $K(r) \in \R_+$ such that, for all $(\omega,t) \in \Omega\times\R_+ $ and $h,g\in H^{m+1}_{\rho}$ with $\|h\|_{\rho,m+1}\vee \|g\|_{\rho,m+1}\leq r$, we have
\begin{align}\label{eq:alpha_Lip}
\|\alpha(t,\omega,h)-\alpha(t,\omega,g)\|_{\rho,m+1}^2 &\leq L_t^r(\omega) \|h-g\|_{\rho,m+1}^2,
\\ \label{eq:beta_Lip} \| \beta(h) - \beta(g) \|_{L_2(\R^d,H_{\rho}^{m+1})}^2 &\leq K(r) \|h-g\|_{\rho,m+1}^2,
\\ \label{eq:gamma_Lip} \int_E \| \gamma(h,x) - \gamma(g,x) \|_{\rho,m+1}^2 F(\ud x) &\leq K(r) \|h-g\|_{\rho,m+1}^2;
\end{align}
\item there exist an optional locally bounded non-negative process $L$ and a constant $K \in \R_+$ such that, for all $(\omega,t,h) \in \Omega\times\R_+ \times H$, we have
\begin{align}\label{eq:alpha_LG}
\|\alpha(t,\omega,h)\|_{\rho,m+1}^2 &\leq L_t(\omega)(1+\|h\|_{\rho,m+1}^2),
\\ \label{eq:beta_LG} \|\beta(h)\|_{L_2(\R^d,H_{\rho}^{m+1})}^2 &\leq K(1+\|h\|_{\rho,m+1}^2),
\\ \label{eq:gamma_LG} \int_E \|\gamma(h,x)\|_{\rho,m+1}^2 F(\ud x) &\leq K(1+\|h\|_{\rho,m+1}^2).
\end{align}
\end{enumerate}
\end{proposition}
\begin{proof}
Taking into account Lemma \ref{lemma-HS} and Remark \ref{rem-product-space}, the estimates \eqref{eq:beta_Lip} and \eqref{eq:beta_LG} follow from \eqref{beta-loc-Lip} and \eqref{LG-beta-i}, respectively. We structure the remaining part of the proof in a similar way to \cite[Proposition 3.1]{fitate2010}.
 Let $i\in\{0,1,\ldots,m\}$ and write
\be	\label{eq:alpha_dec}
\alpha^i(t,h) = \alpha^i_1(h)-\alpha^i_2(t,h)+\alpha^i_3(t,h),
\ee
where, for all $(\omega,t,h) \in \Omega\times\R_+ \times H^{m+1}_{\rho}$, 
\[	\begin{aligned}
\alpha^i_1(h) &:= \sum_{j=1}^d\beta^{i,j}(h)\bar{\beta}^{i,j}(h),\\
\alpha^i_2(t,\omega,h) &:= \sum_{j=1}^d\bigl(\lambda^j_t(\omega)+b^{i,j}_t(\omega)\bigr)\beta^{i,j}(h),\\
\alpha^i_3(t,\omega,h) &:= \int_E\gamma^i(h,x)\Bigl(1-e^{-\bar{\gamma}^i(h,x)}\bigl(1+\psi_t(\omega,x)\bigr)\bigl(1+c^i_t(\omega,x)\bigr)\Bigr)F(\ud x).
\end{aligned}	\]
Throughout the proof, for simplicity of notation we denote $\|h\|:=\|h\|_{\rho,1}$, for $h \in H_{\rho}^{1}$. 
Taking into account \eqref{beta-loc-Lip} and \eqref{LG-beta-i} and making use of the notation introduced in Appendix \ref{app-space}, Proposition \ref{prop-beta} implies that
\begin{align}\label{alpha-i-1}
\alpha_1^i \in \Lip^{\loc}(H_{\rho}^{m+1},H_{\rho}^{0,1}) \cap \LG(H_{\rho}^{m+1},H_{\rho}^{0,1}).
\end{align}
Moreover, we have $\alpha_2^i : \R_+ \times \Omega \times H^{m+1}_{\rho} \to H^{0,1}_{\rho}$ and $\alpha_2^i$ is $\cP \otimes \cB(H^{m+1}_{\rho})$-measurable, because $\lambda^j$ and $b^{i,j}$ are predictable processes, for all $j=1,\ldots,d$. Using \eqref{beta-loc-Lip}, there exists an increasing function $L : \R_+ \to \R_+$ such that, for all $j=1,\ldots,d$,  $r \in \R_+$ and $h,g \in H_{\rho}^{m+1}$ with $\| h \|_{\rho,m+1}\vee\| g \|_{\rho,m+1} \leq r$, we have
\[
\| \beta^{i,j}(h) - \beta^{i,j}(g) \| \leq L(r) \| h-g \|_{\rho,m+1}.
\]
Let $r \in \R_+$ be arbitrary. By Assumption \ref{ass:SPDE}, the non-negative process $L^{i,r} := L(r) \sum_{j=1}^d |\lambda^j_t+b^{i,j}_t |$ is optional and locally bounded. For all $(\omega,t) \in \Omega\times\R_+ $ and $h,g \in H_{\rho}^{m+1}$ with $\| h \|_{\rho,m+1}\vee \| g \|_{\rho,m+1} \leq r$, we obtain
\begin{align}\label{alpha-i-2a}
\| \alpha^i_2(t,\omega,h) - \alpha^i_2(t,\omega,g) \| \leq L_t^{i,r}(\omega) \| h-g \|_{\rho,m+1}.
\end{align}
By \eqref{LG-beta-i} there exists a constant $N_{\beta} \in \R_+$ such that, for all $j=1,\ldots,d$ and $h \in H_{\rho}$, it holds that
\[
\| \beta^{i,j}(h) \| \leq N_{\beta} ( 1 + \| h \|_{\rho,m+1}).
\]
The non-negative process $L^{i} := N_{\beta} \sum_{j=1}^d |\lambda^j_t+b^{i,j}_t |$ is optional and locally bounded. For all $(\omega,t,h) \in \Omega\times\R_+ \times H_{\rho}^{m+1}$, we obtain
\begin{align}\label{alpha-i-2b}
\| \alpha^i_2(t,\omega,h) \| \leq L_t^{i}(\omega) ( 1 + \| h \|_{\rho,m+1}).
\end{align}
Taking into account \eqref{est-psi-c}, \eqref{int-E-2-ass} and \eqref{int-E-3-ass}, we can apply Proposition \ref{prop-gamma} with $(Z,\cZ) = (\R_+ \times \Omega, \cP)$. As a consequence, we have $\alpha_3^i : \R_+ \times \Omega \times H^{m+1}_{\rho} \to H^{0,1}_{\rho}$ and $\alpha_3^i$ is $\cP \otimes \cB(H^{m+1}_{\rho})$-measurable, because the functions $\psi$ and $c^i$ are $\cP \otimes \cB(E)$-measurable. Furthermore, there exists an increasing function $L_1^i : \R_+ \to \R_+$ such that, for all $r \in \R_+$ and $h,g \in H_{\rho}^{m+1}$ with $\| h \|_{\rho,m+1}\vee \| g \|_{\rho,m+1} \leq r$, we have
\begin{align}\label{gamma-i-1}
\int_E \| \gamma^i(h,x) - \gamma^i(g,x) \|^2 F(\ud x) &\leq L_1^i(r) \| h-g \|_{\rho,m+1}^2,
\\ \label{alpha-i-3a} \| \alpha^i_3(t,\omega,h) - \alpha^i_3(t,\omega,g) \| &\leq L_1^i(r)(1 + \Lambda_t(\omega)) \| h-g \|_{\rho,m+1}, 
\quad \text{ for all }(\omega,t) \in \Omega\times\R_+,
\end{align}
and there exists a constant $L_2^i \in \R_+$ such that, for all $h \in H_{\rho}^{m+1}$, we have
\begin{align}\label{gamma-i-2}
\int_E \| \gamma^i(h,x) \|^2 F(\ud x) &\leq L_2^i (1+\|h\|_{\rho,m+1}^2),
\\ \label{alpha-i-3b} \| \alpha^i_3(t,\omega,h) \| &\leq L_2^i(1 + \Lambda_t(\omega)) (1+\|h\|_{\rho,m+1}), 
\quad \text{ for all }(\omega,t) \in \Omega\times\R_+.
\end{align}
Taking into account Remark \ref{rem-product-space}, the estimates \eqref{eq:gamma_Lip} and \eqref{eq:gamma_LG} concerning $\gamma$ follow from \eqref{gamma-i-1} and \eqref{gamma-i-2}. Furthermore, we obtain $\alpha : \R_+ \times \Omega \times H^{m+1}_{\rho} \to H^{0,m+1}_{\rho}$ and $\alpha$ is $\cP \otimes \cB(H^{m+1}_{\rho})$-measurable. The estimates \eqref{eq:alpha_Lip} and \eqref{eq:alpha_LG} concerning $\alpha$ follow from \eqref{alpha-i-1}, \eqref{alpha-i-2a}, \eqref{alpha-i-2b}, \eqref{alpha-i-3a} and \eqref{alpha-i-3b}.
\end{proof}

By relying on Assumption \ref{ass:SPDE} and Proposition \ref{prop:alpha}, we are in a position to prove the main result of this section, which establishes existence and uniqueness of the solution to the real-world HJMM SPDE \eqref{eq:SPDE} with drift \eqref{eq:drift_HJM}. As noted in Remark \ref{rem:extend_fitate}, this result extends \cite[Theorem 3.2]{fitate2010}, which, to the best of our knowledge, represents the most general available result on the well-posedness of the HJMM-SPDE. 
We recall that we are considering the state space $H = H_{\rho}^{m+1}$.

\begin{theorem}\label{thm-HJMM-SPDE}
Suppose that Assumption \ref{ass:SPDE} holds. Then, for every initial family of curves $h_0\in H$, there exists a unique global mild solution $(\eta_t)_{t\geq0}$ with c\`adl\`ag paths to \eqref{eq:SPDE} with $\eta_0=h_0$.
\end{theorem}
\begin{proof}
As a consequence of Theorem \ref{thm-space-forward} and Remark \ref{rem-product-space}, the space $(H,\|\cdot\|_H)$ is a separable Hilbert space and the shift semigroup $(\cS_t)_{t \geq 0}$ is a $C_0$-semigroup on $H$ with generator $A$ given by $Ah = h'$, for all $h \in D(A)$. Moreover, \cite[Lemma 3.5]{Benth-Kruehner} implies that the shift semigroup $(\cS_t)_{t \geq 0}$ is pseudo-contractive.
By Proposition \ref{prop:alpha}, conditions \eqref{Lip-1}--\eqref{LG-3} are satisfied. The claim then follows by Theorem \ref{thm-SPDE}.
\end{proof}

\begin{example}	\label{example:unbounded}
We provide an example of a class of models with unbounded coefficients $\beta$ and $\gamma$ for which Assumption \ref{ass:SPDE} is fulfilled. 
For simplicity, we consider the situation $d=1$ (one-dimensional Wiener process) and $m=0$ (single curve setting) and agree on the notation $H_{\rho} := H_{\rho}^1$, in accordance with Appendix \ref{app-space}.
First of all, let $\lambda$ and $\psi$ be any processes such that items (i)-(ii) of Assumption \ref{ass:SPDE} are satisfied. In particular, we can have $\lambda \equiv 0$ and $\psi \equiv 0$, which corresponds to the situation under a risk-neutral probability $\QQ$ (see Section \ref{sec:RN}). In that case, this example yields a class of risk-neutral models not covered by the results of \cite{fitate2010}, due to the unboundedness of $\beta$ and $\gamma$.
 
We assume that the diffusive coefficient $\beta : H_{\rho} \to H_{\rho}$ is given by a constant direction volatility
\begin{align}\label{beta-in-example}
\beta(h) = \Phi_{\beta}(h) \cdot \lambda_{\beta}, 
\quad\text{ for } h \in H_{\rho}
\end{align}
with $\lambda_{\beta} \in H_{\rho}^0$ and a locally Lipschitz functional $\Phi_{\beta} : H_{\rho} \to \R$ satisfying the growth condition
\begin{align}\label{growth-Phi}
| \Phi_{\beta}(h) | \leq C_1 \sqrt{ 1 + \| h \|_{\rho} }, 
\quad \text{ for all $h \in H$,}
\end{align}
for some constant $C_1 > 0$. Then condition \eqref{beta-loc-Lip} is satisfied and \eqref{LG-beta-i} is fulfilled with $M_{\beta} = C_1 \| \lambda_{\beta} \|_{\rho}$.
The functional $\Phi_{\beta}$ can for instance be of the form $\Phi_{\beta} = f \circ \ell$ with a continuous linear functional $\ell \in L(H_{\rho},\R)$ and a locally Lipschitz function $f : \R \to \R$ satisfying the growth condition
\begin{align}\label{growth-square-root}
|f(\xi)| \leq C_2 \sqrt{1 + |\xi|}, 
\quad \text{ for all $\xi \in \R$,}
\end{align}
for some constant $C_2 > 0$. Indeed, in this case the functional $\Phi_{\beta}$ is locally Lipschitz and the estimate \eqref{growth-Phi} is satisfied with $C_1 = C_2 \sqrt{1 + \| \ell \|}$.
Examples for the linear functional $\ell$ are the point evaluations
\begin{align*}
\mathcal{J}_x : H_{\rho} \to \R, \quad 
\mathcal{J}_x(h) := h(x)
\end{align*}
for any $x \in \R_+$, and the integral functionals
\begin{align*}
\mathcal{I}_u : H_{\rho} \to \R, \quad \mathcal{I}_u(h) := \int_0^u h(x) \, \ud x.
\end{align*}
for any $u \in \R_+$, see \cite[Theorem 5.1.1 and Lemma 4.3.1]{F:01}. Furthermore, for any fixed constant $\eta > 0$ an example of a locally Lipschitz function $f : \R \to \R$ satisfying \eqref{growth-square-root} is given by
\begin{align*}
f(\xi) = \frac{|\xi|}{\sqrt{\eta}} \bbI_{\{ |\xi| \leq \eta \}} + \sqrt{|\xi|} \bbI_{\{ |\xi| > \eta \}}, 
\quad\text{ for } \xi \in \R.
\end{align*}

In order to specify the jump coefficient $\gamma$, we choose another constant $\rho' > \rho$. Let $W : \R_+ \to \R_+$ be a Lipschitz continuous function such that
\begin{align}\label{W-inequ}
W(r) \leq W_{K_{\rho,\rho'}}(r), 
\quad \text{ for all $r \in \R_+$,}
\end{align}
where we recall that $W_{K_{\rho,\rho'}} : \R_+ \to \R_+$ denotes the inverse of the strictly increasing function \eqref{V_K_def} with $K = K_{\rho,\rho'}$. We assume that the jump coefficient $\gamma : H_{\rho} \times E \to H_{\rho'}$ is given by
\begin{align*}
\gamma(h,x) = \Phi_{\gamma}(h,x) \cdot \lambda_{\gamma}, 
\quad \text{ for }(h,x) \in H_{\rho} \times E,
\end{align*}
with $\lambda \in H_{\rho'}^0$ and a mapping $\Phi_{\gamma} : H_{\rho} \times E \to \R_+$ given by
\begin{align}\label{Phi-gamma-in-example}
\Phi_{\gamma}(h,x) = c \cdot W( \kappa(x) \Psi(h) ), 
\quad\text{ for } (h,x) \in H_{\rho} \times E,
\end{align}
where the constant $c \in \R_+$ satisfies $c \leq \| \lambda_{\gamma} \|_{\rho'}^{-1}$ and the function $\Psi : H_{\rho} \to \R_+$ is locally Lipschitz and satisfies the linear growth condition. Then, $\gamma$ is $\cB(H_{\rho}) \otimes \cB(E)$-measurable. Furthermore, using \eqref{W-inequ} and the monotonicity of $W_{K_{\rho,\rho'}}$, we have that
\begin{align*}
\| \gamma(h,x) \|_{\rho'} &= \| \lambda_{\gamma} \|_{\rho'} \, | \Phi_{\gamma}(h,x) | \leq W( \kappa(x) \Psi(h) ) \leq W_{K_{\rho,\rho'}}( \kappa(x) \Psi(h) ) \leq W_{K_{\rho,\rho'}} \big( C_{\Psi} \kappa(x) (1 + \| h \|_{\rho}) \big),
\end{align*}
for all $h \in H_{\rho}$ and $x \in E$, where $C_{\Psi} \in \R_+$ denotes a growth constant of $\Psi$. 
This shows that condition \eqref{int-E-3-ass} holds with $N_{\gamma} = C_{\Psi}$. Moreover, for $r \in \R_+$ and $h,g \in H_{\rho}$ with $\| h \|_{\rho}\vee \| g \|_{\rho} \leq r$, it holds that
\begin{align*}
\| \gamma(h,x) - \gamma(g,x) \|_{\rho'} &= \| \lambda_{\gamma} \|_{\rho'} \, | \Phi_{\gamma}(h,x) - \Phi_{\gamma}(g,x) | \leq | W(\kappa(x) \Psi(h)) - W(\kappa(x) \Psi(g)) |
\\ &\leq L_W | \kappa(x) \Psi(h) - \kappa(x) \Psi(g) |
= L_W \kappa(x) | \Psi(h) - \Psi(g) |
\leq L_W \kappa(x) L_{\Psi}(r) \| h-g \|_{\rho},
\end{align*}
where $L_W \in \R_+$ denotes a Lipschitz constant of $W$ and $L_{\Psi}$ denotes a Lipschitz function of $\Psi$. This shows the validity of condition \eqref{int-E-2-ass} with $M_{\gamma}(r) = L_W \cdot L_{\Psi}(r)$.
We have therefore shown that all requirements of Assumption \ref{ass:SPDE} are satisfied.

To make the choice of the function $W$ in \eqref{Phi-gamma-in-example} more concrete, let us consider two examples of Lipschitz continuous functions $W : \R_+ \to \R_+$ satisfying \eqref{W-inequ}:
\begin{enumerate}
\item we can simply take $W := W_{K_{\rho,\rho'}}$, which obviously satisfies condition \eqref{W-inequ}. Moreover, by the Leibniz rule, letting $K = K_{\rho,\rho'}$, we have that
\begin{align*}
V_K'(r) = r(1+r) K e^{Kr} + 2r e^{Kr} + (1+r) e^{Kr} \geq 1, 
\quad \text{ for all $r \in \R_+$,}
\end{align*}
where we recall that $V_K$ is given by \eqref{V_K_def}. Thus, by the inverse mapping theorem for differentiation we have $W'(r) \leq 1$ for all $r \in \R_+$, which proves that $W$ is Lipschitz continuous.
\item A drawback of example (1) is that $W_{K_{\rho,\rho'}}$ is not available in closed form. For a more practical choice, set $C :=  2 (K_{\rho,\rho'}\vee 1)$ and define the Lipschitz continuous function $W : \R_+ \to \R_+$ as
\begin{align*}
W(r) := \frac{\ln(r+1)}{C}, 
\quad\text{ for } r \in \R_+.
\end{align*}
Obviously, $W$ is the inverse of the strictly increasing function $V : \R_+ \to \R_+$ given by
\begin{align*}
V(r) := \exp(Cr) - 1, 
\quad\text{ for } r \in \R_+.
\end{align*}
Setting $M := C/2=K_{\rho,\rho'}\vee 1$ and recalling \eqref{V_K_def}, we can compute
\begin{align*}
1 + V_{K_{\rho,\rho'}}(r) &\leq 1 + r(1+r)\exp(Mr) = 1 + \sum_{k=0}^{\infty} \frac{M^k r^{k+1}}{k!} + \sum_{k=0}^{\infty} \frac{M^k r^{k+2}}{k!}
\\ &= 1 + \sum_{k=1}^{\infty} \frac{M^{k-1} r^{k}}{k!} + \sum_{k=2}^{\infty} \frac{M^{k-2} r^{k}}{k!} = 1 + \frac{1}{M} \sum_{k=1}^{\infty} \frac{M^{k} r^{k}}{k!} + \frac{1}{M^2} \sum_{k=2}^{\infty} \frac{M^{k} r^{k}}{k!}
\\ &\leq 1 + \sum_{k=1}^{\infty} \frac{2^k M^{k} r^{k}}{k!} 
= \exp(2Mr) = \exp(Cr),
\end{align*}
for all $r\in\R_+$.
Therefore, we obtain $V_{K_{\rho,\rho'}}(r) \leq V(r)$ for all $r \in \R_+$, which implies \eqref{W-inequ}.
\end{enumerate}
\end{example}

\subsection{Monotonicity of term structures}
\label{sec:SPDE_invariant}

In Section \ref{sec:monotone}, we have derived abstract conditions ensuring the monotonicity of the term structures. In this section, we address the issue of monotonicity from an SPDE viewpoint, relating monotonicity to suitable invariance properties of the real-world SPDE \eqref{eq:SPDE}.
This approach has the advantage of providing sufficient conditions for monotonicity without requiring the martingale property of Definition \ref{def:fair}.
We work in the framework of Section \ref{sec-appl-rw-HJMM}, supposing in particular that Assumption \ref{ass:SPDE} holds. 
Consider the closed convex cone $K\subset H$ defined as
\begin{align}\label{cone}
K := \bigl\{ h = (h^0,h^1,\ldots,h^m) \in H : h^1 \geq h^2 \geq \ldots \geq h^m \bigr\},
\end{align}
where we recall that the state space is given by $H = H_{\rho}^{m+1}$. We say that the closed convex cone $K$ is \emph{invariant} for the SPDE \eqref{eq:SPDE} if, for each $h_0 \in K$, we have $\eta \in K$ up to an evanescent set, where $\eta$ denotes the mild solution to the SPDE \eqref{eq:SPDE} with $\eta_0 = h_0$. Note that the cone $K$ is invariant for the shift semigroup; that is $\cS_t K \subset K$ for all $t \geq 0$. The invariance of closed convex cones for SPDEs has been investigated in a general setup in \cite{TappeCones17,TappeCones24} and in \cite{fitate2010} for the special case where $K$ is the cone of non-negative functions. The proof of the following theorem is analogous to that of \cite[Corollary 8.6]{TappeCones24} and, therefore, omitted.

\begin{theorem}\label{thm-cone-inv}
Suppose that, for all $h\in K$,
\begin{align}\label{cone-1}
h + \gamma(h,x) \in K \quad \text{$F(\ud x)$-a.e.}
\end{align}
and, for all $\xi \in \R_+$ and $i,j = 1,\ldots,m$ with $i < j$ such that $h^i(\xi) = h^j(\xi)$, we have
\begin{align}\label{cone-2}
&\beta^i(h)(\xi) = \beta^j(h)(\xi),
\\ \label{cone-3}  &\alpha^i(t,\omega,h)(\xi) - \alpha^j(t,\omega,h)(\xi)
\\ \notag &\quad - \int_E \big( \gamma^i(h,x)(\xi) - \gamma^j(h,x)(\xi) \big) F(\ud x) \geq 0, 
\quad \text{outside of a set of $(\PP\otimes \ud t)$-measure zero.}
%\lambda \otimes \PP$-almost all $(t,\omega) \in \R_+ \times \Omega$,}
\end{align}
Then, the closed convex cone $K$ is invariant for the SPDE \eqref{eq:SPDE}.
\end{theorem}

Taking into account the specific structure of the drift $\alpha$ given by \eqref{eq:drift_HJM}, we can now provide sufficient conditions for the invariance of the closed convex cone $K$ for the real-world HJMM SPDE \eqref{eq:SPDE}.

\begin{assumption}\label{ass-cone-inv}
Condition \eqref{cone-1} is satisfied. Furthermore, we assume that for all $h \in K$, $\xi \in \R_+$ and $i,j = 1,\ldots,m$ with $i < j$ such that $h^i(\xi) = h^j(\xi)$, the following conditions are satisfied:
\begin{enumerate}[(i)]
\item for all $k=1,\ldots,d$, it holds that
\begin{align}\label{cone-2-1}
&\beta^{i,k}(h)(\xi) = \beta^{j,k}(h)(\xi),
\\ \label{cone-2-2} &\beta^{i,k}(h)(\chi) \geq \beta^{j,k}(h)(\chi), \quad \chi \in [0,\xi), \quad \text{if $\beta^{i,k}(h)(\xi) > 0$,}
\\ \label{cone-2-3} &\beta^{i,k}(h)(\chi) \leq \beta^{j,k}(h)(\chi), \quad \chi \in [0,\xi), \quad \text{if $\beta^{i,k}(h)(\xi) < 0$;}
\end{align}
\item for $F(\ud x)$-a.e. $x \in E$, it holds that
\begin{align}\label{cone-4-1}
&\gamma^i(h,x)(\xi) = \gamma^j(h,x)(\xi),
\\ \label{cone-4-2} &\gamma^i(h,x)(\chi) \geq \gamma^j(h,x)(\chi), \quad \chi \in [0,\xi), \quad \text{if $\gamma^i(h,x)(\xi) > 0$,}
\\ \label{cone-4-3} &\gamma^i(h,x)(\chi) \leq \gamma^j(h,x)(\chi), \quad \chi \in [0,\xi), \quad \text{if $\gamma^i(h,x)(\xi) < 0$.}
\end{align}
\end{enumerate}
\end{assumption}

\begin{proposition}\label{prop-cone-inv}
If Assumption \ref{ass-cone-inv} and Condition \ref{cond:order} hold, then the closed convex cone $K$ is invariant for the real-world HJMM SPDE \eqref{eq:SPDE}.
\end{proposition}
\begin{proof}
Let us consider $h \in K$, $\xi \in \R_+$ and $i,j = 1,\ldots,m$ with $i < j$ such that $h^i(\xi) = h^j(\xi)$. Condition \eqref{cone-2} immediately follows from \eqref{cone-2-1}. Taking into account Condition \ref{cond:order} as well as conditions \eqref{cone-2-1} and \eqref{cone-4-1}, the structure \eqref{eq:drift_HJM} of the drift $\alpha$ shows that \eqref{cone-3} is satisfied if and only if
\begin{align*}
&\beta^i(h)(\xi)^{\top} \bar{\beta}^i(h)(\xi) - \int_E \gamma^i(h,x)(\xi) \Big( e^{-\bar{\gamma}^i(h,x)(\xi)} \big( 1 + \psi_t(\omega,x) \big) \big( 1 + c_t^i(\omega,x) \big) \Big) F(\ud x)
\\ &\geq \beta^i(h)(\xi)^{\top} \bar{\beta}^j(h)(\xi) - \int_E \gamma^i(h,x)(\xi) \Big( e^{-\bar{\gamma}^j(h,x)(\xi)} \big( 1 + \psi_t(\omega,x) \big) \big( 1 + c_t^i(\omega,x) \big) \Big) F(\ud x)
\end{align*}
for $(\PP\otimes\ud t)$-a.e. $(\omega,t) \in \Omega\times\R_+$, which is equivalent to
\begin{align*}
&\beta^i(h)(\xi)^{\top} \big( \bar{\beta}^i(h)(\xi) - \bar{\beta}^j(h)(\xi) \big)
\\ &- \int_E \gamma^i(h,x)(\xi) \Big( \big( e^{-\bar{\gamma}^i(h,x)(\xi)} - e^{-\bar{\gamma}^j(h,x)(\xi)} \big) \big( 1 + \psi_t(\omega,x) \big) \big( 1 + c_t^i(\omega,x) \big) \Big) F(\ud x) \geq 0
\end{align*}
for $(\PP\otimes\ud t)$-a.e. $(\omega,t) \in \Omega\times\R_+$. Recalling that $\psi$ is $(-1,+\infty)$-valued and $c^i$ is $[-1,+\infty)$-valued, conditions \eqref{cone-2-2}, \eqref{cone-2-3}, \eqref{cone-4-2}, \eqref{cone-4-3} imply that \eqref{cone-3} is fulfilled. Consequently, in view of Theorem \ref{thm-cone-inv}, the closed convex cone $K$ is invariant for the HJMM SPDE \eqref{eq:SPDE}.
\end{proof}

Proposition \ref{prop-cone-inv} can be applied to deduce sufficient conditions ensuring the monotonicity of the term structures, without relying on the martingale property of Definition \ref{def:fair}. 
In particular, note that the conditions stated in Assumption \ref{ass-cone-inv} are entirely deterministic, making them more tractable for verification in concrete model specifications compared to the martingale property of Definition \ref{def:fair}.
For an initial family of curves $h_0 \in H$, we recall that
\begin{align}\label{spread-cone}
S^i_tB^i(t,T)
= S^i_t\exp \bigg( -\int_0^{T-t} \eta_t^i(\xi) \ud \xi \bigg),
 \qquad \text{ for all }i=1,\ldots,m,
\end{align}
where $\eta = (\eta_t)_{t \geq 0}$ denotes the mild solution to the real-world HJMM SPDE \eqref{eq:SPDE} with $\eta_0 = h_0$.

\begin{proposition}\label{prop-monotonicity}
Suppose that Condition \ref{cond:order}  and Assumption \ref{ass-cone-inv}  hold. If  $h_0 \in K$ and $S^i_0\leq S^j_0$, for all $i,j=1,\ldots,m$ with $i<j$, then, for all $0\leq t\leq T< +\infty$, it holds that
\[
S^i_tB^i(t,T) \leq S^j_tB^j(t,T),
\qquad\text{ for all }i,j=1,\ldots,m\text{ with }i<j.
\]
\end{proposition}
\begin{proof}
According to Proposition \ref{prop-cone-inv}, it holds that $\eta \in K$, meaning that
\[
\eta_t^1(\xi) \geq \eta_t^2(\xi) \geq \ldots \geq \eta_t^m(\xi), 
\qquad \text{for all $(t,\xi)\in\R^2_+$.}
\]
In view of \eqref{spread-cone}, we have that
\[
S^i_tB^i(t,T) = S^i_0\, \cE(Z^i)_t  \exp \bigg( -\int_0^{T-t} \eta_t^i(\xi) \ud \xi \bigg), 
\qquad \text{ for all }i=1,\ldots,m.
\]
Similarly as in the proof of Proposition \ref{prop:order}, Condition \ref{cond:order} implies that $\cE(Z^1) \leq \ldots \leq \cE(Z^m)$, which suffices to prove the desired monotonicity property.
 \end{proof}

\subsection{Existence of affine realizations}	\label{sec:affine}

The concept of finite-dimensional realizations allows reducing the inherently infinite-dimensional structure of HJM models to finite-dimensional factor models, which are more tractable for practical applications. The existence of finite-dimensional realizations has been the subject of substantial investigation in the context of one-dimensional HJM interest rate models under the risk-neutral setup (see, e.g., \cite{Bjoerk-Svensson, Filipovic-Teichmann} and also \cite{bjork2004} for an overview on the topic). For multi-curve HJM interest rate models, conditions for the existence of finite-dimensional realizations have been recently obtained in \cite{FLM25} in the risk-neutral setup and assuming continuous paths.
In this section, we extend these results by providing conditions for the existence of affine finite-dimensional realizations in the context of L\'evy-driven HJM models for multiple term structures under the real-world probability. 

We start by considering an SPDE of the following form on a separable Hilbert space $H$:
\begin{equation}\label{SPDE-Levy}
dX_t = \bigl(A X_t + \alpha(t,X_t)\bigr) \ud t + \beta(t,X_{t-}) \ud L_t,
\qquad X_0 = x_0.
\end{equation}
%\begin{align}\label{SPDE-Levy}
%\left\{
%\begin{array}{rcl}
%dX_t & = & (A X_t + \alpha(t,X_t)) \ud t + \beta(t,X_{t-}) \ud L_t,
%\\ X_0 & = & x_0,
%\end{array}
%\right.
%\end{align}
where $L$ is an $\R^d$-valued L\'{e}vy process, $A$ denotes the generator of a pseudo-contractive $C_0$-semigroup $(\cS_t)_{t \geq 0}$ on $H$ and $\alpha : \Omega\times\R_+ \times H \to H$ and $\beta :\Omega\times \R_+  \times H \to L_2(\R^d,H)$. For each $j=1,\ldots,d$, we introduce the mapping $\beta^j :\Omega\times \R_+  \times H \to H$ given by $\beta^j(\omega,t,h) := \beta(\omega,t,h)e_j$. We assume that conditions \eqref{Lip-1}, \eqref{Lip-2} and \eqref{LG-1}, \eqref{LG-2} from Theorem \ref{thm-SPDE} are satisfied, which ensures existence and uniqueness of mild solutions to the L\'{e}vy-driven SPDE \eqref{SPDE-Levy}.

\begin{definition}
Let $\phi = (\phi^{x_0})_{x_0 \in D(A)}$ be a family of $D(A)$-valued mappings $\phi^{x_0} \in C^1(\R_+;H)$ and $T \in L(\R^n,V)$ an isomorphism, where $V \subset D(A)$ is an $n$-dimensional subspace, for some $n\in\N$. The SPDE \eqref{SPDE-Levy} admits an \emph{$n$-dimensional affine realization} induced by $(\phi,T)$ if, for every $x_0 \in D(A)$, there exists an $\R^n$-valued c\`{a}dl\`{a}g adapted process $Y^{x_0}$ such that the $D(A)$-valued process $X^{x_0}$ given by
\begin{align}\label{X-real}
X_t^{x_0} := \phi^{x_0}(t) + T Y_t^{x_0}, 
\qquad \text{ for all }t \in \R_+,
\end{align}
is a strong solution to the SPDE \eqref{SPDE-Levy} with $X_0^{x_0} = x_0$.
\end{definition}

The notion of affine realization is closely connected to invariant foliations (see, e.g., \cite{Tappe-Wiener, Tappe-Levy}, where the existence of affine realizations has been treated for one-dimensional HJM interest rate models under the risk-neutral setup). In this section, we do not pursue a systematic investigation of affine realizations in our general setup, rather we establish sufficient conditions which apply to the real-world HJMM SPDE \eqref{eq:SPDE} later on. We introduce the mapping $\nu : \Omega\times\R_+ \times D(A) \to H$ as $\nu(\omega,t,h) := Ah + \alpha(\omega,t,h)$.

\begin{proposition}\label{prop-FDR-coeff}
Let $\phi = (\phi^{x_0})_{x_0 \in D(A)}$ be a family of $D(A)$-valued mappings $\phi^{x_0} \in C^1(\R_+;H)$ with $\phi^{x_0}(0) = x_0$ such that $A \phi^{x_0} : \R_+ \to H$ is continuous for each $x_0 \in D(A)$, and let $V \subset D(A)$ be an $n$-dimensional subspace. Suppose that, for all $x_0 \in D(A)$, $(\omega,t) \in\Omega\times \R_+$ and $v \in V$, we have
\begin{align}\label{inv-cond-V-1}
\nu(\omega,t,\phi^{x_0}(t)+v) &\in \frac{\ud}{\ud t} \phi^{x_0}(t) + V,
\\ \label{inv-cond-V-2} 
\beta^j (\omega,t,\phi^{x_0}(t)+v) &\in V, 
\quad \text{ for all }j=1,\ldots,d.
\end{align}
Then, for every isomorphism $T \in L(\R^n,V)$, the SPDE \eqref{SPDE-Levy} admits an $n$-dimensional affine realization induced by $(\phi,T)$.
\end{proposition}
\begin{proof}
Let $T \in L(\R^n,V)$ be an arbitrary isomorphism and $x_0 \in D(A)$. By \eqref{inv-cond-V-1} and \eqref{inv-cond-V-2}, there exist mappings $\mu^{x_0} : \Omega\times\R_+ \times \R^n \to \R^n$ and $\gamma^{x_0} :\Omega\times \R_+  \times \R^n \to L_2(\R^d,\R^n)$ such that, for all $(\omega,t,y) \in \Omega\times\R_+ \times \R^n$, we have
\begin{align}\label{inv-cond-coeff-1}
\nu(\omega,t,\phi^{x_0}(t)+Ty) &= \frac{\ud}{\ud t} \phi^{x_0}(t) + T \mu^{x_0}(\omega,t,y),
\\ \label{inv-cond-coeff-2} \beta^j(\omega,t,\phi^{x_0}(t)+Ty) &= T \gamma^{x_0,j}(\omega,t,y), 
\quad \text{ for all }j=1,\ldots,d.
\end{align}
Noting that $A \phi^{x_0} : \R_+ \to H$ is continuous, from \eqref{inv-cond-coeff-1} and \eqref{inv-cond-coeff-2} we deduce that $\mu^{x_0}$ and $\gamma^{x_0}$ satisfy the local Lipschitz conditions \eqref{Lip-a}, \eqref{Lip-b} and the linear growth conditions \eqref{LG-a}, \eqref{LG-b}. Hence, by Theorem \ref{thm-SDE} there exists a unique strong solution $Y^{x_0}$ to the $\R^n$-valued SDE
\begin{equation}\label{SDE-Levy}
dY_t^{x_0} = \mu^{x_0}(t,Y_t) \ud t + \gamma^{x_0}(t,Y_{t-}) \ud L_t,
\qquad Y_0^{x_0} = 0.
\end{equation}
%\begin{align}\label{SDE-Levy}
%\left\{
%\begin{array}{rcl}
%dY_t^{x_0} & = & \mu^{x_0}(t,Y_t) \ud t + \gamma^{x_0}(t,Y_{t-}) \ud L_t,
%\\ Y_0^{x_0} & = & 0.
%\end{array}
%\right.
%\end{align}
Let us now define the $D(A)$-valued process $X^{x_0}$ by \eqref{X-real}. In view of \eqref{inv-cond-coeff-1} and \eqref{inv-cond-coeff-2}, we obtain
\begin{align*}
X_t^{x_0} &= \phi^{x_0}(t) + T Y_t^{x_0} 
\\ &= \phi^{x_0}(0) + \int_0^t \frac{\ud}{\ud s} \phi^{x_0}(s) \ud s + \int_0^t T \mu^{x_0}(s,Y_s^{x_0} )\ud s + \int_0^t T \gamma^{x_0}(s,Y_{s-}^{x_0} ) \ud L_s
\\ &= x_0 + \int_0^t \nu(s,X_s^{x_0} )\ud s + \int_0^t \beta(s,X_{s-}^{x_0} ) \ud L_s, 
\qquad \text{ a.s. for all }t \in \R_+,
\end{align*}
thus proving that $X^{x_0}$ is a strong solution to the SPDE \eqref{SPDE-Levy} with $X_0^{x_0} = x_0$.
\end{proof}

We point out that the proof of Proposition \ref{prop-FDR-coeff} is constructive in the sense that the state processes $Y^{x_0}$, for each $x_0 \in D(A)$, is given by the solution to the $\R^n$-valued SDE \eqref{SDE-Levy}.

\begin{proposition}\label{prop-real-suff}
Let $V \subset D(A)$ be an $n$-dimensional $A$-invariant subspace such that the following conditions are fulfilled:
\begin{enumerate}
\item for all $(\omega,t,h) \in \Omega\times\R_+ \times H$, we have
\begin{align}\label{beta-in-V}
\beta^j(\omega,t,h) \in V, 
\quad\text{ for all } j=1,\ldots,d;
\end{align}
\item there exists a $D(A)$-valued mapping $\alpha_1 \in C^1(\R_+;H)$ such that $A \alpha_1 : \R_+ \to H$ is continuous and a mapping $\alpha_2 :\Omega\times \R_+ \times H \to V$ such that
\begin{align}\label{decomp-alpha}
\alpha(\omega,t,h) = \alpha_1(t) + \alpha_2(\omega,t,h), 
\quad\text{ for all } (\omega,t,h) \in\Omega\times \R_+  \times H.
\end{align}
\end{enumerate}
Define the family $\phi = (\phi^{x_0})_{x_0 \in D(A)}$ by
\[
\phi^{x_0}(t) := \cS_t x_0 + \int_0^t \cS_{t-s} \alpha_1(s) \ud s, 
\quad \text{ for all }t \in \R_+.
\]
Then, for every isomorphism $T \in L(\R^n,V)$, the SPDE \eqref{SPDE-Levy} has an $n$-dimensional affine realization induced by $(\phi,T)$.
\end{proposition}
\begin{proof}
For $x_0 \in D(A)$, the mapping $\phi^{x_0}$ is a solution to the inhomogeneous initial value problem
\[
\frac{\ud}{\ud t} \phi(t) = A \phi(t) + \alpha_1(t),
\qquad \phi(0) = x_0,
\]
%\begin{align*}
%\left\{
%\begin{array}{rcl}
%\frac{\ud}{\ud t} \phi(t) & = & A \phi(t) + \alpha_1(t),
%\\ \phi(0) & = & x_0,
%\end{array}
%\right.
%\end{align*}
see for instance \cite[Theorem 4.2.4]{Pazy}. In particular, it follows that $A \phi^{x_0} : \R_+ \to H$ is continuous. Condition \eqref{inv-cond-V-2} is satisfied due to \eqref{beta-in-V}. Furthermore, for all $(\omega,t) \in \Omega\times\R_+ $ and $v \in V$, we obtain
\begin{align*}
\nu(\omega,t,\phi^{x_0}(t)+v) &= A \phi^{x_0}(t) + Av + \alpha_1(t) + \alpha_2(\omega,t,\phi^{x_0}(t)+v)
\\ &= \frac{\ud}{\ud t} \phi^{x_0}(t) + A v + \alpha_2(\omega,t,\phi^{x_0}(t)+v) \in \frac{\ud}{\ud t} \phi^{x_0}(t) + V,
\end{align*}
where in the last step we have used that $V$ is $A$-invariant. This shows that condition \eqref{inv-cond-V-1} holds. The claim then follows by applying Proposition \ref{prop-FDR-coeff}.
\end{proof}

\begin{proposition}\label{prop-real-trans}
Suppose that the SPDE \eqref{SPDE-Levy} has an $n$-dimensional affine realization induced by $(\phi,T)$. Let $\ell \in L(H,\R^n)$ be such that the linear operator $U := \ell T \in L(\R^n)$ is an isomorphism and let $R \in L(\R^n,V)$ be the isomorphism $R := T U^{-1}$. Then, for each $x_0 \in D(A)$, the following hold:
\begin{enumerate}
\item the strong solution $X^{x_0}$ to the SPDE \eqref{SPDE-Levy} with $X_0^{x_0} = x_0$ satisfies
\begin{align}\label{X-real-trans}
X_t^{x_0} = \varphi^{x_0}(t) + R Z_t^{x_0}, 
\quad\text{ for all } t \in \R_+,
\end{align}
where $Z^{x_0} := \ell(X^{x_0})$, and the mapping $\varphi^{x_0} : \R_+ \to H$ is given by
\begin{align}\label{varphi-trans}
\varphi^{x_0}(t) = (\Id - R \ell) \phi^{x_0}(t), 
\quad\text{ for all } t \in \R_+;
\end{align}
\item the process $Z^{x_0}$ is the strong solution to the $\R^n$-valued SDE
\begin{equation}\label{SDE-Z-Levy}
dZ_t^{x_0} = \bar{\mu}^{x_0}(t,Z_t^{x_0}) \ud t + \bar{\gamma}^{x_0}(t,Z_{t-}^{x_0}) \ud L_t,
\qquad Z_0^{x_0} = \ell(x_0),
\end{equation}
%\begin{align}\label{SDE-Z-Levy}
%\left\{
%\begin{array}{rcl}
%dZ_t^{x_0} & = & \bar{\mu}^{x_0}(t,Z_t^{x_0}) \ud t + \bar{\gamma}^{x_0}(t,Z_{t-}^{x_0}) \ud L_t,
%\\ Z_0^{x_0} & = & \ell(x_0),
%\end{array}
%\right.
%\end{align}
where $\bar{\mu}^{x_0} : \Omega\times\R_+ \times \R^n \to \R^n$ and $\bar{\gamma}^{x_0} : \Omega\times\R_+ \times \R^n \to L_2(\R^d,\R^n)$ are given by
\begin{align}\label{mu-trans}
\bar{\mu}^{x_0}(\omega,t,z) &= \ell \nu(\omega,t,\varphi^{x_0}(t) + R z),
\\ \label{gamma-trans} \bar{\gamma}^{x_0}(\omega,t,z) &= \ell \beta(\omega,t,\varphi^{x_0}(t) + R z).
\end{align}
\end{enumerate}
\end{proposition}
\begin{proof}
Applying the operator $\ell$ to equation \eqref{X-real} we see that
\[
Y_t^{x_0} = U^{-1} \ell ( X_t^{x_0} - \phi^{x_0}(t) ), 
\quad\text{ for all } t \in \R_+.
\]
Inserting this expression for $Y^{x_0}$ into \eqref{X-real} we arrive at \eqref{X-real-trans}. Furthermore, we have
\begin{align*}
\ell(X_t^{x_0}) &= \ell(x_0) + \int_0^t \ell \nu(s,X_s^{x_0}) \ud s + \int_0^t \ell \beta(s,X_{s-}^{x_0}) \ud L_s,
\quad\text{ for all } t \in \R_+.
\end{align*}
Together with \eqref{X-real-trans}, this completes the proof.
\end{proof}

After these preparations, we are now in a position to study the existence of affine realizations for the real-world HJMM SPDE \eqref{eq:SPDE} on the state space $H := H_{\rho}^{m+1}$, for some $\rho > 0$. For simplicity, we first assume that the SPDE \eqref{eq:SPDE} is driven by a one-dimensional Brownian motion $W$. Furthermore, we assume that the volatility $\beta$ is constant and given by
\begin{align}\label{beta-Vasicek}
\beta = \big( c_i e^{\delta_i \cdot} \big)_{i=0,1,\ldots,m}
\end{align}
with constants $c_i \in \R$ and $\delta_i < - \rho / 2$ for $i=0,1,\ldots,m$. In view of \eqref{eq:drift_HJM}, the drift $\alpha$ is then given by
\[
\alpha^i(\omega,t,h) = \beta^i \bar{\beta}^i - ( \lambda_t(\omega) + b_t^i(\omega) ) \beta^i, 
\qquad\text{ for } i=0,1,\ldots,m,
\]
where $b^0 \equiv 0$. We define the $(m+1)$-dimensional $D(\ud / \ud \xi)$-invariant subspace $V \subset D(\ud / \ud \xi)$ as
\begin{align}\label{V-subspace}
V := \big\{ \big( v^i e^{\delta_i \cdot} \big)_{i = 0,1,\ldots,m} : v = (v^0,v^1,\ldots,v^m) \in \R^{m+1} \big\}.
\end{align}
The drift $\alpha$ admits then a decomposition of the form \eqref{decomp-alpha} with
\[
\alpha_1 = \big( \beta^i \bar{\beta}^i \big)_{i=0,1,\ldots,m}
\qquad\text{ and } \qquad
\alpha_2(t,\omega) = \big( -( \lambda_t(\omega) + b_t^i(\omega) ) \beta^i \big)_{i=0,1,\ldots,m}.
\]
In view of Proposition \ref{prop-real-suff}, the real-world HJMM SPDE \eqref{eq:SPDE} admits an $(m+1)$-dimensional affine realization. We are especially interested in constructing a realization where the state process is given by the short end of the forward curves, in analogy to the concept of short-rate realization in the context of interest rate models (see, e.g., \cite{Bjoerk-Svensson}).
For this purpose, let $\ell \in L(H,\R^{m+1})$ be given by
\[
\ell(h) = (h^0(0),\ldots,h^m(0)), 
\qquad\text{ for } h \in H.
\]

\begin{proposition}\label{prop-Vasicek}
For each $h_0 \in D(\ud / \ud \xi)$, the following hold:
\begin{enumerate}
\item the strong solution $\eta^{h_0}$ to the real-world HJMM SPDE \eqref{eq:SPDE} with $\eta_0^{h_0} = h_0$ satisfies
\[
\eta_t^{h_0} = \varphi^{h_0}(t) + \big( Z_t^{h_0,i} e^{\delta_i \cdot} \big)_{i=0,1,\ldots,m}, 
\qquad \text{ for all }t \in \R_+,
\]
where $Z^{h_0} := \ell(\eta^{h_0})$ and the mapping $\varphi^{h_0} : \R_+ \to H$ has  components
\begin{align}\label{varphi-Vasicek}
\varphi^{h_0,i}(t) := \cS_t h_0^i - h_0^i(t) e^{\delta_i \cdot} + \frac{c_i^2}{2 \delta_i^2} \big( e^{2 \delta_i t} - 1 \big) \big( e^{\delta_i \cdot} - 1 \big) e^{\delta_i \cdot}, 
\;\text{ for all } t \in \R_+\text{ and }i=0,1,\ldots,m;
\end{align}
\item the state process $Z^{h_0}$ is the strong solution to the $\R^{m+1}$-valued SDE
\[
dZ_t^{h_0} = \bar{\mu}^{h_0}(t,Z_t^{h_0}) \ud t + c \, \ud W_t,
\qquad Z_0^{h_0} = \ell(h_0),
\]
%\begin{align*}
%\left\{
%\begin{array}{rcl}
%dZ_t^{h_0} & = & \bar{\mu}^{h_0}(t,Z_t^{h_0}) \ud t + c \, \ud W_t
%\\ Z_0^{h_0} & = & \ell(h_0),
%\end{array}
%\right.
%\end{align*}
where $c := (c_0,c_1,\ldots,c_m) \in \R^{m+1}$ and the coefficient $\bar{\mu}^{h_0} :\Omega\times \R_+ \times \R^{m+1} \to \R^{m+1}$ has components
\[
\bar{\mu}^{h_0,i}(\omega,t,z) := -c_i(\lambda_t(\omega) + b_t^i(\omega)) + \kappa^{h_0,i}(t) + \delta_i z^i, 
\qquad\text{ for all } i=0,1,\ldots,m,
\]
where, for each $i=0,1,\ldots,m$, the mapping $\kappa^{h_0,i} : \R_+ \to \R$ is given by
\[
\kappa^{h_0,i}(t) = \frac{\ud}{\ud t} h_0^i(t) - \delta_i h_0^i(t) + \frac{c_i^2}{2 \delta_i} \big( e^{2 \delta_i t} - 1 \big), 
\qquad\text{ for all } t \in \R_+.
\]
\end{enumerate}
\end{proposition}
\begin{proof}
Let $T \in L(\R^{m+1},V)$ be the linear isomorphism given by
\[
T y := \big( y^i e^{\delta_i \cdot} \big)_{i=0,1,\ldots,m}.
\]
Using Proposition \ref{prop-real-suff}, the real-world HJMM SPDE \eqref{eq:SPDE} has an $(m+1)$-dimensional affine realization induced by $(\phi,T)$, where the family $\phi = (\phi^{h_0})_{h_0 \in D(\ud / \ud \xi)}$ is given by
\[
\phi^{h_0}(t) = \cS_t h_0 + \int_0^t \cS_{t-s} \alpha_1 \, \ud s, 
\qquad\text{ for  all } t \in \R_+.
\]
Note that $U := \ell T \in L(\R^{m+1})$ is simply given by $U = \Id$. The isomorphism $R := T U^{-1} \in L(\R^{m+1},V)$ then reduces to $R = T$ and, therefore, $R \ell \in L(H,V)$ is given by
\[
R \ell h = \big( h^i(0) e^{\delta_i \cdot} \big)_{i=0,1,\ldots,m}, 
\quad\text{ for } h \in H. 
\]
A straightforward calculation shows that
\[
\int_0^t \cS_{t-s} \alpha_1 \, \ud s = \bigg( \frac{c_i^2}{2 \delta_i^2} \Big( \big( e^{2 \delta_i t} - 1 \big) e^{\delta_i \cdot} - 2 \big( e^{\delta_i t} - 1 \big) \Big) e^{\delta_i \cdot} \bigg)_{i=0,1,\ldots,m}.
\]
Therefore, we have that
\[
R \ell \int_0^t \cS_{t-s} \alpha_1 \, \ud s = \bigg( \frac{c_i^2}{2 \delta_i^2} \big( e^{2 \delta_i t} - 2 e^{\delta_i t} + 1 \big) e^{\delta_i \cdot} \bigg)_{i=0,1,\ldots,m}.
\]
In turn, this implies that the mapping $\varphi^{h_0} : \R_+ \to H$ defined according to \eqref{varphi-trans} is given by \eqref{varphi-Vasicek}. Furthermore, the mapping $\bar{\gamma}^{h_0} : \Omega\times\R_+ \times \R^{m+1} \to \R^{m+1}$ defined according to \eqref{gamma-trans} is given by
\[
\bar{\gamma}^{h_0}(\omega,t,z) = \ell \beta(\omega,t,\varphi^{h_0}(t) + Rz) = \ell \big( c_i e^{\delta_i \cdot} \big)_{i=0,1,\ldots,m} = c
\]
and the mapping $\bar{\mu}^{h_0} :\Omega\times \R_+ \times \R^{m+1} \to \R^{m+1}$ defined according to \eqref{mu-trans} is given by
\begin{align*}
\bar{\mu}^{h_0}(\omega,t,z) &= \ell \nu(\omega,t,\varphi^{h_0}(t) + Rz) = \ell \frac{\ud}{\ud \xi} ( \varphi^{h_0}(t) + Rz ) + \ell \alpha_2(\omega,t,\varphi^{h_0}(t) + Rz)
\\ &= \ell \frac{\ud}{\ud \xi} \varphi^{h_0}(t) + \big( \delta_i z^i \big)_{i=0,1,\ldots,m} + \big( -c_i (\lambda_t(\omega) + b_t^i(\omega)) \big)_{i=0,1,\ldots,m}. 
\end{align*}
Taking into account \eqref{varphi-Vasicek}, we have
\begin{align*}
\frac{\ud}{\ud \xi} \varphi^{h_0}(t) = \bigg( \cS_t \frac{\ud}{\ud \xi} h_0^i - \delta_i h_0^i(t) e^{\delta_i \cdot} + \frac{c_i^2}{2 \delta_i} \big( e^{2 \delta_i t} - 1 \big) \big( 2 e^{2 \delta_i \cdot} - e^{\delta_i \cdot} \big) \bigg)_{i=0,1,\ldots,m},
\end{align*}
and, therefore,
\[
\ell \frac{\ud}{\ud \xi} \varphi^{h_0}(t) = \bigg( \frac{\ud}{\ud t} h_0^i(t) - \delta_i h_0^i(t) + \frac{c_i^2}{2 \delta_i} ( e^{2 \delta_i t} - 1 ) \bigg)_{i=0,1,\ldots,m}.
\]
At this stage, the result follows by an application of Proposition \ref{prop-real-trans}.
\end{proof}

\begin{remark}
In the particular case $m=0$ and $\lambda \equiv 0$, the real-world HJMM SPDE \eqref{eq:SPDE} reduces to the Hull-White extension of the Vasi\v{c}ek model under a  risk-neutral probability. In this case, it is straightforward to check that the result of Proposition \ref{prop-Vasicek} coincides with the result obtained when performing the inversion of the yield curve (see, for instance, \cite[Section 5.4.5]{fil09}).
\end{remark}

Proposition \ref{prop-real-suff} can also be applied to the real-world HJMM SPDE \eqref{eq:SPDE} driven by a general L\'{e}vy process. However, the results of \cite{PT15} indicate that some restrictions will arise on the functions $\psi$ and $c$, already in the one-dimensional case. For simplicity, let $L$ be a one-dimensional L\'{e}vy process of the form $L = W + x * (\mu^L - \nu)$, where $W$ is a Brownian motion and $\nu(\ud t, \ud x) = F(\ud x) \ud t$ is the compensator of $\mu^L$ with L\'{e}vy measure $F$. Similarly as above, we assume that the volatility $\beta$ is constant and given by \eqref{beta-Vasicek}. Furthermore, we assume that $\psi$ and $c$ are deterministic. In view of \eqref{eq:drift_HJM}, the drift term $\alpha$ is then given by
\[
\alpha^i(\omega,t,h) = \beta^i \bar{\beta}^i - \bigl( \lambda_t(\omega) + b_t^i(\omega) \bigr) \beta^i
+ \beta^i \int_{\R} x \big( 1 - e^{-x \bar{\beta}^i} ( 1 + \psi_t(x) ) ( 1 + c_t^i(x) ) \big) F(\ud x),
\]
for all $ i=0,1,\ldots,m$, where $b^0 \equiv 0$ and $c^0 \equiv 0$. Let $V \subset D(\ud / \ud \xi)$ be the $(m+1)$-dimensional $D(\ud / \ud \xi)$-invariant subspace given by \eqref{V-subspace}. Since $\psi$ and $c$ are deterministic, the drift $\alpha$ admits a decomposition of the form \eqref{decomp-alpha} with
\begin{align*}
\alpha_1(t) &= \bigg( \beta^i \bar{\beta}^i + \beta^i \int_{\R} x \big( 1 - e^{-x \bar{\beta}^i} ( 1 + \psi_t(x) ) ( 1 + c_t(x) ) \big) F(\ud x) \bigg)_{i=0,1,\ldots,m},
\\ \alpha_2(\omega,t) &= \big( -( \lambda_t(\omega) + b_t^i(\omega) ) \beta^i \big)_{i=0,1,\ldots,m}.
\end{align*}
In view of Proposition \ref{prop-real-suff}, we can conclude that the real-world HJMM SPDE \eqref{eq:SPDE} admits an $(m+1)$-dimensional affine realization, provided that $\alpha_1$ is a $D(\ud / \ud \xi)$-valued mapping of class $C^1(\R_+;H)$ such that $\frac{\ud}{\ud \xi} \alpha_1 : \R_+ \to H$ is continuous.

\begin{appendix}

\section{Proof of Theorem \ref{thm:FTAP}}\label{app:prop_FTAP}

\begin{proof}
Let us define $\cXfin_1:=\bigcup_{n\geq1}\cX^n_1$, so that $\cX_1=\overline{\cXfin_1}$, with the bar denoting the closure in \'Emery's semimartingale topology. For each  $n\in\N$, we denote by $\cXfin_{1,[0,n]}:=\{X_{\cdot\wedge n}:X\in\cXfin_1\}$ the set of all elements of $\cXfin_1$ stopped at $n$ and we define $\cX_{1,[0,n]}$ analogously. 
Since the set $\cXfin_1$ is stable under stopping and stopping is a continuous operation in the semimartingale topology (see \cite[Remark II.6.c]{Memin80}), it holds that $\overline{\cXfin_{1,[0,n]}}=\cX_{1,[0,n]}$. 
Moreover, it is easy to see that the sets $1+\cXfin_{1,[0,n]}$ satisfy the requirements of \cite[Definition 1.1]{Kardaras13b}, since fork-convexity can be shown analogously to \cite[Lemma 2.1]{CKT16}.
Suppose first that NUPBR holds. Then, for each $n\in\N$, the set $\cX_1(n)$ is bounded in probability and therefore, by \cite[Theorem 1.7]{Kardaras13b}, there exists $\widehat{X}^n\in\cX_{1,[0,n]}$ with $\widehat{X}^n>-1$ such that $(1+X)/(1+\widehat{X}^n)$ is a supermartingale, for every $X\in\cX_{1,[0,n]}$.
For each $n\in\N$, let us denote $D^n:=1/(1+\widehat{X}^n)$. Similarly as in the proof of \cite[Proposition 1]{CCFM17}, we show that the elements $D^n$, $n\in\N$, can be concatenated into a supermartingale deflator. 
For all $t\geq0$, let $n(t):=\min\{n\in\N : n>t\}$ and define the c\`adl\`ag process 
\be\label{eq:defZ}
D_t := \prod_{k=1}^{n(t)}\frac{D^k_{k\wedge t}}{D^k_{(k-1)\wedge t}},
\qquad\text{ for all }t\geq0.
\ee
Observe that this definition implies that, if $t\in(m-1,m]$ for some $m\in\N$, then $D_t=\prod_{k=1}^m(D^k_{k\wedge t}/D^k_{k-1})$.
This follows from \eqref{eq:defZ} by noting that if $t\in(m-1,m)$ then $n(t)=m$, while if $t=m$ then $n(t)=m+1$.
Let $X\in\cX_1$ and $s<t$. 
Suppose first that $t\in(s,n(s)]\subseteq(n(s)-1,n(s)]$. In this case, we can compute
\begin{align*}
\EE\left[D_t(1+X_t)|\cF_s\right]
&= \EE\left[\prod_{k=1}^{n(s)}\frac{D^k_{k\wedge t}}{D^k_{k-1}}(1+X_t)\Bigg|\cF_s\right]	\\
&= \prod_{k=1}^{n(s)-1}\frac{D^k_k}{D^k_{k-1}}\,
\EE\left[\frac{D^{n(s)}_t}{D^{n(s)}_{n(s)-1}}(1+X_t)\Bigg|\cF_s\right]	\\
&\leq \prod_{k=1}^{n(s)-1}\frac{D^k_k}{D^k_{k-1}}
\frac{D^{n(s)}_s(1+X_s)}{{D^{n(s)}_{n(s)-1}}}	
= \prod_{k=1}^{n(s)}\frac{D^k_{k\wedge s}}{D^k_{k-1}}(1+X_s)
= D_s(1+X_s),
\end{align*}
where the inequality follows from the fact that $D^{n(s)}(1+X)$ is a supermartingale on $[0,n(s)]$. 
This shows that the supermartingale property holds for $t\in(s,n(s)]$. 
To prove that $D$ is a supermartingale deflator, it suffices to show that the supermartingale property holds for $t\in(s,n(s)+l]$, for all $l\in\N$. We proceed by induction on $l$. Suppose that the supermartingale property holds for $t\in(s,n(s)+l]$, for some $l \in \mathbb{N}$. By the induction hypothesis, it suffices to prove the property for $t \in (n(s)+l, n(s)+l+1]$.
To this effect, we compute
\begin{align*}
\EE\left[D_t(1+X_t)|\cF_s\right]
&= \EE\left[\EE\left[D_t(1+X_t)|\cF_{n(s)+l}\right]\big|\cF_s\right]	\\
&= \EE\left[D_{n(s)+l}\,\EE\left[\frac{D^{n(s)+l+1}_t}{D^{n(s)+l+1}_{n(s)+l}}(1+X_t)\Bigg|\cF_{n(s)+l}\right]\Bigg|\cF_s\right]	\\
&\leq \EE\left[D_{n(s)+l}(1+X_{n(s)+l})|\cF_s\right]
\leq D_s(1+X_s),
\end{align*}
where the first inequality follows from the supermartingale property of $D^{n(s)+l+1}(1+X)$ on the interval $[0,n(s)+l+1]$ and the second inequality from the induction step. Since $D_0=1$, this completes the proof that the process $D$  defined in \eqref{eq:defZ} is a supermartingale deflator in the sense of Definition \ref{def:deflators}.
We now prove that $1/D\in 1+\cX_1$. To this effect, note that
\[
\widehat{X}_t 
:= \frac{1}{D_t}-1 
= \prod_{k=1}^{n(t)}\frac{1+\widehat{X}^k_{k\wedge t}}{1+\widehat{X}^k_{(k-1)\wedge t}}-1,
\qquad\text{ for all }t\geq0.
\]
Observe first that $\widehat{X}_{\cdot\wedge 1}=\widehat{X}^1\in\cX_{1,[0,1]}\subseteq\cX_1$. Proceeding by induction, suppose that $\widehat{X}_{\cdot\wedge n}\in\cX_{1,[0,n]}$, for some $n\in\N$. The definition of $\widehat{X}$ implies that
\begin{align}
1+\widehat{X}_{t\wedge(n+1)}
&= \ind_{\{t<n\}}(1+\widehat{X}_{t\wedge n})
+ \ind_{\{t\geq n\}}\prod_{k=1}^{n(t)}\frac{1+\widehat{X}^k_{k\wedge t\wedge (n+1)}}{1+\widehat{X}^k_{(k-1)\wedge t\wedge (n+1)}}	\notag\\
&= \ind_{\{t<n\}}(1+\widehat{X}_{t\wedge n})
+ \ind_{\{t\geq n\}}\frac{1+\widehat{X}_n}{1+\widehat{X}^{n+1}_n}(1+\widehat{X}^{n+1}_{t}).
\label{eq:construction_defl}
\end{align}
As argued at the beginning of this proof, the sets $1+\cXfin_{1,[0,n]}$, for $n\in\N$, satisfy the requirements of \cite[Definition 1.1]{Kardaras13b}. Therefore, by the arguments given in \cite[pages 1358--1359]{Kardaras13b}, fork-convexity also holds for their closures $1+\cX_{1,[0,n]}$, for all $n\in\N$.
Since $\widehat{X}_{\cdot\wedge n}\in\cX_{1,[0,n]}\subseteq\cX_{1,[0,n+1]}$ and $\widehat{X}^{n+1}\in\cX_{1,[0,n+1]}$, equation \eqref{eq:construction_defl} and fork-convexity of $1+\cX_{1,[0,n+1]}$ imply that $\widehat{X}_{\cdot\wedge(n+1)}\in\cX_{1,[0,n+1]}\subseteq\cX_1$. 
We have thus shown that $\widehat{X}_{\cdot\wedge n}\in\cX_1$, for all $n\in\N$. Since $\widehat{X}$ is a semimartingale and $\widehat{X}_{\cdot\wedge n}=\ind_{\dbra{0,n}}\cdot\widehat{X}$, for all $n\in\N$, the closedness of $\cX_1$ in the semimartingale topology implies that $\widehat{X}\in\cX_1$.

Conversely, if a process $D$ is a supermartingale deflator and $X\in\cX_1$, then the supermartingale property of $D(1+X)$ together with the fact that $D_0\leq 1$ implies that $0\leq\EE[D_T(1+X_T)]\leq 1$, for all $T\in\R_+$. This implies that the set $D_T\cX_1(T)$ is bounded in probability and, hence, NUPBR holds.

Finally, suppose that $D$ is a local martingale deflator and let $n\in\N$ and $A\in\cA^n$. By definition, $DX^A$ is an $\R^n$-valued local martingale. For every $H\in L(X^A)$ with $X:=H\cdot X^A\geq-1$, the process $DX$ is a local martingale (see, e.g., \cite[Lemma 4.2]{Fontana15}). Since every non-negative local martingale with integrable initial value is a supermartingale, this implies that $D(1+X)$ is a supermartingale, meaning that $D$ is a supermartingale deflator for the set $\bigcup_{n\geq1}\cX^n_1$. By Fatou's lemma, the same property holds for the closure $\cX_1$, thus proving that $D$ is a supermartingale deflator.
\end{proof}

\section{Conditional expectation with respect to a Dol\'eans measure}
\label{app:Doleans}

We recall the notion of conditional expectation with respect to a Dol\'eans measure, as introduced in \cite{Jacod1976} (see also \cite[Section III.3c]{jashi03}).
Let  $\mu$ be an integer-valued random measure on $\R_+\times E$ with compensator $\nu$. 
The positive Dol\'eans measure $M_{\mu}$ on $(\Omega\times\R_+\times E,\cF\otimes\cB(\R_+)\otimes\cB(E))$ is defined by 
\[
M_{\mu}[\varphi]:=\EE\left[\int_0^{\infty}\int_E\varphi_t(x)\mu(\ud t,\ud x)\right],
\qquad\text{for all measurable functions $\varphi:\Omega\times\R_+\times E\to\R_+$.}
\] 
We denote by $M_{\mu}[\varphi|\widetilde{\cP}]$ the conditional expectation relative to $M_{\mu}$ of a measurable function $\varphi$ with respect to the sigma-field $\widetilde{\cP}:=\cP\otimes\cB(E)$. 
The conditional expectation is well-defined for every non-negative measurable function $\varphi$ and can be extended to real-valued measurable functions $\varphi$ as long as the measure $|\varphi_t(\omega,x)|\mu(\omega;\ud t,\ud x)$ is $\widetilde{\cP}$-$\sigma$-finite, which in particular holds if the process $|\varphi|\ast\mu$ is locally integrable.
By definition, $M_{\mu}[\varphi|\widetilde{\cP}]$ is the $M_{\mu}$-a.e. unique $\widetilde{\cP}$-measurable function $U$ such that
\be\label{eq:cond_expec}
M_{\mu}[\varphi V] = M_{\mu}[UV],
\qquad\text{ for all $\widetilde{\cP}$-measurable bounded functions $V$}.
\ee

In the proof of Theorem \ref{thm:HJM} we make use of the following lemma, which can be deduced from \cite{Jacod1976}. We provide a self-contained proof which relies only on standard notions found in \cite{jashi03}.

\begin{lemma}	\label{lem:cond_expec}
Let $\varphi:\Omega\times\R_+\times E\to\R$ be a measurable function. The  increasing process $|\varphi|\ast\mu$ is locally integrable if and only if $M_{\mu}[|\varphi||\widetilde{\cP}]\ast\nu$ is locally integrable. In this case, the  compensator of the finite variation process $\varphi\ast\mu$ is given by $M_{\mu}[\varphi|\widetilde{\cP}]\ast\nu$.
\end{lemma}
\begin{proof}
Suppose there exists a sequence $\{\tau_n\}_{n\in\N}$  of stopping times increasing a.s. to infinity as $n\to+\infty$ such that $\EE[(|\varphi|\ast\mu)_{\tau_n}]<+\infty$, for all $n\in\N$. 
For each $n\in\N$, the function $H_n:\Omega\times\R_+\times E\to\{0,1\}$ defined by $H_n:=\ind_{\dbra{0,\tau_n}\times E}$ is $\widetilde{\cP}$-measurable and bounded. Therefore, property \eqref{eq:cond_expec} implies that
\[
\EE\bigl[(|\varphi|\ast\mu)_{\tau_n}\bigr]
= M_{\mu}\bigl[|\varphi|H_n\bigr]
= M_{\mu}\bigl[M_{\mu}[|\varphi||\widetilde{\cP}]H_n\bigr]
= \EE\bigl[(M_{\mu}[|\varphi||\widetilde{\cP}]\ast\mu)_{\tau_n}\bigr]
= \EE\bigl[(M_{\mu}[|\varphi||\widetilde{\cP}]\ast\nu)_{\tau_n}\bigr],
\]
for all $n\in\N$, where in the last equality we made use of \cite[Theorem II.1.8]{jashi03}. This shows that the process $M_{\mu}[|\varphi||\widetilde{\cP}]\ast\nu$ is locally integrable. The converse implication can be shown in the same way.
To prove the second part of the lemma, by localization we can assume that $\EE[(|\varphi|\ast\mu)_{\infty}]<+\infty$. Let $\tau$ be an arbitrary stopping time. Similarly as above, letting the function $H$ be defined by $H:=\ind_{\dbra{0,\tau}\times E}$,
\[
\EE[(\varphi\ast\mu)_{\tau}]
= M_{\mu}[\varphi H]
= M_{\mu}\bigl[M_{\mu}[\varphi|\widetilde{\cP}] H\bigr]
= \EE\bigl[(M_{\mu}[\varphi|\widetilde{\cP}]\ast\mu)_{\tau}\bigr]
= \EE\bigl[(M_{\mu}[\varphi|\widetilde{\cP}]\ast\nu)_{\tau}\bigr],
\]
thus implying that the process $\varphi\ast\mu-M_{\mu}[\varphi|\widetilde{\cP}]\ast\nu$ is a martingale. In view of \cite[Theorem I.3.18]{jashi03}, this suffices to deduce that $M_{\mu}[\varphi|\widetilde{\cP}]\ast\nu$ is the compensator of $\varphi\ast\mu$.
\end{proof}

\section{Locally Lipschitz and locally bounded functions}\label{app-functions}

In this appendix, we collect some technical results on locally Lipschitz and locally bounded functions that are used in Section \ref{sec:SPDE}. In the following, we denote by $X,Y,Z$ some generic normed spaces. Moreover, we call $(X,m)$ a commutative algebra if $m : X \times X \to X$ is a continuous symmetric bilinear operator.
We denote by $\Lip(X,Y)$ the space of all Lipschitz continuous functions from $X$ to $Y$ and by $\B(X,Y)$ the space of all bounded functions from $X$ to $Y$.

\begin{definition}	\label{def:lcl_Lip}
A function $f : X \to Y$ is said to be \emph{locally Lipschitz} if there exists a function $L_f : \R_+ \to \R_+$ such that, for every $r \in \R_+$, we have
\[
\| f(x_1) - f(x_2) \| \leq L_f(r) \| x_1 - x_2 \|, 
\qquad \text{ for all $x_1,x_2 \in X$ with $\| x_1 \|\vee \| x_2 \| \leq r$.}
\]
We denote by $\Lip^{\loc}(X,Y)$ the space of all locally Lipschitz functions $f : X \to Y$.
\end{definition}

We call the function $L_f$ appearing in Definition \ref{def:lcl_Lip} a \emph{Lipschitz function} of $f$. Without loss of generality, we can assume that $L_f$ is increasing.
If the function $L_f$ is bounded, then it can be chosen constant and in this case the function $f$ is Lipschitz continuous.

\begin{definition}	\label{def:bdd}
A function $f : X \to Y$ is said to be \emph{locally bounded} if there exists a function $B_f : \R_+ \to \R_+$ such that, for every $r \in \R_+$, we have
\[
\| f(x) \| \leq B_f(r), 
\qquad \text{ for all $x \in X$ with $\| x \| \leq r$.}
\]
We denote by $\B^{\loc}(X,Y)$ the space of all locally bounded functions $f : X \to Y$.
\end{definition}

We call the function $B_f$ appearing in Definition \ref{def:bdd} a \emph{boundedness function} of $f$. Without loss of generality, we can assume that $B_f$ is increasing.
If the function $B_f$ is bounded, then it can be chosen constant and in this case the function $f$ is bounded.

\begin{definition}
A function $f : X \to Y$ is said to satisfy the \emph{linear growth condition} if there exists a constant $C \in \R_+$ such that
\[
\| f(x) \| \leq C (1 + \| x \|), 
\qquad \text{for all $x \in X$.}
\]
We denote by $\LG(X,Y)$ the space of all functions $f : X \to Y$ satisfying the linear growth condition. 
\end{definition}

The following lemma recalls a well-known property of locally Lipschitz functions.

\begin{lemma}\label{lemma-loc-Lip-loc-bdd}
It holds that $\Lip^{\loc}(X,Y) \subset \B^{\loc}(X,Y)$.
\end{lemma}
%\begin{proof}
%Let $f \in \Lip^{\loc}(X,Y)$ be arbitrary, and let $L_f : \R_+ \to \R_+$ be a Lipschitz function of $f$. We define $B_f : \R_+ \to \R_+$ as
%\begin{align*}
%B_f(r) := r L_f(r) + \| f(0) \|, \quad r \in \R_+.
%\end{align*}
%Let $r \in \R_+$ be arbitrary. Then, for each $x \in X$ with $\| x \| \leq r$ we obtain
%\begin{align*}
%\| f(x) \| \leq \| f(x) - f(0) \| + \| f(0) \| \leq L_f(r) \| x \| + \| f(0) \| \leq r L_f(r) + \| f(0) \| = B_f(r),
%\end{align*}
%showing that $B_f$ is a boundedness function of $f$.
%\end{proof}

\begin{lemma}\label{lemma-Lip-loc-algebra}
Let $(Y,m)$ be a commutative algebra. Let $f,g \in \Lip^{\loc}(X,Y)$ be arbitrary and denote by $fg = f \cdot g : X \to Y$ the product $(fg)(x) = m(f(x),g(x))$, for $x \in X$. Then, the following hold:
\begin{enumerate}
\item $fg \in \Lip^{\loc}(X,Y)$;
\item let $L_f,L_g,B_f,B_g : \R_+ \to \R_+$ be Lipschitz and boundedness functions of $f$ and $g$. Then, Lipschitz and boundedness functions of the product $fg$ are given by, for all $r\in\R_+$,
\begin{align*}
L_{fg}(r) &= \| m \| (L_f(r) B_g(r) + L_g(r) B_f(r)),
\\ B_{fg}(r) &= \| m \| B_f(r) B_g(r).
\end{align*}
\end{enumerate}
\end{lemma}
\begin{proof}
Let $r \in \R_+$ be arbitrary. Then, for all $x_1,x_2 \in X$ with $\| x_1 \|\vee \| x_2 \| \leq r$, we have that
\begin{align*}
&\| f(x_1) g(x_1) - f(x_2) g(x_2) \| \leq \| f(x_1) (g(x_1) - g(x_2)) \| + \| (f(x_1) - f(x_2)) g(x_2) \|
\\ &\leq \| m \| \, \| f(x_1) \| \, \| g(x_1) - g(x_2) \| + \| m \| \, \| f(x_1) - f(x_2) \| \, \| g(x_2) \|
\\ &\leq \| m \| (B_f(r) L_g(r) + L_f(r) B_g(r)) \| x_1 - x_2 \|.
\end{align*}
Furthermore, for all $x \in X$ with $\| x \| \leq r$, we have that
\[
\| f(x) g(x) \| \leq \| m \| \| f(x) \| \| g(x) \| \leq \| m \| B_f(r) B_g(r).
\]
\end{proof}

\begin{lemma}\label{lemma-Lip-loc-comp}
Let $f \in \Lip^{\loc}(X,Y)$ and $g \in \Lip^{\loc}(Y,Z)$ be arbitrary. Then the following hold:
\begin{enumerate}
\item $g \circ f \in \Lip^{\loc}(X,Z)$;
\item let $L_f,L_g,B_f,B_g : \R_+ \to \R_+$ be Lipschitz and boundedness functions of $f$ and $g$. Then, Lipschitz and boundedness functions of the composition $g \circ f$ are given by, for all $r\in\R_+$,
\begin{align*}
L_{g \circ f}(r) &= L_f(r) L_g(B_f(r)), 
\\ B_{g \circ f}(r) &= B_g(B_f(r)).
\end{align*}
\end{enumerate}
\end{lemma}
\begin{proof}
Let $r \in \R_+$. For all $x_1,x_2 \in X$ with $\| x_1 \|\vee \| x_2 \| \leq r$, we have $\| f(x_1) \|\vee \| f(x_2) \| \leq B_f(r)$. Therefore,
\[
\| g(f(x_1)) - g(f(x_2)) \| \leq L_g(B_f(r)) \| f(x_1) - f(x_2) \| \leq L_g(B_f(r)) L_f(r) \| x_1 - x_2 \|.
\]
Furthermore, for all $x \in X$ with $\| x \| \leq r$, we have $\| f(x) \| \leq B_f(r)$ and, hence, $\| g(f(x)) \| \leq B_g(B_f(r))$.
\end{proof}

\begin{lemma}\label{lemma-Bochner-Lip}
Let $(E,\cE,\mu)$ be a measure space, $Y$ a separable Banach space and $f : X \times E \to Y$ a $\cB(X) \otimes \cE$-measurable function. Suppose that the following conditions are satisfied:
\begin{enumerate}
\item $f(\cdot,z) \in \Lip^{\loc}(X,Y)$, for every $z \in E$;
\item for every $z \in E$, there exists a Lipschitz function $L_{f(\cdot,z)} : \R_+ \to \R_+$ of $f(\cdot,z)$ such that $z \mapsto L_{f(\cdot,z)}(r)$ belongs to $L^1(\mu)$, for every $r \in \R_+$;
\item for every $z \in E$, there exists a boundedness function $B_{f(\cdot,z)} : \R_+ \to \R_+$ of $f(\cdot,z)$ such that $z \mapsto B_{f(\cdot,z)}(r)$ belongs to $L^1(\mu)$, for every $r \in \R_+$.
\end{enumerate}
Then, the following hold:
\begin{enumerate}
\item the Bochner integrals
\[
g(x) := \int_E f(x,z) \mu(\ud z), 
\qquad \text{ for }x \in X,
\]
provide a well-defined function $g \in \Lip^{\loc}(X,Y)$;
\item Lipschitz and boundedness functions of $g$ are given by, for all $r\in\R_+$,
\begin{align*}
L_g(r) &= \int_E L_{f(\cdot,z)}(r) \mu(\ud z), 
\\ B_g(r) &= \int_E B_{f(\cdot,z)}(r) \mu(\ud z).
\end{align*}
\end{enumerate}
\end{lemma}
\begin{proof}
Let $r \in \R_+$. For each $x \in X$ with $\| x \| \leq r$, we have that
\[
\| g(x) \| \leq \int_E \| f(x,z) \| \mu(\ud z) \leq \int_E B_{f(\cdot,z)}(r) \mu(\ud z).
\]
Furthermore, for all $x_1,x_2 \in X$ with $\| x_1 \|\vee\| x_2 \| \leq r$, we have that
\[
\| g(x_1) - g(x_2) \| \leq \int_E \| f(x_1,z) - f(x_2,z) \| \mu(\ud z) \leq \bigg( \int_E L_{f(\cdot,z)}(r) \mu(\ud y) \bigg) \| x_1 - x_2 \|.
\]
\end{proof}

\section{Properties of multi-dimensional Filipovi\'{c} spaces}\label{app-space}

In this appendix, we collect some technical results on multi-dimensional Filipovi\'{c} spaces which are needed for the SPDE analysis of Section \ref{sec:SPDE}. Building on the previous results of \cite{F:01} and \cite{Tappe-Wiener}, we extend those results by considering locally Lipschitz and locally bounded functions in a multi-dimensional setting. 
We start by recalling that, for any $\rho > 0$, the Filipovi\'{c} space $H_{\rho}$ is the space of all absolutely continuous functions $h : \mathbb{R}_+ \rightarrow \mathbb{R}$ such that
\[
\| h \|_{\rho} := \bigg( |h(0)|^2 + \int_0^{\infty} |h'(x)|^2 e^{\rho x} \ud x \bigg)^{1/2} < +\infty.
\]

\begin{theorem}\cite[Section 5]{F:01}\label{thm-space-forward}
The following statements are true:
\begin{enumerate}
\item $(H_{\rho}, \| \cdot \|_{\rho})$ is a separable Hilbert space;
\item for each $x \in \R_+$, the point evaluation $h \mapsto h(x) : H_{\rho} \to \R$ is a continuous linear functional;
\item the translation semigroup $(\cS_t)_{t \geq 0}$ is a $C_0$-semigroup on $H_{\rho}$;
\item its generator $A$ is given by $Ah = h'$, for all $h \in \cD(A)$, and the domain is
\[
\cD(A) = \{ h \in H : h' \in H \};
\]
\item for each $h \in H_{\rho}$, the limit $h(\infty) := \lim_{x \to \infty} h(x)$ exists; 
\item $H_{\rho}^0 := \{ h \in H_{\rho} : h(\infty) = 0 \}$ is a closed subspace of $H_{\rho}$.
\end{enumerate}
\end{theorem}

In the following, we fix two constants $\rho$ and $\rho'$ such that $0 < \rho < \rho'$.

\begin{lemma}\label{lemma-H-bdd}
It holds that $H_{\rho} \subset L^{\infty}(\R_+)$ and the embedding operator
\[
\Id : (H_{\rho},\| \cdot \|_{\rho}) \to (L^{\infty}(\R_+),\| \cdot \|_{\infty}) 
\]
is a bounded linear operator with $\| \Id \| \leq C_{\rho}$, where
\begin{align}\label{C-rho}
C_{\rho} := 1 + \frac{1}{\sqrt{\rho}}.
\end{align}
\end{lemma}

\begin{proof}
Let $w : \R_+ \to [1,\infty)$ be the weight function given by $w(x) = e^{\rho x}$, for $x \in \R_+$. By inequality (5.4) in \cite{F:01} it holds that $H_{\rho} \subset L^{\infty}(\R_+)$ and $\| \Id \| \leq 1 + C_1$, where the constant $C_1 > 0$ is given by
\[
C_1 = \| w^{-1} \|_{L^1(\R_+)}^{1/2} = \bigg( \int_0^{\infty} e^{-\rho x} \ud x \bigg)^{1/2} = \frac{1}{\sqrt{\rho}}.
\]
\end{proof}

\begin{lemma}\cite[Lemma 4.2]{Tappe-Wiener}\label{lemma-H-algebra}
The pair $(H_{\rho},m)$, where $m : H_{\rho} \times H_{\rho} \to H_{\rho}$ denotes the pointwise multiplication $m(h,g) := hg = h \cdot g$, is a commutative algebra. Furthermore, $m(H_{\rho}^0 \times H_{\rho}) = H_{\rho}^0$.
\end{lemma}

\begin{lemma}\cite[Theorem 4.1]{Tappe-Wiener}\label{lemma-emb-H-H}
It holds that $H_{\rho'} \subset H_{\rho}$ and the embedding operator
\[
\Id : (H_{\rho'},\| \cdot \|_{\rho'}) \to (H_{\rho},\| \cdot \|_{\rho})
\]
is a bounded linear operator with $\| \Id \| \leq 1$.
\end{lemma}

Let $H_{\rho}^1 := \{ h \in H_{\rho} : h(0) = 0 \}$, which is a closed subspace of $H_{\rho}$, because the point evaluation at zero is a continuous linear functional. Moreover, we define the integral operator $\cI$ by $\cI h := \int_0^{\cdot} h(\eta) \ud \eta$.

\begin{lemma}\cite[Lemma 4.3]{Tappe-Wiener}\label{lemma-int-operator}
It holds that $\cI \in L(H_{\rho'}^0,H_{\rho}^1)$ with $\| \cI \| \leq C_{\rho,\rho'}$, where
\begin{align}\label{C-rho-rho}
C_{\rho,\rho'} := \sqrt{ \frac{1}{\rho' (\rho' - \rho)} }.
\end{align}
\end{lemma}

We then consider the mapping $\mathbf{S}$ given by $\mathbf{S} h := h \cdot \cI h$.

\begin{lemma}\cite[Corollary 5.1.2]{F:01}\label{lemma-mapping-S}
It holds that $\mathbf{S} \in \Lip^{\loc}(H_{\rho}^0,H_{\rho}^0)$ and there exists a constant $C > 0$ such that $\| \mathbf{S} h \|_{\rho} \leq C \| h \|_{\rho}^2$, for all $h \in H_{\rho}^0$.
\end{lemma}

\begin{lemma}\label{lemma-composition-in-H}
Let $\varphi : \R \to \R$ be a function of class $C^1$. Then, for every $h \in H_{\rho}$, we have $\varphi \circ h \in H_{\rho}$.
In particular, for every $h \in H_{\rho}$, we have $\exp(h) \in H_{\rho}$.
\end{lemma}
\begin{proof}
By Lemma \ref{lemma-H-bdd} the function $h$ is bounded. Furthermore, the function $\varphi$ is locally Lipschitz, and hence $\varphi \circ h$ is absolutely continuous. Since $h$ is bounded, there exists a constant $M \geq 0$ such that $|\varphi'(h(x))| \leq M$, for all $x \in \R_+$. Therefore, we obtain
\[
\| \varphi \circ h \|_{\rho}^2 = |\varphi(h(0))|^2 + \int_{\R_+} | \varphi'(h(x)) h'(x) |^2 e^{\rho x} \ud x
\leq |\varphi(h(0))|^2 + M^2 \| h \|_{\rho}^2 < +\infty, 
\]
thus proving that $\varphi \circ h \in H_{\rho}$.
\end{proof}

\begin{lemma}\label{lemma-exp}
Let $\lambda > 0$ and define the function $e_{\lambda} : H_{\rho}^1 \to H_{\rho}$ by
\[
e_{\lambda}(h) := 1 - \lambda \exp(h), 
\quad\text{ for all } h \in H_{\rho}^1.
\]
Then, the following hold:
\begin{enumerate}
\item $e_{\lambda} \in \Lip^{\loc}(H_{\rho}^1, H_{\rho})$;
\item there exists a constant $K > 0$, not depending on $\lambda$, such that Lipschitz and boundedness functions of $e_{\lambda}$ are given by, for all $r\in\R_+$,
\begin{align*}
L_{e_{\lambda}}(r) &= K \lambda (1+r) \exp(C_{\rho}r), 
\\ B_{e_{\lambda}}(r) &= 1 + \lambda + \lambda r \exp(C_{\rho} r).
\end{align*}
\end{enumerate}
\end{lemma}
\begin{proof}
By Lemma \ref{lemma-composition-in-H}, the mapping $e_{\lambda} : H_{\rho}^1 \to H_{\rho}$ is well-defined. Let $r \in \R_+$ and $h,g \in H_{\rho}^1$ with $\| h \|_{\rho}\vee \| g \|_{\rho} \leq r$. Using Lemma \ref{lemma-H-bdd} we obtain
\begin{align*}
\| \exp(h) - \exp(g) \|_{\rho}^2 &= \int_0^{\infty} | h'(x) \exp(h(x)) - g'(x) \exp(g(x)) |^2 e^{\rho x} \ud x
\\ &\leq 2 \int_0^{\infty} | h'(x) ( \exp(h(x)) - \exp(g(x)) ) |^2 e^{\rho x} \ud x
\\ &\quad + 2 \int_0^{\infty} | ( h'(x) - g'(x) ) \exp(g(x)) |^2 e^{\rho x} \ud x
\\ &\leq 2 \int_0^{\infty} | \exp(C_{\rho} r) (h(x) - g(x)) |^2 |h'(x)|^2 e^{\rho x} \ud x
\\ &\quad + 2 \int_0^{\infty} | ( h'(x) - g'(x) ) \exp(C_{\rho} r) |^2 e^{\rho x} \ud x
\\ &\leq 2 \exp(C_{\rho} r)^2 C_{\rho}^2 r^2 \| h-g \|_{\rho}^2 + 2 \exp(C_{\rho} r)^2 \| h-g \|_{\rho}^2,
\end{align*}
and, therefore,
\[
\| \exp(h) - \exp(g) \|_{\rho} \leq \sqrt{2} (C_{\rho} r + 1) \exp(C_{\rho} r) \| h-g \|_{\rho}.
\]
Let $h \in H_{\rho}^1$ with $\| h \|_{\rho} \leq r$, for some $r\in\R_+$. By Lemma \ref{lemma-H-bdd} we have that
\begin{align*}
\| \exp(h) \|_{\rho}^2 &= 1 + \int_0^{\infty} | \exp(h(x)) h'(x) |^2 e^{\rho x} \ud x
\\ &\leq 1 + \exp(C_{\rho} r)^2 \int_0^{\infty} | h'(x) |^2 e^{\rho x} \ud x \leq 1 + r^2 \exp(C_{\rho} r)^2,
\end{align*}
and, hence,
\[
\| \exp(h) \|_{\rho} \leq 1 + r \exp(C_{\rho} r).
\]
Therefore, it follows that
\[
\| 1 - \lambda \exp(h) \|_{\rho} \leq 1 + \lambda(1 + r \exp(C_{\rho} r)) = 1 + \lambda + \lambda r \exp(C_{\rho} r).
\]
\end{proof}

For what follows, let $X$ be a normed space.

\begin{proposition}\label{prop-beta}
Let $\beta \in \Lip^{\loc}(X, H_{\rho}^0)$ be such that, for some constant $K > 0$, we have
\[
\| \beta(h) \|_{\rho} \leq K \sqrt{1 + \| h \|_X}, \quad \text{ for all }h \in X.
\]
Then, the product $\alpha := \mathbf{S} \circ \beta = \beta \cdot \cI \beta$ satisfies $\alpha \in \Lip^{\loc}(X, H_{\rho}^0) \cap \LG(X, H_{\rho}^0)$.
\end{proposition}
\begin{proof}
The result follows as a direct consequence of Lemma \ref{lemma-Lip-loc-comp} and Lemma \ref{lemma-mapping-S}.
\end{proof}

Let us now introduce the constant
\[
K_{\rho,\rho'} := C_{\rho} C_{\rho,\rho'} = \bigg( 1 + \frac{1}{\sqrt{\rho}} \bigg) \sqrt{\frac{1}{\rho'(\rho' - \rho)}},
\]
where we recall that $C_{\rho}$ and $C_{\rho,\rho'}$ are given by \eqref{C-rho} and \eqref{C-rho-rho}, respectively. For a constant $K > 0$, we introduce the strictly increasing function
\[
V_K : \R_+ \to \R_+, \quad V_K(r) := r (1+r) \exp(Kr),
\]
and we denote by $W_K : \R_+ \to \R_+$ its inverse. Note that $V_K(r) \geq r$ for all $r \in \R_+$, which implies $W_K(r) \leq r$ for all $r \in \R_+$.

\begin{proposition}\label{prop-gamma-1}
Let $\gamma \in \Lip^{\loc}(X, H_{\rho'}^0)$ and $\lambda > 0$. We define the product 
\[
\alpha_{\lambda} := \gamma  \bigl(1 - \lambda \exp(-\cI \gamma)\bigr).
\] 
Then, the following hold:
\begin{enumerate}
\item $\alpha_{\lambda} \in \Lip^{\loc}(X, H_{\rho}^0)$;
\item there exists a constant $K > 0$ such that Lipschitz and boundedness functions of $\alpha_{\lambda}$ are given by
\begin{align*}
L_{\alpha_{\lambda}}(r) &= K L_{\gamma}(r) \big( (1 + \lambda) + \lambda V_{K_{\rho,\rho'}}(B_{\gamma}(r)) \big),
\\ B_{\alpha_{\lambda}}(r) &= K \big( B_{\gamma}(r) + \lambda V_{K_{\rho,\rho'}}(B_{\gamma}(r)) \big),
\end{align*}
for all $r\in\R_+$,
where $L_{\gamma}$ and $B_{\gamma}$ are any Lipschitz and boundedness functions of $\gamma$.
\end{enumerate}
\end{proposition}
\begin{proof}
By Lemma \ref{lemma-int-operator} we have that $\cI \in \Lip(H_{\rho'}^0,H_{\rho}^1)$, a Lipschitz constant is given by $L_{\cI} = C_{\rho,\rho'}$ and a boundedness function is given by $B_{\cI}(r) = C_{\rho,\rho'} r$, for all $r \in \R_+$. Let us define $\Gamma := -\cI \circ \gamma$. By Lemma \ref{lemma-Lip-loc-comp}, we have that $\Gamma \in \Lip^{\loc}(X,H_{\rho}^1)$ and Lipschitz and boundedness functions are given by
\begin{align*}
L_{\Gamma}(r) &= L_{\gamma}(r) L_{\cI} = C_{\rho,\rho'} L_{\gamma}(r),
\\ B_{\Gamma}(r) &= B_{\cI}(B_{\gamma}(r)) = C_{\rho,\rho'} B_{\gamma}(r),
\end{align*}
for all $r\in\R_+$. Let us then define $e_{\lambda} : H_{\rho}^1 \to H_{\rho}$ by
\[
e_{\lambda}(h) := 1 - \lambda \exp(h), 
\qquad\text{ for all } h \in H_{\rho}^1.
\]
By Lemma \ref{lemma-exp}, it holds that $e_{\lambda} \in \Lip^{\loc}(H_{\rho}^1, H_{\rho})$ and there exists a constant $K_1 > 0$, not depending on $\lambda$, such that Lipschitz and boundedness functions of $e_{\lambda}$ are given by
\begin{align*}
L_{e_{\lambda}}(r) &= K_1 \lambda (1+r) \exp(C_{\rho}r), 
\\ B_{e_{\lambda}}(r) &= 1 + \lambda + \lambda r \exp(C_{\rho} r).
\end{align*}
By Lemma \ref{lemma-Lip-loc-comp}, we have $e_{\lambda} \circ \Gamma \in \Lip^{\loc}(X,H_{\rho})$ and Lipschitz and boundedness functions are given by
\begin{align*}
L_{e_{\lambda} \circ \Gamma}(r) &= L_{\Gamma}(r) L_{e_{\lambda}}(B_{\Gamma}(r)) = K_1 C_{\rho,\rho'} L_{\gamma}(r) \lambda (1+B_{\Gamma}(r)) \exp(C_{\rho} B_{\Gamma}(r))
\\ &= K_1 C_{\rho,\rho'} L_{\gamma}(r) \lambda (1+C_{\rho,\rho'} B_{\gamma}(r)) \exp(K_{\rho,\rho'} B_{\gamma}(r)),
\\ B_{e_{\lambda} \circ \Gamma}(r) &= B_{e_{\lambda}}(B_{\Gamma}(r)) = 1 + \lambda + \lambda B_{\Gamma}(r) \exp(C_{\rho} B_{\Gamma}(r))
\\ &= 1 + \lambda + \lambda C_{\rho,\rho'} B_{\gamma}(r) \exp(K_{\rho,\rho'} B_{\gamma}(r)),
\end{align*}
for all $r\in\R_+$.
Therefore, there exists a constant $K_2 > 0$, only depending on $\rho$ and $\rho'$, such that Lipschitz and boundedness functions of $e_{\lambda} \circ \Gamma$ are given by
\begin{align*}
L_{e_{\lambda} \circ \Gamma}(r) &= K_2 \lambda L_{\gamma}(r) (1+B_{\gamma}(r)) \exp(K_{\rho,\rho'} B_{\gamma}(r)),
\\ B_{e_{\lambda} \circ \Gamma}(r) &= K_2 \big( 1 + \lambda + \lambda B_{\gamma}(r) \exp(K_{\rho,\rho'} B_{\gamma}(r)) \big).
\end{align*}
Combining Lemma \ref{lemma-Lip-loc-algebra} and Lemmata \ref{lemma-H-algebra}, \ref{lemma-emb-H-H}, it follows that $\alpha_{\lambda} = \gamma \cdot (e_{\lambda} \circ \Gamma) \in \Lip^{\loc}(X, H_{\rho}^0)$ and that Lipschitz and boundedness functions are given by, for all $r\in\R_+$,
\begin{align*}
L_{\alpha_{\lambda}}(r) &= \| m \| ( L_{\gamma}(r) B_{e_{\lambda} \circ \Gamma}(r) + B_{\gamma}(r) L_{e_{\lambda} \circ \Gamma}(r) )
\\ &= \| m \| \big( L_{\gamma}(r) K_2 ( 1 + \lambda + \lambda B_{\gamma}(r) \exp(K_{\rho,\rho'} B_{\gamma}(r)) )
\\ &\quad + B_{\gamma}(r) K_2 \lambda L_{\gamma}(r) (1+B_{\gamma}(r)) \exp(K_{\rho,\rho'} B_{\gamma}(r)) \big)
\\ &= \| m \| L_{\gamma}(r) K_2 \big( (1 + \lambda) + ( \lambda B_{\gamma}(r) + \lambda B_{\gamma}(r) ( 1 + B_{\gamma}(r)) ) \exp(K_{\rho,\rho'} B_{\gamma}(r)) \big)
\\ &\leq \| m \| L_{\gamma}(r) K_2 \big( (1 + \lambda) + 2 \lambda V_{K_{\rho,\rho'}}(B_{\gamma}(r)) \big),
\\ B_{\alpha_{\lambda}}(r) &= \| m \| B_{\gamma}(r) B_{e_{\lambda} \circ \Gamma}(r)
\\ &= \| m \| B_{\gamma}(r) K_2 \big( 1 + \lambda + \lambda B_{\gamma}(r) \exp(K_{\rho,\rho'} B_{\gamma}(r)) \big).
\end{align*}
\end{proof}

\begin{proposition}\label{prop-gamma}
Let $(E,\cE,F)$ be a measure space and $(Z,\cZ)$ a measurable space. Let $\gamma : X \times E \to H_{\rho'}^0$ be a $\cB(X) \otimes \cE$-measurable function and $\lambda : Z \times E \to (0,\infty)$ a $\cZ \otimes \cE$-measurable function. Suppose that there exist a nonnegative function $\kappa \in L^1(F) \cap L^2(F) \cap L^3(F)$, an increasing function $M : \R_+ \to \R_+$, a constant $N \in \R_+$ and a function $\Lambda : Z \to \R_+$ such that the following conditions are satisfied:
\begin{enumerate}
\item $\gamma(\cdot,x) \in \Lip^{\loc}(X, H_{\rho'}^0)$, for every $x \in E$;
\item for every $x \in E$, there exists a Lipschitz function $L_{\gamma(\cdot,x)}$ of $\gamma(\cdot,x)$ such that
\begin{align}\label{L-gamma-kappa-M}
L_{\gamma(\cdot,x)}(r) \leq \kappa(x) M(r), 
\qquad \text{ for all }r \in \R_+;
\end{align}
\item for every $x \in E$, there exists a boundedness function $B_{\gamma(\cdot,x)}$ of $\gamma(\cdot,x)$ such that
\begin{align}\label{B-w-est}
B_{\gamma(\cdot,x)}(r) \leq W_{K_{\rho,\rho'}}(N \kappa(x)(1+r)),
\qquad \text{ for all }r \in \R_+;
\end{align}
\item it holds that
\begin{align}\label{lambda-prod}
|\lambda(z,x)| \leq \Lambda(z) \kappa(x), 
\qquad \text{ for all $z \in Z$ and $x \in E$.}
\end{align}
\end{enumerate}
Then, the following hold:
\begin{enumerate}
\item the Bochner integrals
\begin{align}\label{alpha-z-h-def}
\alpha(z,h) := \int_E \gamma(h,x) \cdot \big( 1 - \lambda(z,x) \exp(- \cI \gamma(h,x)) \big) F(\ud x), 
\qquad\text{ for } (z,h) \in Z \times X,
\end{align}
provide a well-defined $\cZ \otimes \cB(X)$-measurable function $\alpha : Z \times X \to H_{\rho}^0$;
\item there exists an increasing function $L_1 : \R_+ \to \R_+$ such that, for all $r \in \R_+$ and $h,g \in X$ with $\| h \|_X\vee \| g \|_X \leq r$, it holds that
\begin{align}\label{alpha-3-1}
\int_E \| \gamma(h,x) - \gamma(g,x) \|_{\rho}^2 F(\ud x) &\leq L_1(r) \| h-g \|_X^2,
\\ \label{alpha-3-2} \| \alpha(z,h) - \alpha(z,g) \|_{\rho} &\leq L_1(r) ( 1 + \Lambda(z) ) \| h-g \|_X, \qquad\text{ for all } z \in Z;
\end{align}
\item there exists a constant $L_2 \in \R_+$ such that, for all $h,g \in X$, it holds that
\begin{align}\label{alpha-3-1-b}
\int_E \| \gamma(h,x) \|_{\rho}^2 F(\ud x) &\leq L_2 (1 + \| h \|_X^2),
\\ \label{alpha-3-3} \| \alpha(z,h) \|_{\rho} &\leq L_2 ( 1 + \Lambda(z) ) (1 + \| h \|_X), 
\qquad\text{ for all } z \in Z.
\end{align}
\end{enumerate}
\end{proposition}
\begin{proof}
Let $r \in \R_+$. By Lemma \ref{lemma-emb-H-H} and \eqref{L-gamma-kappa-M}, for all $h,g \in X$ with $\| h \|_X\vee \| g \|_X \leq r$, it holds that
\begin{align*}
\int_E \| \gamma(h,x) - \gamma(g,x) \|_{\rho}^2 F(\ud x) &\leq \bigg( \int_E L_{\gamma(\cdot,x)}(r)^2 F(\ud x) \bigg) \| h - g \|_X^2
\\ &\leq M(r)^2 \bigg( \int_E \kappa(x)^2 F(\ud x) \bigg) \| h - g \|_X^2,
\end{align*}
thus showing \eqref{alpha-3-1}. Let $h \in X$ be arbitrary and set $r := \| h \|_X$. Using Lemma \ref{lemma-emb-H-H}, estimate \eqref{B-w-est} and the inequality $W_{K_{\rho,\rho'}}(r) \leq r$, we obtain
\begin{align*}
\int_E \| \gamma(h,x) \|_{\rho}^2 F(\ud x) &\leq \int_E B_{\gamma(\cdot,x)}(r)^2 F(\ud x) \leq \int_E W_{K_{\rho,\rho'}}(N \kappa(x) (1 + r))^2 F(\ud x)
\\ &\leq N^2 (1+r)^2 \int_E \kappa(x)^2 F(\ud x) \leq 2 N^2 \int_E \kappa(x)^2 F(\ud x) (1 + \| h \|_X^2),
\end{align*}
which proves \eqref{alpha-3-1-b}. In view of Proposition \ref{prop-gamma-1}, we can define the mapping $\bar{\alpha} : Z \times X \times E \to H_{\rho}^0$ by
\[
\bar{\alpha}(z,h,x) := \gamma(h,x) \cdot \big( 1 - \lambda(z,x) \exp(- \cI \gamma(h,x)) \big), \quad (z,h,x) \in Z \times X \times E.
\]
Note that $\bar{\alpha}$ is $\cZ \otimes \cB(X) \otimes \cE$-measurable, because $\gamma$ is $\cB(X) \otimes \cE$-measurable and $\lambda$ is $\cZ \otimes \cE$-measurable. Let $z \in Z$ and $x \in E$ be arbitrary. Taking into account \eqref{lambda-prod}, Proposition \ref{prop-gamma-1} implies that $\bar{\alpha}(z,\cdot,x) \in \Lip^{\loc}(X, H_{\rho'}^0)$ and there exists a constant $K > 0$ such that Lipschitz and boundedness functions of $\bar{\alpha}(z,\cdot,x)$ are given by
\begin{align*}
L_{\bar{\alpha}(z,\cdot,x)}(r) &= K L_{\gamma(\cdot,x)}(r) \big( 1 + \Lambda(z) \kappa(x) + \Lambda(z) \kappa(x) V_{K_{\rho,\rho'}} (B_{\gamma(\cdot,x)}(r)) \big),
\\ B_{\bar{\alpha}(z,\cdot,x)}(r) &= K \big( B_{\gamma(\cdot,x)}(r) + \Lambda(z) \kappa(x) V_{K_{\rho,\rho'}} (B_{\gamma(\cdot,x)}(r)) \big), 
\end{align*}
for all $r\in\R_+$.
By \eqref{L-gamma-kappa-M}, \eqref{B-w-est} and the relations $V_{K_{\rho,\rho'}}(W_{K_{\rho,\rho'}}(r)) = r$ and $W_{K_{\rho,\rho'}}(r) \leq r$, we obtain
\begin{align*}
L_{\bar{\alpha}(z,\cdot,x)}(r) &\leq K \kappa(x) M(r) \big( 1 + \Lambda(z) \kappa(x) + N \Lambda(z) \kappa(x)^2 (1+r) \big),
\\ B_{\bar{\alpha}(z,\cdot,x)}(r) &\leq K \big( N \kappa(x)(1+r) + N \Lambda(z) \kappa(x)^2 (1+r) \big) = K N \big( \kappa(x) + \Lambda(z) \kappa(x)^2 \big) (1+r),
\end{align*}
for all $r\in\R_+$.
The Bochner integrals \eqref{alpha-z-h-def} are given by
\[
\alpha(z,h) = \int_E \bar{\alpha}(z,h,x) F(\ud x), 
\quad \text{ for all }(z,h) \in Z \times X.
\]
Therefore, by Lemma \ref{lemma-Bochner-Lip} the Bochner integrals \eqref{alpha-z-h-def} provide a well-defined function $\alpha : Z \times X \to H_{\rho'}^0$ and $\alpha(z,\cdot) \in \Lip^{\loc}(X, H_{\rho'}^0)$, for every $z \in Z$. Furthermore, the mapping $\alpha$ is $\cZ \otimes \cB(X)$-measurable, because $\bar{\alpha}$ is $\cZ \otimes \cB(X) \otimes \cE$-measurable. Moreover, by Lemma \ref{lemma-Bochner-Lip}, for all $z \in Z$, $r \in \R_+$ and $h,g \in X$ with $\| h \|_X\vee \| g \|_X \leq r$, it holds that
\begin{align*}
&\| \alpha(z,h) - \alpha(z,g) \|_{\rho} \leq \bigg( \int_E L_{\bar{\alpha}(z,\cdot,x)}(r) F(\ud x) \bigg) \| h-g \|_X
\\ &\leq K M(r) \bigg( \int_E \kappa(x) F(\ud x) + \Lambda(z) \int_E \kappa(x)^2 F(\ud x) + N \Lambda(z) (1+r) \int_E \kappa(x)^3 F(\ud x) \bigg) \| h-g \|_X,
\end{align*}
thus proving \eqref{alpha-3-2}. Finally, let $z \in Z$ and $h \in X$ and set $r := \| h \|_X$. Then, \eqref{alpha-3-3} follows by noting that, as a consequence of Lemma \ref{lemma-Bochner-Lip}, we have
\begin{align*}
\| \alpha(z,h) \|_{\rho} &\leq \bigg( \int_E B_{\bar{\alpha}(z,\cdot,x)}(r) F(\ud x) \bigg)
\\ &\leq K N \bigg( \int_E \kappa(x) F(\ud x) + \Lambda(z) \int_E \kappa(x)^2 F(\ud x) \bigg) (1+\| h \|_X).
\end{align*}
\end{proof}

\end{appendix}

\bibliographystyle{abbrv}

\bibliography{biblio_multicurve_modifications}

\begin{thebibliography}{10}

\bibitem{AminJarrow91}
K.~Amin and R.~Jarrow.
\newblock Pricing foreign currency options under stochastic interest rates.
\newblock {\em Journal of International Money and Finance}, 10:310--329, 1991.

\bibitem{BalintSchweizer20}
D.~A. B\'alint and M.~Schweizer.
\newblock Large financial markets, discounting, and no asymptotic arbitrage.
\newblock {\em Theory of Probability and its Applications}, 65(2):191--223,
  2020.

\bibitem{Barbarin08}
J.~Barbarin.
\newblock {H}eath--{J}arrow--{M}orton modelling of longevity bonds and the risk
  minimization of life insurance portfolios.
\newblock {\em Insurance: Mathematics and Economics}, 43:41--55, 2008.

\bibitem{BenthKoekebakker08}
F.~E. Benth and S.~Koekebakker.
\newblock Stochastic modeling of financial electricity contracts.
\newblock {\em Energy Economics}, 30:1116--1157, 2008.

\bibitem{Benth-Kruehner}
F.~E. Benth and P.~Kr\"{u}hner.
\newblock Representation of infinite-dimensional forward price models in
  commodity markets.
\newblock {\em Communications in Mathematics and Statistics}, 2(1):47--106,
  2014.

\bibitem{BenthKruhnerbook}
F.~E. Benth and P.~Kr\"{u}hner.
\newblock {\em Stochastic Models for Prices Dynamics in Energy and Commodity
  Markets}.
\newblock Springer, Cham, 2023.

\bibitem{Benth_et_al19}
F.~E. Benth, M.~Piccirilli, and T.~Vargiolu.
\newblock Mean-reverting additive energy forward curves in a
  {H}eath-{J}arrow-{M}orton framework.
\newblock {\em Mathematics and Financial Economics}, 13:543--577, 2019.

\bibitem{BRbook}
T.~Bielecki and M.~Rutkowski.
\newblock {\em Credit Risk: Modeling, Valuation and Hedging}.
\newblock Springer, Berlin - Heidelberg, 2002.

\bibitem{bjork2004}
T.~Bj{\"o}rk.
\newblock On the geometry of interest rate models.
\newblock In R.~A. e.~a. Carmona, editor, {\em Paris-Princeton Lectures on
  Mathematical Finance 2003}, pages 133--216. Springer, Berlin - Heidelberg,
  2004.

\bibitem{BDMKR97}
T.~Bj\"{o}rk, G.~Di~Masi, Y.~Kabanov, and W.~J. Runggaldier.
\newblock Towards a general theory of bond markets.
\newblock {\em Finance and Stochastics}, 1(2):141--174, 1997.

\bibitem{bjkaru97}
T.~Bj\"{o}rk, Y.~Kabanov, and W.~J. Runggaldier.
\newblock Bond market structure in the presence of marked point processes.
\newblock {\em Mathematical Finance}, 7(2):211--239, 1997.

\bibitem{Bjoerk-Svensson}
T.~Bj{\"o}rk and L.~Svensson.
\newblock On the existence of finite dimensional realizations for nonlinear
  forward rate models.
\newblock {\em Mathematical Finance}, 11(2):205--243, 2001.

\bibitem{BLNSP10}
N.~Bruti-Liberati, C.~Nikitopoulos-Sklibosios, and E.~Platen.
\newblock Real-world jump-diffusion term structure models.
\newblock {\em Quantitative Finance}, 10(1):23--37, 2010.

\bibitem{CCFM17}
H.~Chau, A.~Cosso, C.~Fontana, and O.~Mostovyi.
\newblock Optimal investment with intermediate consumption under no unbounded
  profit with bounded risk.
\newblock {\em Journal of Applied Probability}, 54:710--719, 2017.

\bibitem{ChoulliSchweizer16}
T.~Choulli and M.~Schweizer.
\newblock Locally $\phi$-integrable $\sigma$-martingale densities for general
  semimartingales.
\newblock {\em Stochastics}, 88(2):191--266, 2016.

\bibitem{CFG:16}
C.~Cuchiero, C.~Fontana, and A.~Gnoatto.
\newblock A general {HJM} framework for multiple yield curve modeling.
\newblock {\em Finance and Stochastics}, 20(2):267--320, 2016.

\bibitem{CFGaffine}
C.~Cuchiero, C.~Fontana, and A.~Gnoatto.
\newblock Affine multiple yield curve models.
\newblock {\em Mathematical Finance}, 29(2):568--611, 2019.

\bibitem{CKT16}
C.~Cuchiero, I.~Klein, and J.~Teichmann.
\newblock A new perspective on the fundamental theorem of asset pricing for
  large financial markets.
\newblock {\em Theory of Probability and its Applications}, 60(4):561--579,
  2016.

\bibitem{DeDonnoPratelli05}
M.~De~Donno and A.~Pratelli.
\newblock A theory of stochastic integration for bond markets.
\newblock {\em Annals of Applied Probability}, 15(4):2773--2791, 2005.

\bibitem{DeDonnoPratelli06}
M.~De~Donno and A.~Pratelli.
\newblock Stochastic integration with respect to a sequence of semimartingales.
\newblock In M.~\'Emery and M.~Yor, editors, {\em In {M}emoriam
  {Paul}-{A}ndr\'e {M}eyer - S\'eminaire de Probabilit\'es {XXXIX}}, pages 119
  --135. Springer, Berlin - Heidelberg, 2006.

\bibitem{FernholzKaratzas09}
R.~Fernholz and I.~Karatzas.
\newblock Stochastic portfolio theory: an overview.
\newblock In A.~Bensoussan and Q.~Zhang, editors, {\em Mathematical Modeling
  and Numerical Methods in Finance}, volume~XV of {\em Handbook of Numerical
  Analysis}, pages 89--167. North-Holland, Oxford, 2009.

\bibitem{F:01}
D.~Filipovi{\'c}.
\newblock {\em Consistency {P}roblems for {H}eath-{J}arrow-{M}orton {I}nterest
  {R}ate {M}odels}.
\newblock Springer, Berlin, 2001.

\bibitem{fil09}
D.~Filipovi\'c.
\newblock {\em Term-Structure Models. A Graduate Course}.
\newblock Springer, Berlin - Heidelberg, 2009.

\bibitem{FTT:10b}
D.~Filipovi\'c, S.~Tappe, and J.~Teichmann.
\newblock Jump-diffusions in {H}ilbert spaces: existence, stability and
  numerics.
\newblock {\em Stochastics}, 82(5):475--820, 2010.

\bibitem{fitate2010}
D.~Filipovi\'c, S.~Tappe, and J.~Teichmann.
\newblock Term structure models driven by {W}iener processes and {P}oisson
  measures: existence and positivity.
\newblock {\em SIAM Journal on Financial Mathematics}, 1(1):523--554, 2010.

\bibitem{Filipovic-Teichmann}
D.~Filipovi\'c and J.~Teichmann.
\newblock Existence of invariant manifolds for stochastic equations in infinite
  dimension.
\newblock {\em Journal of Functional Analysis}, 197(2):398--432, 2003.

\bibitem{FollmerSchweizer91}
H.~F\"ollmer and M.~Schweizer.
\newblock Hedging of contingent claims under incomplete information.
\newblock In M.~H.~A. Davis and R.~J. Elliott, editors, {\em Applied Stochastic
  Analysis}, pages 389--414. Gordon and Breach, London - New York, 1991.

\bibitem{Fontana15}
C.~Fontana.
\newblock Weak and strong no-arbitrage conditions for continuous financial
  markets.
\newblock {\em International Journal of Theoretical and Applied Finance},
  18(1):1550005, 2015.

\bibitem{FGGS20}
C.~Fontana, Z.~Grbac, S.~G\"umbel, and T.~Schmidt.
\newblock Term structure modeling for multiple curves with stochastic
  discontinuities.
\newblock {\em Finance and Stochastics}, 24:465--511, 2020.

\bibitem{FGS24}
C.~Fontana, Z.~Grbac, and T.~Schmidt.
\newblock Term structure modelling with overnight rates beyond stochastic
  continuity.
\newblock {\em Mathematical Finance}, 34(1):151--189, 2024.

\bibitem{FLM25}
C.~Fontana, G.~Lanaro, and A.~Murgoci.
\newblock The geometry of multi-curve interest rate models.
\newblock {\em Quantitative Finance}, 25(2):323--342, 2025.

\bibitem{FR13}
C.~Fontana and W.~Runggaldier.
\newblock Diffusion-based models for financial markets without martingale
  measures.
\newblock In F.~Biagini, A.~Richter, and H.~Schlesinger, editors, {\em Risk
  Measures and Attitudes}, EAA Series, pages 45--81. Springer, London, 2013.

\bibitem{GnoattoLavagnini26}
A.~Gnoatto and S.~Lavagnini.
\newblock Cross-currency {H}eath-{J}arrow-{M}orton framework in the
  multiple-curve setting.
\newblock {\em SIAM Journal on Financial Mathematics}, forthcoming, 2026.

\bibitem{GR15}
Z.~Grbac and W.~J. Runggaldier.
\newblock {\em Interest Rate Modeling: Post-Crisis Challenges and Approaches}.
\newblock Springer, 2015.

\bibitem{HJM:92}
D.~Heath, R.~Jarrow, and A.~Morton.
\newblock Bond pricing and the term structure of interest rates: a new
  methodology for contingent claims valuation.
\newblock {\em Econometrica}, 60(1):77--105, 1992.

\bibitem{Jacod1976}
J.~Jacod.
\newblock Un th\'eor\`eme de repr\'esentation pour les martingales
  discontinues.
\newblock {\em Z. Wahrscheinlichkeitstheorie und Verw. Gebiete}, 34:225--244,
  1976.

\bibitem{jashi03}
J.~Jacod and A.~Shiryaev.
\newblock {\em Limit Theorems for Stochastic Processes}.
\newblock Springer, Berlin - Heidelberg - New York, second edition, 2003.

\bibitem{JT:95}
R.~Jarrow and S.~Turnbull.
\newblock Pricing derivatives on financial securities subject to credit risk.
\newblock {\em Journal of Finance}, 50(1):53--85, 1995.

\bibitem{JarrowTurnbull98}
R.~Jarrow and S.~Turnbull.
\newblock A unified approach for pricing contingent claims on multiple term
  structures.
\newblock {\em Review of Quantitative Finance and Accounting}, 10:5--19, 1998.

\bibitem{JarrowYildirim03}
R.~Jarrow and Y.~Yildirim.
\newblock Pricing treasury inflation protected securities and related
  derivatives using an {HJM} model.
\newblock {\em Journal of Financial and Quantitative Analysis}, 38(2):337--358,
  2003.

\bibitem{KK07}
I.~Karatzas and C.~Kardaras.
\newblock The num\'eraire portfolio in semimartingale financial models.
\newblock {\em Finance and Stochastics}, 11(4):447--493, 2007.

\bibitem{KKbook}
I.~Karatzas and C.~Kardaras.
\newblock {\em Portfolio Theory and Arbitrage: A Course in Mathematical
  Finance}.
\newblock American Mathematical Society, Providence (RI), 2021.

\bibitem{Kardaras10}
C.~Kardaras.
\newblock The continuous behavior of the num\'eraire portfolio under small
  changes in information structure, probabilistic views and investment
  constraints.
\newblock {\em Stochastic Processes and their Applications}, 120(3):331--347,
  2010.

\bibitem{Kardaras13b}
C.~Kardaras.
\newblock On the closure in the {E}mery topology of semimartingale
  wealth-process sets.
\newblock {\em Annals of Applied Probability}, 23(4):1355--1376, 2013.

\bibitem{Kardaras24}
C.~Kardaras.
\newblock Stochastic integration with respect to arbitrary collections of
  continuous semimartingales and applications to mathematical finance.
\newblock {\em Annals of Applied Probability}, 34(3):2566--2599, 2024.

\bibitem{phdkoval}
N.~Koval.
\newblock {\em Time-inhomogeneous {L\'evy} Processes in Cross-Currency Market
  Models}.
\newblock PhD thesis, University of Freiburg, 2005.

\bibitem{Memin80}
J.~Memin.
\newblock Espaces de semi martingales et changement de probabilit\'e.
\newblock {\em Z. Wahrscheinlichkeitstheorie Verw. Gebiete}, 52:9--39, 1980.

\bibitem{Metivier}
M.~M\'{e}tivier.
\newblock {\em Semimartingales: a Course on Stochastic Processes}.
\newblock De Gruyter, Berlin - New York, 1982.

\bibitem{Miller-Platen}
S.~Miller and E.~Platen.
\newblock A two-factor model for low interest rate regimes.
\newblock {\em Asia-Pacific Financial Markets}, 11(1):107--133, 2004.

\bibitem{Pazy}
A.~Pazy.
\newblock {\em Semigroups of Linear Operators and Applications to Partial
  Differential Equations}.
\newblock Springer, New York, 1983.

\bibitem{PH}
E.~Platen and D.~Heath.
\newblock {\em A Benchmark Approach to Quantitative Finance}.
\newblock Springer, Berlin - Heidelberg, 2006.

\bibitem{PT15}
E.~Platen and S.~Tappe.
\newblock Real-world forward rate dynamics with affine realizations.
\newblock {\em Stochastic Analysis and Applications}, 33(4):573--608, 2015.

\bibitem{Tappe-Wiener}
S.~Tappe.
\newblock An alternative approach on the existence of affine realizations for
  {HJM} term structure models.
\newblock {\em Proceedings of The Royal Society of London. Series A.
  Mathematical, Physical and Engineering Sciences}, 466(2122):3033--3060, 2010.

\bibitem{Tappe-Levy}
S.~Tappe.
\newblock Existence of affine realizations for {L}{\'{e}}vy term structure
  models.
\newblock {\em Proceedings of The Royal Society of London. Series A.
  Mathematical, Physical and Engineering Sciences}, 468(2147):3685--3704, 2012.

\bibitem{Tappe-YW}
S.~Tappe.
\newblock The {Y}amada-{W}atanabe theorem for mild solutions to stochastic
  partial differential equations.
\newblock {\em Electronic Communications in Probability}, 18(24):1--13, 2013.

\bibitem{TappeCones17}
S.~Tappe.
\newblock Invariance of closed convex cones for stochastic partial differential
  equations.
\newblock {\em Journal of Mathematical Analysis and Applications},
  451(2):1077--1122, 2017.

\bibitem{TappeCones24}
S.~Tappe.
\newblock Invariant cones for jump-diffusions in infinite dimensions.
\newblock {\em Nonlinear Differential Equations and Applications}, 31:107,
  2024.

\bibitem{TappeWeber14}
S.~Tappe and S.~Weber.
\newblock Stochastic mortality models: an infinite-dimensional approach.
\newblock {\em Finance and Stochastics}, 18(1):209--248, 2014.

\bibitem{Veraar12}
M.~Veraar.
\newblock The stochastic {F}ubini theorem revisited.
\newblock {\em Stochastics}, 84(4):543--551, 2012.

\end{thebibliography}

\end{document}